\documentclass[11pt]{article}

\makeatletter
\def\maxwidth{ %
  \ifdim\Gin@nat@width>\linewidth
    \linewidth
  \else
    \Gin@nat@width
  \fi
}
\makeatother

\RequirePackage{amsmath,amssymb}
\RequirePackage{amsthm}
\RequirePackage{natbib}
\RequirePackage[colorlinks,allcolors=blue]{hyperref}

\usepackage{fullpage}
\usepackage{setspace}

\usepackage{bbm,graphicx,bm} 
\usepackage{enumitem}
\usepackage{multirow}
\usepackage[all,import]{xy}
\usepackage{appendix}
\usepackage{comment}
\usepackage{color}
\usepackage{soul}
\usepackage{pdflscape}
\usepackage{subcaption}

\usepackage{booktabs}
\usepackage{longtable}
\usepackage{array}
\usepackage[table]{xcolor}
\usepackage{wrapfig}
\usepackage{float}
\usepackage{colortbl}
\usepackage{pdflscape}
\usepackage{tabu}
\usepackage{threeparttable}
\usepackage{threeparttablex}
\usepackage[normalem]{ulem}
\usepackage{makecell}
\usepackage{afterpage}
\usepackage{thmtools}
\usepackage{thm-restate}

\usepackage{chngcntr}
\usepackage{apptools}
\AtAppendix{\counterwithin{lemma}{section}}
\AtAppendix{\counterwithin{theorem}{section}}
\AtAppendix{\counterwithin{proposition}{section}}

\theoremstyle{definition}
\newtheorem{assumption}{Assumption}

\declaretheorem[name=Proposition,numberwithin=section]{proposition}

\newcommand{\R}{\ensuremath{\mathbb{R}}}

\newcommand{\obs}{\text{obs}}
\newcommand{\mis}{\text{mis}}

\def\b1{\boldsymbol{1}}

\definecolor{RED}{RGB}{255,0,0}

\usepackage{soul}
\usepackage[utf8]{inputenc}

\usepackage{xr}

\title{Estimating the effects of a California gun control program with Multitask Gaussian Processes%
\thanks{We would like to thank Alyssa Bilinski, Phil Cook, Cass Crifasi, Alex D'Amour, Peng Ding, John Donohue, Naoki Egami, Max Gopelrud, Jess Kunke, Luke Miratrix, Jacob Montgomery, Jesse Rothstein, Yotam Shem-Tov, Elizabeth Stuart,  Liyang Sun,  and seminar participants at Polmeth 2021 and the University of Washington for helpful comments and discussion. We are especially grateful to Eli Sherman for his involvement in an earlier version of this project. A portion of this paper was previously circulated with the title ``The Effect of a Targeted Effort to Remove Firearms from Prohibited Persons on State Murder Rates.'' This research was supported in part by the Hellman Family Fund at UC Berkeley and by the Institute of Education Sciences, U.S. Department of Education, through Grant R305D200010. The opinions expressed are those of the authors and do not represent views of the Institute or the U.S. Department of Education.}}

\author{Eli Ben-Michael, David Arbour, Avi Feller, Alexander Franks, and Steven Raphael}
\date{\today}

\begin{document}

\maketitle

\thispagestyle{empty}
\pagenumbering{gobble}

\singlespacing
\begin{abstract}
Gun violence is a critical public safety concern in the United States.
In 2006 California implemented a unique firearm monitoring program, the Armed and Prohibited Persons System (APPS), to address gun violence in the state.
The APPS program first identifies those firearm owners who become prohibited from owning one due to federal or state law, then confiscates their firearms.
Our goal is to assess the effect of APPS on California murder rates using annual, state-level crime data across the US for the years before and after the introduction of the program.
To do so, we adapt a non-parametric Bayesian approach, multitask Gaussian Processes (MTGPs), to the panel data setting. MTGPs allow for flexible and parsimonious panel data models that nest many existing approaches and allow for direct control over both dependence across time and dependence across units, as well as natural uncertainty quantification.
We extend this approach to incorporate non-Normal outcomes, auxiliary covariates, and multiple outcome series, which are all important in our application.
We also show that this approach has attractive Frequentist properties, including a representation as a weighting estimator with separate weights over units and time periods.
Applying this approach, we find that the increased monitoring and enforcement from the APPS program substantially decreased homicides in California. We also find that the effect on murder is driven entirely by declines in gun-related murder with no measurable effect on non-gun murder. Estimated cost per murder avoided are substantially lower than conventional estimates of the value of a statistical life, suggesting a very high benefit-cost ratio for this enforcement effort.
\end{abstract}

\clearpage
\pagenumbering{arabic}
\onehalfspacing



\clearpage
\section{Introduction}
\label{sec:intro}

Every year, thousands of people in the United States are murdered with a firearm. In 2020 alone, there were over 45,000 firearm-related deaths in the U.S., including deaths by suicide and murder; firearm-related injuries are now the leading cause of death among U.S. children and adolescents \citep{goldstick2022current}. This mortality rate remains multiple times that of other wealthy nations, largely due to differences in the lethality of violence in the United States.  
 	
Addressing gun violence is thus one of the most pressing public safety issues in the United States.  Policy proposals intended to curb gun violence vary considerably.  At one end of the spectrum are proposals to increase gun availability to deter potential assailants through such policies as shall-issue concealed carry laws \citep{lott1997crime} or the arming of school teachers \citep{green2018nyt}.  At the other end are policy proposals to restrict access to firearms through universal background checks, temporarily removing firearms from individuals deemed a danger to themselves or others \citep{wintemute2019extreme}, and restricting access to assault weapons with high capacity magazines \citep{donohue2019assault}.    

Federal and state laws already prohibit certain individuals from owning firearms.  Under federal law, individuals with prior felony convictions as well as some individuals with a history of mental illness are prohibited from purchasing and owning firearms.  Many states go further, restricting access to firearms for individuals convicted of violent misdemeanors, individuals that are subjected to restraining orders, and juvenile offenders, among other categories \citep{giffords2019}. While law enforcement officials confiscate firearms of prohibited persons discovered to be in possession, until recently there was no effort to proactively and comprehensively recover firearms from prohibited individuals.

California was the first and remains the only state to undertake such an effort.  In 2006 California implemented the Armed and Prohibited Persons System (APPS), a monitoring program that identifies known firearm owners who become prohibited from owning a firearm and then subsequently attempts to retrieve all prohibited weapons. Through this program, state law enforcement has made hundreds of thousands of contacts with prohibited individuals and has removed tens of thousands of firearms from their possession, often permanently.\footnote{While California is not the only state that screens existing gun owners for prohibiting conditions, it is the only state to screen daily for triggering events, to carefully monitor whether the person has been disarmed, and to have a law enforcement unit in charge of following up, disarming such individuals, and storing the firearms.}

In this paper, we assess whether the implementation of APPS affected murder rates in California, using annual, state-level data for the 50 states. 
To do so, we adapt a non-parametric Bayesian approach, \emph{multitask Gaussian Processes} (MTGPs), to estimate the causal effect of APPS.
Originally developed for the setting with multiple outcomes rather than multiple units \citep{bonilla2008multi, alvarez2017gp}, Multitask GPs separately parameterize dependence across time and dependence across units. This allows for a flexible and parsimonious panel data model that nests many existing approaches: Our primary model is a natural generalization of low-rank factor models commonly used in panel data settings \citep[e.g.,][]{xu2017gsynth, athey2018matrix} with an additional prior that encourages smoothness in the underlying factors. 
We also take advantage of the large literature on the statistical and computational properties of Gaussian Processes \citep{williams2006gaussian, gelman2013bayesian} to extend this model. Immediate extensions include: allowing for a count (Poisson) observational model; allowing for multiple outcome series; incorporating a mean model; and incorporating auxiliary covariates. 

Applying these ideas to the impact of APPS yields large, negative effect estimates on gun-related homicides in California. Specifically, we estimate over 100 murders prevented per year, a decline of at least 5 percent relative to California's murder rate immediately prior to the program. In addition, we find that the effect on murder is driven entirely by declines in gun-related murder with no measurable effect on non-gun murder. 
Back-of-the-envelope calculations find that program expenditures are likely under \$100,000 per murder prevented, much lower than conventional estimates of the value of a statistical life, suggesting a very high benefit-cost ratio for this enforcement effort.

Methodologically, we believe that the MTGP framework is a promising approach for estimating causal effects with panel data more broadly.
First, we can naturally incorporate many panel data diagnostics and other model checks in terms of standard Bayesian workflows, such as posterior predictive checks. 
This is especially important in settings such as ours, with relatively few units and (pre-treatment) time periods. 
The Bayesian approach also  
yields coherent uncertainty quantification, appropriately propagating uncertainty from the different sources of information in the model. In particular, the specific model we consider automatically produces uncertainty intervals that grow over time post-treatment. Surprisingly, this last point is a departure from many existing approaches, which sometimes have constant uncertainty intervals post-treatment.
We also report many flexible summaries of the posterior distribution, allowing us to obtain posterior distributions for estimands like the benefit-cost ratio, a key quantity in policy analysis.

Finally, we adapt results from \citet{kanagawa2018gaussian} showing that the Gaussian Process can be represented as a weighting estimator.
We then use this representation to show that the overall weights decompose into separate unit and time weights. As a result, the proposed MTGP approach has the form of a doubly-weighted average over units and time periods, similar to recent proposals in the panel data literature \citep{arkhangelsky2019synthetic, ben2018augmented}.

As we discuss in Section \ref{sec:related}, our proposed approach combines the growing literature on Gaussian Processes for causal inference \citep[see][]{oganisian2020practical} with the robust literature on estimating causal effects for a single treated unit with panel data \citep[see][]{samartsidis2019assessing}. Of particular relevance are \citet{modi2019generative} and \citet{carlson2020gp}, who also consider GPs in this context, albeit with different structures. More broadly, our proposed framework is a natural (Bayesian) generalization of recent proposals for low-rank factor models with panel data \citep[e.g.,][]{xu2017gsynth, athey2018matrix}; see also recent extensions from \citet{pang2020bayesian}.

Our paper proceeds as follows.
In Section \ref{sec:apps_intro} we give an overview of the APPS program and other institutional details and present an initial analysis of the effect of APPS.
In Section \ref{sec:setup_section} we describe the underlying causal problem and review related work.
In Section \ref{sec:mtgp_main} we review single- and multi-task Gaussian Processes and apply these ideas to our application setting. 
In Section \ref{sec:weights} we represent the MTGP approach as a weighting estimator.
In Section \ref{sec:apps_main_analysis} we report model diagnostics and the estimated impact of APPS under different models.
Finally, in Section \ref{sec:discussion} we discuss open questions and possible extensions.
The Appendix includes additional computational details, diagnostics, and results.
All code and replication files are available at \href{https://github.com/darbour/mtgp_panel}{\texttt{github.com/darbour/mtgp\_panel}}.

\section{Armed and Prohibited Persons System}
\label{sec:apps_intro}

\subsection{Description of the system}

The Armed and Prohibited Persons System is a law enforcement program devoted to proactively removing firearms from known gun owners who become prohibited from owning a gun. The program has two components: an intricate data infrastructure and monitoring system that updates daily; and an enforcement operation, which includes both non-sworn analysts and a small force of sworn state officers, dedicated solely to enforcing state firearms restrictions for gun owners who become prohibited.
While the state has maintained the database on gunowners and associated firearms since the mid 1990s, the active effort to remove firearms from prohibited individuals did not begin until December 2006 \citep{CA_DOJ_2015}.

The core of the data system is the APPS database, a list of known gunowners that draws on Dealer Record of Sales (DROS) transactions, voluntary gun registrations with the state, and an Automated Firearm System (AFS) database that records all known firearms involved in sales through licensed dealers.
The database then links each individual known to own a firearm to specific firearms identified by serial numbers.  
This system has recorded owners of handguns and their associated weapons since 1996.
In 2014, the state expanded the system to include all gunowners and all firearms, including long guns. 
The known population of gunowners in the APPS database has grown substantially over time, 
from under 900,000 in 2006 to over 2.5 million in 2019, or roughly 6 percent of the state's resident population.\footnote{The system likely covers only a portion of gunowners in the state since the underlying list of gunowners used by the system is generated largely by new firearms sales.  For example, the list excludes people who legally purchase firearms in other states and fail to register them in California, as well as people who purchased handguns prior to 1996 and long guns prior to 2014 and have not registered these firearms.}

The APPS data are cross referenced daily against several databases that record events that trigger a firearm prohibition.  These include the state's Automated Criminal History System (the state's criminal history repository), the Mental Health Reporting system (a database used by mental health providers to record involuntary commitments and other mental health events that may trigger a prohibition), the California Restraining and Protective Order system, and a statewide Wanted Persons System.  
When the cross-reference uncovers a triggering event for an individual in APPS, the case is analyzed by state criminal intelligence specialists to determine whether the person is actually prohibited and if the person is still armed.  Local law enforcement across the state conducts gun relinquishment efforts for those who become prohibited;  hence, local agencies may have already removed the person's firearms.\footnote{In 2016, this became the official responsibility of county probation departments across the state for those newly convicted of a prohibiting offense.} In the event that the APPS database reveals that the person is still armed,
the case is referred to the state Bureau of Firearms (BOF) for an enforcement action.
Of these, 54 percent were prohibited due to a felony conviction, 11 percent were prohibited due to a misdemeanor conviction, 19 percent were prohibited due to an existing restraining order, 18 percent were prohibited due to a mental health prohibition, 24 percent were prohibited due to provisions of the Federal Brady Act, while the remaining 7 percent were prohibited due to active probation status \citep{CA_DOJ_2019}.

BOF officers attempt to contact these individuals, usually at their last known address, and request to perform a consensual search of the premises.  If refused, officers will return with a search warrant.  In many instances, officers only recover some of the firearms associated with the individual in question, possibly due to a prior unrecorded sale, the individual turning over firearms to friends and family, the gun being lost or stolen, or local law enforcement already collecting the firearms at conviction and failing to update the person's record in APPS.  In addition, officers frequently discover unregistered weapons along with the registered firearms.\footnote{For an abundance of examples of the product of these enforcement actions, see \cite{CA_DOJ_2015, CA_DOJ_2016, CA_DOJ_2017, CA_DOJ_2018, CA_DOJ_2019}.}
Once all guns associated with prohibited individuals are removed from possession, the individuals are removed from the APPS database as they are no longer armed.  Gunowners are also removed from the APPS database upon death. Between 2007 and 2017, the BOF directly retrieved nearly 32,000 firearms from prohibited individuals.  However, the total number of firearms relinquished statewide is likely much larger given the role of local law enforcement and the fact the BOF works cases that are somehow not handled locally.

\subsection{Related research on access to firearms}
\label{sec:related_firearm}

A growing literature evaluates the impact of state-level changes in access to firearms on homicide rates; see \citet{webster2015effects} and \citet{cook2017saving}. 
Overall, recent studies find that more permissive firearm regulations, such as more permissive concealed carry laws and Right to Carry (RTC) laws, are associated with increases in gun-related deaths \citep{donohue2019right, siegel2017easiness}.\footnote{While \citet{lott1997crime} find the opposite, their results have been convincingly disputed on methodological grounds \citep{aneja2011impact, nrc2005firearms}.} 
Similarly, policy changes that restrict firearm access, such as permit-to-purchase laws and comprehensive background checks, are associated with decreases in gun-related deaths \citep{rudolph2015association, castillo2019california}. 
Methodologically, many of these studies also use the synthetic control method, which we consider next \citep{rudolph2015association, kagawa2018repeal, castillo2019california, donohue2019right, mccourt2020purchaser, pear2022firearm}.

Several qualitative studies complement these quantitative estimates. 
\citet{wintemute2019extreme} present a series of California case studies of individuals subjected to gun law restraining orders, where the APPS was central in identifying associated firearms.  In addition, case studies listed in recent years of the California Department of Justice annual APPS report indicate that BOF enforcement teams are often directly involved in carrying out these orders \citep{CA_DOJ_2018, CA_DOJ_2019}.  In the case studies documented by Wintemute and colleagues, the individuals disarmed as a result of these orders all made specific and highly alarming threats against intimate partners, coworkers, educational institutions, and in some instances health and mental health providers.  Among the cases where the individual in question threatened violence, none had committed a violent act at the time of publication.  See also \citet{swanson2017implementation}, who find that such restraining orders may also be effective in reducing gun suicide.  

Finally, there is currently a randomized control trial of the APPS enforcement program in the field \citep[see][for the study protocol]{wintemute2017evaluation}. The study began in 2015 to randomly assign regions within California for earlier or later enforcement actions among regions without well-established working relationships with the state Department of Justice.  To date, results from this evaluation have not been reported.

\subsection{Synthetic control estimates}
\label{sec:synth}

\begin{figure}[tbp]
  \centering
    \begin{subfigure}[t]{0.45\textwidth}  
  \includegraphics[width=\textwidth]{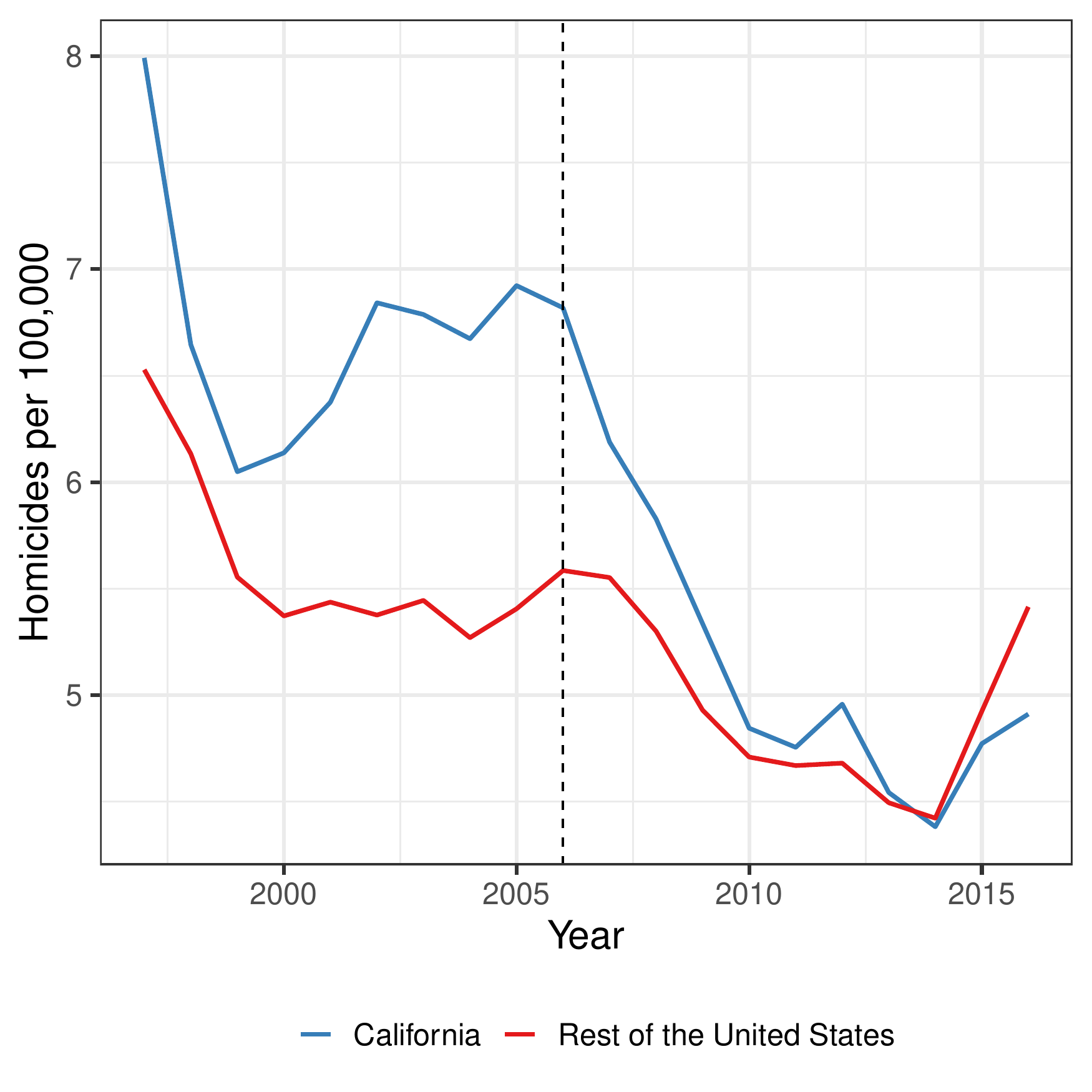}
    \caption{}
    \label{fig:murder_raw}
    \end{subfigure}%
    ~
    \begin{subfigure}[t]{0.45\textwidth}  
    \includegraphics[width=\textwidth]{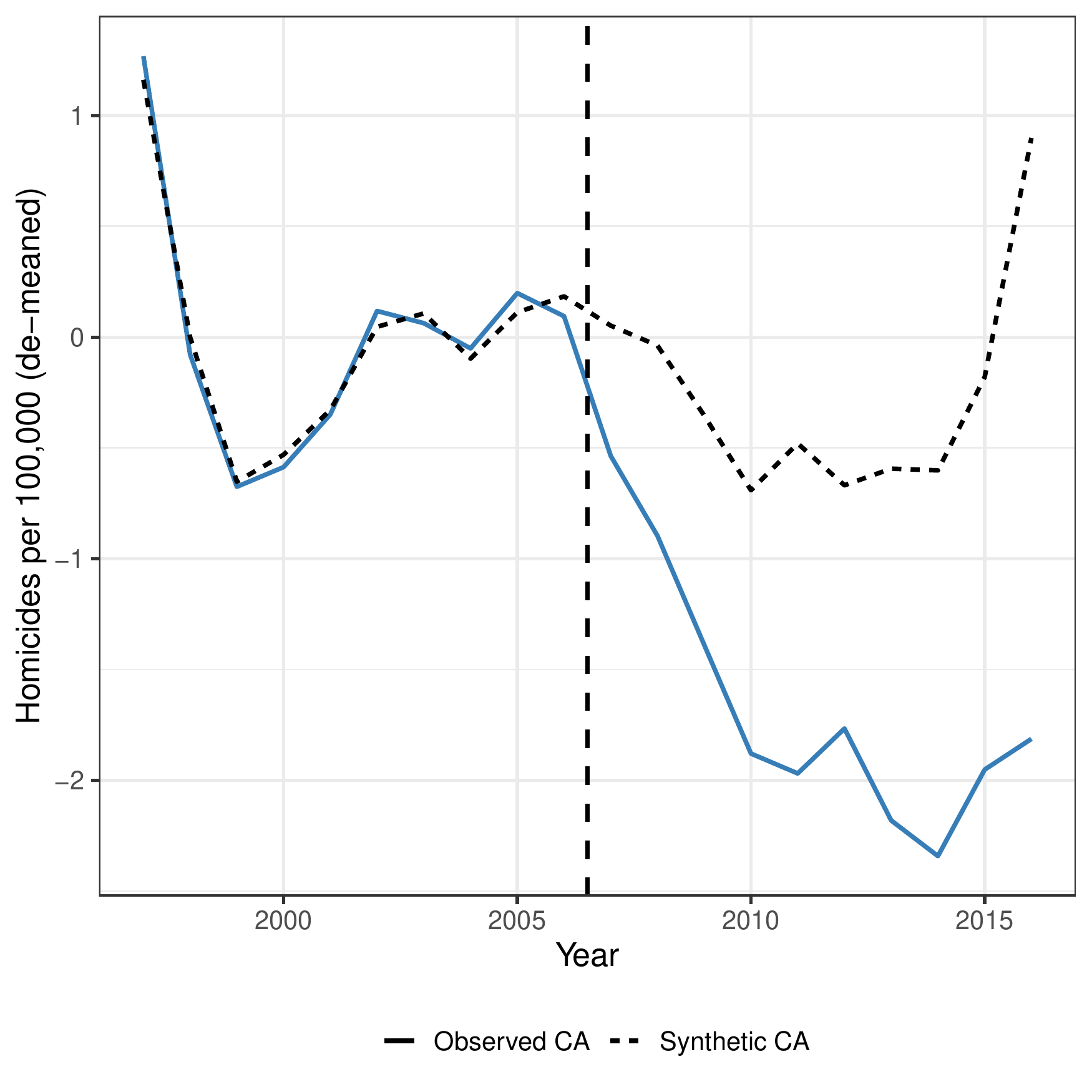}
    \caption{}
    \label{fig:synth_plot}
      \end{subfigure}
\caption{Annual homicide rate per 100,000, 1997 through 2016. The dotted line is 2006, the year the APPS program was launched. (a) California and the rest of the United States. (b) California and Synthetic California, matched based on the pre-intervention years 1997-2006, de-meaned} 
\label{fig:synth_results}
\end{figure}

Figure \ref{fig:murder_raw} shows the annual homicide rate per 100,000 people in California versus the rest of the United States (a population-weighted average of the remaining 49 states); the dotted line indicates 2006, the year the APPS program was first launched (in December). 
The figure is based on data from the Federal Bureau of Investigation (FBI) Uniform Crime Reports (UCR).
The figure shows several striking patterns.
First, the overall murder rate in the United States largely declined over this period, with a major increase in the final years of the panel. 
Second, the murder rate in California was substantially higher than the national average in the late 1990s and mid 2000s, 
with an average difference in murders per 100,000 of 1.1 over the ten-year period through 2006.
This gap closes after 2006: the California murder rate falls and remains below the rate for the rest of the nation by 2014. 

The methodological question is whether and how we can use this panel data structure to estimate the impact of APPS on the murder rate. Furthermore, we would like to separately estimate effects on \emph{gun related} and \emph{non-gun related} homicides. The APPS program should primarily impact gun-related homicides by removing firearms from a high-risk population, though substitution effects may lead to changes in non-gun-related homicides. Therefore our methodological goal also includes estimating effects on \emph{multiple outcomes} simultaneously.

Unfortunately, relatively simple approaches are unlikely to perform well here. For example, the workhorse ``difference-in-differences'' (DiD) regression model assumes \emph{parallel trends}: in the absence of the intervention, the differences between the murder rates in California and the rest of the country are constant over time. This is clearly violated in Figure \ref{fig:murder_raw}, with much wider differences in the mid 2000s than in the late 1990s.

An alternative to DiD for settings such as this where the parallel trends assumption is implausible is 
the \emph{synthetic control method} \citep[SCM][]{abadie2010synthetic}. This procedure finds a weighted combination of other states, \emph{Synthetic California}, that approximates the pre-APPS murder rate for California.
Here, we use a variant of SCM that first subtracts out the pre-APPS homicide rate in each state to isolate trends, then finds a synthetic control that approximates the pre-APPS murder rate trends for California \citep{doudchenko2016balancing, ferman2019synthetic, ben2018augmented}.
See Appendix \ref{sec:appendix_synth} for further details on constructing the synthetic control.

Figure \ref{fig:synth_plot} shows the de-meaned California murder rate (the solid line) alongside the de-meaned murder rate for our synthetic comparison group of states (the dashed line).
The synthetic control has excellent pre-treatment fit: the two series are very close for the pre-treatment period of 1997 through 2006.
Seven states have non-zero weight and contribute to the synthetic control: Ohio is by far the most important donor state (weight = 0.65), followed by Hawaii (0.1), Mississippi (0.09), Louisiana (0.06), Arkansas (0.06), Alaska (0.03) and Idaho (0.01).
After 2006, California's murder rate is markedly lower than the rate for synthetic California, with an average difference of -1.41 murders per 100,000. This difference widens over the ten-year period following the implementation of APPS.

While a promising initial analysis, several open questions remain. First, we would like to be able to estimate the effects on both gun related and non-gun related homicides in order to evaluate our hypothesis that the APPS program should have a larger impact on gun related homicides.
Second, the methodology for panel data with multiple outcomes is under-developed, and it is unclear how this SCM approach can be applied more generally.
Relatedly, the number of homicides per state are inherently count outcomes, a fact ignored in existing SCM analyses but one that we would like to incorporate explicitly. 
Finally, quantifying uncertainty with synthetic controls is challenging. The two most common approaches of placebo and conformal inference \citep{abadie2010synthetic, chernozhukov2017exact} rest on versions of exchangeability either over units or over time. These assumptions are difficult to justify in our setting, and so it is unclear whether the large differences in Figure \ref{fig:synth_plot} is driven by true effects or statistical noise.
Moreover, the uncertainty intervals from conformal inference approaches are constant over time even though, intuitively, we anticipate far greater uncertainty for the counterfactual murder rate in 2017 than immediately after the implementation of APPS in 2007.
This adds more difficulty in understanding whether the effect is truly increasing over time as we see in Figure \ref{fig:synth_plot}.
These concerns motivate our alternative approach, which we turn to next.

\section{Setup and Background}
\label{sec:setup_section}

\subsection{Problem setup}
\label{sec:setup_subsection}

We now describe the standard panel data setting with $i = 1, \ldots, N$ units observed for $t = 1, \ldots, T$ time periods. In our application, we observe $N = 50$ states for $T = 20$ years. Let $W_i$ be an indicator that unit $i$ is treated at time $T_0 < T$ where units with $W_i = 0$ never receive the treatment. While our setup is more general, we restrict our attention to the case where a single unit receives treatment, and follow the convention that this is the first one, $W_1 = 1$. The remaining $N_0 = N - 1$ units are untreated. In our application, the treatment of interest is the APPS program, treated at time $T_0 = 10$; California is the treated unit and the remaining states are possible controls.
Note that while only California implement a program like APPS, as we discuss above, other states have enacted various gun-related policies.
Thus, the un-treated condition $W_1 = 0$ corresponds to the policy environment California would have experienced in the absence of APPS.

We describe our problem of interest in the potential outcomes framework. We assume no interference between units and stable treatment, which allows us to write the potential outcomes for unit $i$ in time $t$ under control and treatment as $Y_{it}(0)$ and $Y_{it}(1)$ respectively.\footnote{While this is not always the case, in our application, it is reasonable to assume that any state could launch an APPS-style program at any point. Thus, both potential outcomes could reasonably exist for all units and all time periods.}
Our primary outcome of interest is the overall murder rate; we consider multiple outcomes in Section \ref{sec:multiple_outcomes}.
For ease of exposition, we initially ignore the presence of background covariates from our analyses; we discuss extensions to incorporate them in Section \ref{sec:outcome_model}.

Since the first unit is treated at time $T_0$, the observed outcomes $Y_{it}$ are therefore:
$$
Y_{i t}=\left\{\begin{array}{ll}{Y_{i t}(0)} & {\text { if } W_{i}=0 \text { or } t \leq T_{0}} \\ {Y_{i t}(1)} & {\text { if } W_{i}=1 \text { and } t>T_{0}}\end{array}\right. .
$$
The goal is to estimate the causal effect for the treated unit at each post-treatment time $t > T_0$, $\tau_{t} = Y_{1t}(1) - Y_{1t}(0) = Y_{1t} - Y_{1t}(0)$, as well as the average treatment effect on the treated unit after time $T_0$, 
$$\tau = \frac{1}{T-T_0}\sum_{t=T_0+1}^T \tau_t.$$
Since we observe the treated outcome $Y_{1t} = Y_{1t}(1)$ for the treated unit at $t > T_0$, the challenge is to impute the missing potential outcome for the treated unit at time $t > T_0$. 
Note that we adopt a finite sample causal inference framework here: our goal is to estimate the impact for California, rather than for a hypothetical population \citep[see][]{imbens2015causal}.

\subsection{Bayesian causal inference}
\label{sec:causal}

We adopt a Bayesian approach to estimating causal quantities. There are many subtle issues that arise in Bayesian causal inference that are not central to our discussion. For instance, identification in the Bayesian approach is conceptually distinct from identification in non-Bayesian causal inference: if the prior distribution is proper then the posterior distribution will also be proper, regardless of whether the parameters in the likelihood are fully or partially identified \citep[see][]{gustafson2010bayesian, imbens2015causal}. 
See \cite{ding2018causal} and \citet{oganisian2020practical} for additional background. \citet{pang2020bayesian} and \citet{menchetti2020estimating} give overviews of Bayesian causal inference for panel or time series data.

At a high level, Bayesian causal inference views missing potential outcomes as unobserved random variables to be imputed \citep{rubin1978bayesian, ding2018causal}. 
We therefore break the problem into two parts: first, estimate a Bayesian model using the observed data; second, use this model to obtain the posterior (predictive) distribution for the missing potential outcomes.
Specifically, let $\mathbf{Y^{\mis}} = \{Y_{it}(0)\}_{i = 1, t > T_0} \in \mathbb{R}^{T - T_0}$ be the missing potential outcomes for the treated unit after $T_0$, and let $\mathbf{O}^{\obs} = (\{Y_{it}(0)\}_{i > 1},$  $\{Y_{it}(0)\}_{i = 1, t \leq T_0},$ $\{Y_{it}(1)\}_{i = 1, t > T_0},$ $\mathbf{W})$ be the set of observed data, including all control outcomes, the treated unit's pre-treatment outcomes, the treated unit's observed outcomes post-treatment, and the set of treatment assignments. We introduce auxiliary covariates in Section \ref{sec:outcome_model}.

The goal is to impute the missing potential outcomes, $\mathbf{Y^{\mis}}$, by drawing from the posterior predictive distribution, 
$\mathbb{P}(\mathbf{Y^{\mis}} | \theta, \mathbf{O^{\obs}})$, governed by a model parameter $\theta$. 
We can then estimate $\hat{\tau}_t = Y_{1t} - Y^\ast_{1t}(0)$ at post-treatment time $t > T_0$, where $Y_{1t}^\ast(0)$ is a draw from the posterior predictive distribution for $Y_{1t}(0)$ for unit 1 at time $t$.
To do so, we follow 
\citet{richardson2011transparent} and use the fact that: 
\begin{align*}
    \mathbb{P}\{\mathbf{Y^{\mis}}, \theta \mid \mathbf{O^{\obs}}\} =  \mathbb{P}\{\theta \mid \mathbf{O^{\obs}} \} \cdot \mathbb{P}\{\mathbf{Y^{\mis}} \mid \theta, \mathbf{O^{\obs}}\},
\end{align*}
where $\mathbb{P}\{\mathbf{Y^{\mis}}, \theta | \mathbf{O^{\obs}}\}$ is the posterior distribution of our model parameters and the missing potential outcomes given observed data. 
Now, to obtain the target posterior predictive distribution, we first specify a prior distribution $\mathbb{P}(\theta)$ and use Bayes rule to obtain the posterior distribution $\mathbb{P}(\theta | \mathbf{O^{\obs}})$. With this posterior over the model parameters, we can repeatedly draw samples of $\mathbf{Y^{\mis}}$, yielding the joint posterior on the left-hand side of the above expression. Obtaining the target posterior predictive distribution is then a matter of rearranging the above expression to isolate $\mathbb{P}(\mathbf{Y^{\mis}} | \mathbf{O^{\obs}})$, marginalizing over $\theta$.
This approach rests entirely on the form of the model and the estimated model parameters. See \citet{pang2020bayesian} for an alternative framing that explicitly states sufficient conditions in the form of identifying assumptions, and \citet{miratrix2020ITS} for a discussion of model-based uncertainty in a similar setting.

Formally, we assume that units and time periods are exchangeable conditional on the model parameters $\theta$, which have a prior distribution $p(\theta)$.
Following \cite{ding2018causal}, we partition the parameter space into $\theta^m$, which governs the marginal potential outcome distribution, and $\theta^a$, which governs the association between $Y_{it}(0)$ and $Y_{it}(1)$. In principle, we could treat $\theta^a$ as a sensitivity parameter; for the purposes of this paper, we assume that the potential outcomes are conditionally independent given data and model parameters, allowing us to ignore $\theta^a$.\footnote{It is tempting to sidestep this issue entirely and view this framework as solely focusing on the control potential outcomes. The challenge is that a (formal) Bayesian approach conditions on all data and parameters. Thus, this additional restriction makes explicit the fact that we are not incorporating any information from the observed $Y_{it}(1)$. See \citet{ding2018causal} for further discussion.} The posterior for $\theta^m$ is then:
\begin{align*}
\mathbb{P}\{\theta^m \mid \mathbf{O^{\obs}} \} \propto
p(\theta^m)  \prod_{(i,t)\in \mathcal{C}} \mathbb{P}(Y_{it}(0) \mid \theta^m, W_{it}),
\end{align*}
where $\mathcal{C}$ denotes the ``control'' unit and time periods: $t = 1, \ldots, T_0$ for $i = 1$ and $t = 1, \ldots, T$ for other units. Let $Y^\ast_{1t}(0)$ denote the posterior predictive distribution over the missing potential outcomes, $Y^\ast_{1t}(0) = \mathbb{P}\{\mathbf{Y^{\mis}} \mid \theta,  \mathbf{O^{\obs}} \} $. 
The estimated treatment effect at time $t$ is then the observed treated potential outcome minus the predicted missing potential outcome at post-treatment time $t > T_0$, $\hat{\tau}_t = Y_{1t} - Y^\ast_{1t}(0)$, where repeatedly drawing $Y^\ast_{1t}(0)$ from the posterior predictive distribution propagates the uncertainty.

A key consequence of this setup is that we never use post-treatment data from the treated unit in our modeling (or, equivalently, we have an arbitrarily flexible model for the observed treated data).
Rather, our focus is on modeling the observed control potential outcomes.
And by imputing $Y_{it}^\ast(0)$ from the posterior \emph{predictive} distribution, our inference is inherently finite sample, allowing us to directly reason about the effect of APPS on California's murder rate \citep{imbens2015causal}.

Finally, while the setup here is quite general, in practice we restrict ourselves to a structure where the control potential outcomes are a \emph{model component} plus additive noise. 

\begin{assumption}[Additive, separable error structure]
\label{a:additive_noise}
The control potential outcomes, $Y_{it}(0)$, are generated as:
\begin{align*}
    Y_{it}(0) = \text{model}_{it}(\theta) + \varepsilon_{it} \qquad\qquad \mathbb{E}[\varepsilon_{it}] = 0
\end{align*}
where $\text{model}_{it}(\theta)$ is the outcome model for unit $i$ at time $t$, governed by parameters $\theta$; and $\varepsilon_{it}$ is mean-zero noise.
\end{assumption}

This assumption, sometimes referred to \emph{strict exogeneity} in the panel data setting \citep[see, e.g.,][]{imai2019use},
places no restrictions on the model terms and $\theta$. This allows for \emph{non-parametric} modelling of the outcomes where $\theta$ is potentially infinitely dimensional --- but does not allow treatment assignment to depend on the errors. In other words, given the model and the possibly infinitely many parameters $\theta$, the error distribution for the treated unit's (control) potential outcomes are exchangeable with other units' errors.
In our analysis of the APPS program, one implication of Assumption \ref{a:additive_noise} is that California's decision to implement APPS must not be in response to an \emph{idiosyncratic shock} to the homicide rate. This is consistent with the history of the program \citep[see, e.g.,][]{CA_DOJ_2015}, which largely came about due to changes in IT infrastructure and an entrepreneurial California Attorney General.
As we discuss in Section \ref{sec:discussion}, we can relax Assumption \ref{a:additive_noise} within the fully Bayesian workflow, for instance, via sensitivity analysis \citep{franks2019flexible}.

\subsection{Related methodological research}
\label{sec:related}

Our paper bridges several robust literatures. First, we build on the many existing methods for estimating causal effects for a single treated unit with panel data; see \citet{samartsidis2019assessing} for a recent review. 
Within this, there are several important threads. Most directly relevant are recent papers on Bayesian implementations of (or alternatives to) the synthetic control method, including \citet{tuomaala2019bayesian, kim2020bayesian, pang2020bayesian,pinkney2021improved}. 
Of particular interest is \citet{causalImpact}, who propose a Bayesian structural time series model for estimating causal effects, focused on a single treated series; see \citet{menchetti2020estimating} for a recent extension. We explore connecections between \citet{causalImpact} and our proposed method in Appendix \ref{sec:appendix_causal_impact}.
There are also several recent papers that directly estimate factor models for causal effects \citep[e.g.,][]{xu2017gsynth, athey2018matrix}, as well as approaches that directly address multiple outcomes \citep{samartsidis2020bayesian}. 
As we discuss below, we incorporate many of the novel ideas in these papers, including many of the prior choices. However, our proposed MTGP framework is typically more general, and, we believe, more conducive to estimating causal effects with panel data.

Next, there is a small but growing set of papers on the use of Gaussian Processes for estimating causal effects. To date, nearly all of these papers have focused on the cross-sectional setting \citep{alaa2017bayesian, schulam2017reliable,ray2018semiparametric, huang2019gpmatch, branson2019nonparametric,witty2020gp, ren2021bayesian}. See \citet{oganisian2020practical} for a recent review. 

We are aware of very few papers specifically using GPs to estimate causal effects in panel data settings, although, as we discuss in Section \ref{sec:weights}, many existing estimators can be written as special cases of this approach.\footnote{See also \citet{karch2018gaussian}, who use GPs for latent growth curve modeling in psychometrics, and \citet{huang2015identification}, who instead focus on causal discovery.}
Two working papers are especially relevant for our work.
The first is recent work from \citet{modi2019generative}, who also propose a GP approach for estimating causal effects in this setting.
Their main proposal, however, is very different than ours, with a focus on GP for the frequency domain. Thus, we view our model as a useful complement to theirs. 
A second recent proposal is \citet{carlson2020gp}, who also uses a GP approach but does not exploit the multitask structure.
In independent work, \citet{montgomery2022mtgp} also propose using a multitask GP for panel data. Unlike in our application, however, they consider the setting with many units divided into a treatment and control group and are focused on modeling so-called ``parallel-ish'' trends between these groups. In a different direction, \citet{Antonelli2021} propose estimating effects under staggered adoption of treatment using a Bayesian vector autoregressive model that allows for spatial and time dependence, which is in some ways analogous to the multitask structure we propose.  Finally, while not explicitly causal, of particular relevance to our approach is \citet{flaxman2015fast}, who propose a hierarchical GP model that is similar to what we discuss here.

\section{Multitask Gaussian Processes for Causal Inference}
\label{sec:mtgp_main}

We now give a high-level description of Multitask Gaussian Processes for causal effects with panel data, beginning with the simpler setting of single-task GPs. We discuss specific model implementation choices in Appendix \ref{sec:implementation}.
To fix terminology, we use \emph{tasks} throughout to refer to units; thus, a single-task GP is for a single unit while multi-task GPs are for multiple units. This is in contrast to the original use \citep[e.g.,][]{bonilla2008multi} where tasks instead refers to outcomes.

\subsection{Review: Single Task Gaussian Processes}
\label{sec:single_gp_intro}

We begin by reviewing the setup and properties of a Single-Task GP, which focuses solely on the treated unit; see \citet{williams2006gaussian} for a textbook discussion. While the original GP setup is quite general, to fix ideas we initially focus on applying GPs to the \emph{interrupted time series} or \emph{horizontal regression} setting \citep{athey2018matrix}, in which we use the pre-treatment outcomes for California to forecast post-treatment outcomes in the absence of the intervention. 
Specifically, we initially model the (control) murder rate for California at time $t$, $Y_{1t}(0)$, as the sum of a model component $f_{1t}$ and a mean-zero, independent noise component (as in Assumption \ref{a:additive_noise}):\footnote{For simplicity, we consider a de-meaned outcome series for now, allowing for a global intercept in Section \ref{sec:extensions}. Nonetheless, as we show in Equation \eqref{eq:gp_predmean}, the posterior mean for a future data point can be nonzero even when $\mathbf{f}_1$ has mean zero; see \citet{rasmussen2004gaussian}.} 
$$
\begin{aligned} \label{eq:single_task_GP_simple}
Y_{1t}(0) &= f_{1t} + \varepsilon_{1t} \\
\mathbf{f}_1 &\sim \mathcal{GP}\left(0, k_{\text{time}}\right)\\ 
\varepsilon_{1t} &\overset{\text{iid}}{\sim} N(0, \sigma^2)
\end{aligned}
$$
where $\mathcal{GP}(\cdot)$ is a Gaussian Process prior with corresponding \emph{time kernel}
$k_{\text{time}}$, $\mathbf{f}_1 = (f_{11},\ldots,f_{1T})$ is the vector of model components for the treated unit, and where we initially consider homoskedastic noise $\sigma^2$, which we relax below. 

The key idea is that the GP prior over $\mathbf{f}_1$ incorporates smoothness over time via $k_{\text{time}}(t, t')$, where similar values of $t$ imply larger covariances. Here we use the squared exponential kernel:
    $$k_{\text{time}}(t,t') = \exp\left(-\frac{|t - t'|^2}{2\rho}\right),$$
where $\rho$ is the length scale.

The choice of kernel function plays an important role in defining the behavior of the model. 
For the family of kernels we consider, the relation between the time difference $|t - t'|$ and the kernel value $k_{\text{time}}(t,t')$ is central in determining how we expect the model to behave over time. If $k_{\text{time}}(t,t')$ is large even for far-apart time periods, then the model components $f_{1t}, f_{1t'}$ are assumed to be highly correlated; this corresponds to an assumption that the underlying model is very smooth over time. Conversely, if $k_{\text{time}}(t,t')$ is low even for close time periods, then the model components are assumed to be close to independent, which corresponds to assuming that the model varies strongly over time. 
The length scale hyper-parameter $\rho$ explicitly controls the level of smoothness in the model components.
Many other kernels are possible and appropriate for other settings, such as periodic kernels; again see \citet{williams2006gaussian} for a textbook discussion.

This GP model implies that the control potential outcomes follow a multivariate Normal distribution, so the conditional distributions have a convenient form. Specifically, as in Section \ref{sec:causal}, we can partition the vector of control potential outcomes for the treated unit into the observed and missing components: $\mathbf{Y}^{\obs}_1 = (Y_{11}^{\obs}, Y_{12}^{\obs}, \ldots, Y_{1T_0}^{\obs})$, and $\mathbf{Y}^{\mis} = (Y_{1T_0+1}(0), Y_{1T_0+2}(0), \ldots, Y_{1T}(0))$. 
Next, we define $\mathbf{K}_\text{{time}} \in \R^{T \times T}$ as the \emph{time kernel matrix}, where $\mathbf{K}_{\text{time}, tt'} = k_\text{time}(t,t')$.
Suppressing the superscript \emph{time} to reduce clutter, we can then partition this multivariate Normal model as:
\begin{align}
    \begin{pmatrix}
    \mathbf{Y}^{\mis}_1 \\
    \mathbf{Y}^{\obs}_1
    \end{pmatrix} \;\; \sim \;\; \text{MVN}\left[ 
    \begin{pmatrix}
    0 \\ 0
    \end{pmatrix}, \;
    \begin{pmatrix}
    \mathbf{K}_{\text{mis}} & \mathbf{K}_{\text{mis,obs}} \\
    \mathbf{K}_{\text{obs,mis}} & \mathbf{K}_{\text{obs}}
    \end{pmatrix} + \sigma^2\mathbb{I}
    \right],
\end{align}
where $\mathbf{K}_{\obs} \in \R^{T_0 \times T_0}$ is the time kernel matrix for observed points, $\mathbf{K}_{\mis} \in \R^{(T-T_0) \times (T - T_0)}$ is the time kernel matrix for unobserved points, and $\mathbf{K}_{\obs, \mis} \in \R^{T_0 \times (T - T_0)}$ is a rectangular matrix of time kernel evaluations for pre- and post-treatment times, whose $t,t'$ element is the kernel value $k_\text{time}(t, T_0 + t')$.

We can then write the the posterior predictive distribution for California's outcomes at future time points, $t > T_0$, as:
\begin{align}
\nonumber
\mathbf{Y^{\mis}_1} | \mathbf{Y}_1^{\obs} &\sim \text{MVN}(\boldsymbol{\mu}, \mathbf{\Sigma})\\
\label{eq:gp_predmean}
    \boldsymbol{\mu} &=  \mathbf{K}_{\text{mis,obs}}(\mathbf{K}_{\obs} + \sigma^2\mathbb{I})^{-1}
    \mathbf{Y}_1^\obs  \\
\label{eq:gp_predcov}
    \mathbf{\Sigma} &= (\mathbf{K}_{\mis}+\sigma^2\mathbb{I}) - \mathbf{K}_{\text{mis,obs}}(\mathbf{K}_{\obs} + \sigma^2\mathbb{I})^{-1}\mathbf{K}_{\text{obs,mis}}.
\end{align}
As we discuss in Section \ref{sec:weights}, the posterior mean is a linear combination of the pre-treatment observations, $\mathbf{Y}_1^{\obs}$.

\subsection{Multitask Gaussian Processes}
\label{sec:mtgp_intro}

We now turn to modeling the murder rate for all states jointly, rather than modeling California's outcome series alone. The primary change is that each observation is now a unit-time pair $(i,t)$, so the kernel measuring similarity between points is doubly-indexed: $k(~(i,t), (i', t')~)$. 
The \emph{multitask} GP framework, originally developed for the setting with a single unit but many outcomes \citep{goovaerts1997geostatistics, bonilla2008multi}, is a natural approach for capturing similarity across multiple dimensions.
As in the single-task GP setting and Assumption \ref{a:additive_noise}, we assume that the control potential outcomes consist of a model component plus additive (Normal) noise:
\begin{align*}
    Y_{it}(0) &= f_{it} + \varepsilon_{it} \\
    \mathbf{f} &\sim \mathcal{GP}(0, k)\\ 
    \varepsilon_{it} &\overset{\text{iid}}{\sim} N(0, \sigma^2),  
\end{align*}
where $\mathbf{f} = (f_{11},\ldots,f_{NT})$ is the vector of model components for all units and time periods.
We again begin with the simplified setting in which the outcome series is de-meaned by a global intercept; we incorporate a mean model in Section \ref{sec:extensions}.

While this framework is quite general, we focus on two common simplifications: allowing for separable unit and time covariances, and imposing a low-rank structure on the unit covariance.

\paragraph{Separable unit and time kernels.}
The first key idea is to decompose the GP kernel into separate kernels that capture similarity across time (the ``time covariance'') and similarity across units (the ``unit covariance''):
\begin{equation} \label{eq:kernel_decomp}
k(~(i,t), (i', t')~) = k_{\text{unit}}(i, i') \times k_{\text{time}}(t, t'),
\end{equation} 
where $k_{\text{unit}}$ defines similarity between the $N$ units and  $k_{\text{time}}$ defines similarity between the $T$ time periods.
This is a strong restriction that implies that the covariance across units, $ k_{\text{unit}}$,  is constant in time. (Conversely, this implies that the time covariance $k_{\text{time}}$ is constant across units.) We discuss some approaches for relaxing this restriction in Section \ref{sec:discussion}.

\paragraph{Low-rank unit covariance.}
Even with the separable kernel restriction, the resulting model is overparameterized. In particular, the similarity between the $T$ time points, $k_{\text{time}}(t, t')$, has a natural parameterization in terms of the difference in time, $|t - t'|$, which is a scalar. By contrast, the similarity between the $N$ units, $k_{\text{unit}}(i, i')$, involves $N(N-1)/2$ comparisons.
Following \citet{goovaerts1997geostatistics} and \citet{bonilla2008multi}, we  simplify the problem by using the \emph{intrinsic coregionalization model} (ICM), which imposes a low-rank assumption on the unit covariance $k_{\text{unit}}$ \citep[see also][Section 4.2]{flaxman2015fast}. 

Specifically, we assume that the unit kernel matrix $\mathbf{K}_{\text{unit}} \in \R^{N \times N}$, where $\mathbf{K}_{\text{unit}, {i,i'}} = k_\text{unit}(i,i')$, has the form $\mathbf{K}_{\text{unit}} = \boldsymbol{\beta}\boldsymbol{\beta}'$ for some matrix $\boldsymbol{\beta} \in \R^{N \times J}$, where $J$ is the rank of the unit kernel matrix, and where we place a hyperprior on $\boldsymbol{\beta}$ and infer it from the data.\footnote{Equivalently, we can write this restriction as $k_{\text{unit}}(i, i') = \langle \boldsymbol{\beta}_i, \boldsymbol{\beta}_{i'}\rangle$, i.e., a linear kernel between weights for units $i$ and $j$, inferred from the data.}
While initially proposed in terms of restrictions on the matrices, we can equivalently write the ICM model as an assumption that the outcome model for each of the $i=1,\dots,N$ units is a unit-specific linear combination of $J$ latent draws from a Gaussian process, i.e.,
\begin{equation}
    \label{eq:icm_generative}
    f_{it} = \sum_{j=1}^J \beta_{ij}u_{jt},
\end{equation}
where each latent $\mathbf{u}_j$ is itself a GP with kernel $k_{\text{time}}$,
and where each unit has its own set of unit-specific factor loadings $\boldsymbol{\beta}_{i} = (\beta_{i1}, \ldots, \beta_{iJ})$.
Thus, the underlying model components $\mathbf{u}_1, \ldots,\mathbf{u}_J$ are shared across units and induce a dependence across them.
The number of shared latent functions $J$ corresponds to the rank of the unit kernel matrix $\mathbf{K}_{\text{unit}}$. This is a key hyper-parameter governing the overall complexity of the model; a higher rank $J$ will lead to a more flexible model that potentially has the risk of overfitting. We discuss choosing the rank $J$ via Bayesian model criticism in Section \ref{sec:PPCs}, though we could also impose a horseshoe prior over $J$ \citep{carvalho2009horseshoe}.
The resulting model is then:
\begin{align*}
Y_{it}(0) &= f_{it} + \varepsilon_{it} \\
f_{it} &= \sum_{j=1}^J \beta_{ij} u_{jt} \\
\mathbf{u}_j &\sim \mathcal{GP}(0,  k_{\text{time}})\\
\beta_{ij} &\overset{\text{iid}}{\sim} \mathcal{N}(0,1) \\
\varepsilon_{it} & \overset{\text{iid}}{\sim} \mathcal{N}(0,\sigma^2).
\end{align*}
We estimate this model using the probabilistic programming language, \texttt{Stan} \citep{stan}. See Appendix \ref{sec:implementation} for hyperprior distributions and additional implementation details; many other hyperprior distributions are possible.

As above, we can re-write this model in matrix notation and can represent the decomposition in Equation \eqref{eq:kernel_decomp} as the \emph{Kronecker product} between the two covariances: 
$$\mathbf{K} \equiv \mathbf{K}_{\text{unit}} \otimes \mathbf{K}_{\text{time}} \in \R^{NT \times NT},$$
where any one entry of $\mathbf{K}$ can be written as the product of the unit and time covariance, as in Equation \eqref{eq:kernel_decomp}.

We can similarly divide this (admittedly unwieldy) kernel matrix into a  matrix for observed unit-time pairs, $\mathbf{K}_{\obs} \in \R^{((n-1)T + T_0) \times ((n-1)T + T_0)}$, a matrix for the unobserved counterfactual post-treatment outcomes for the treated unit $\mathbf{K}_{\mis} \in \R^{(T - T_0) \times (T - T_0)}$, and a rectangular matrix for the kernel between the outcomes for the observed unit-times and the post-treatment outcomes for the treated unit $\mathbf{K}_{\obs, \mis} \in \R^{((n-1)T + T_0) \times (T - T_0)}$. With this setup, the joint distribution of the missing post-treatment outcomes for the treated unit and the observed untreated outcomes are multivariate Normal:
\begin{align}
    \nonumber
    \begin{pmatrix}
    \mathbf{Y^{\mis}} \\
    \mathbf{Y^{\obs}}
    \end{pmatrix} \;\; \sim \;\; \text{MVN}\left[ 
    \begin{pmatrix}
    0 \\ 0
    \end{pmatrix}, \;
    \begin{pmatrix}
    \mathbf{K}_{\mis} & \mathbf{K}_{\text{mis,obs}} \\
    \mathbf{K}_{\text{obs,mis}} & \mathbf{K}_{\text{obs}} 
    \end{pmatrix} + \sigma^2\mathbb{I}
    \right].
\end{align}
The resulting conditional mean and variance have the same form as in Equations \eqref{eq:gp_predmean} and \eqref{eq:gp_predcov}.

As we discuss in Section \ref{sec:weights} below, there are two important special cases of this model, which correspond to setting either the time or unit covariances to be the identity.
Setting the \emph{time} covariance matrix to the identity, $\mathbf{K}_{\text{time}} = \mathbb{I}_T$, corresponds to a model with no smoothness in the underlying latent factors, $\mathbf{u}$, allowing the function value to change arbitrarily between time points. The resulting approach is a (Bayesian) linear factor model \citep[e.g.,][]{xu2017gsynth, athey2018matrix, samartsidis2020bayesian}. In Section \ref{sec:weights} below, we further connect this to ``vertical regression," which finds a set of weights over units that optimize some imbalance criterion \citep{doudchenko2016balancing}.
On the other hand, setting the \emph{unit} covariance matrix to the identity, $\mathbf{K}_{\text{unit}} = \mathbb{I}_N$, corresponds to a model in which there is no correlation in the underlying process for each unit, and each model component is an independent GP. The model therefore reduces back to the single-task GP discussed in Section \ref{sec:single_gp_intro} above.
Finally, while we focus on the ICM model with separate time and unit covariances, there is a large literature on more elaborate generalizations, such as the \emph{semiparametric latent factor model} and the \emph{linear model of coregionalization}; see \citet{alvarez2017gp} for a review.

\paragraph{Computational considerations.}
An important practical constraint in applying multi-task Gaussian processes is the prohibitive computational burden incurred when naively taking samples. Specifically, the Cholesky decomposition of the covariance matrix defined by $\boldsymbol{K}_\text{time} \otimes \boldsymbol{K}_\text{unit}$ has computational complexity cubic with respect to the product of the number of units and time points, $\mathcal{O}((TN)^3)$.
Fortunately, we can take advantage of the structure of the covariance matrix to dramatically reduce this complexity to $\mathcal{O}(T^3 + NJ)$ using properties of Kronecker products and low-rank approximations to each constituent matrix. 
Let $\boldsymbol{C}_\text{time}$ and $\boldsymbol{C}_\text{unit}$ denote the decomposition of the time and unit matrices, $\boldsymbol{K}_\text{time}$ and $\boldsymbol{K}_\text{unit}$, such that $\boldsymbol{K} = \boldsymbol{C}\boldsymbol{C}^\top$, and let $\boldsymbol{C}_\otimes$ be the decomposition of the Kronecker product of $\boldsymbol{K}_\text{time}$ and $\boldsymbol{K}_\text{unit}$.
From the mixed product property of Kronecker products, we have
\begin{align*}
 \boldsymbol{C}_\otimes\boldsymbol{C}_\otimes^\top = \left(\boldsymbol{C}_\text{time} \otimes \boldsymbol{C}_\text{unit}\right)\left(\boldsymbol{C}_\text{time} \otimes \boldsymbol{C}_\text{unit}\right)^\top,
\end{align*}
which allows for the separate decomposition of each covariance. This is $\mathcal{O}\left(T^3\right)$ for the time covariance,\footnote{In cases where $T$ is large, a number of methods can be used to improve computational efficiency further by reducing the complexity of taking the decomposition of $\boldsymbol{K}_\text{time}$ \citep[see][]{solin2020hilbert, wilson2015thoughts, flaxman2015fast}.} and $\mathcal{O}\left(NJ\right)$ for the unit covariance, which uses our low-rank representation $\boldsymbol{K}_\text{unit} = \boldsymbol{\beta}\boldsymbol{\beta}^\top$, $\boldsymbol{\beta} \in \mathbb{R}^{N\times J}$.\footnote{We can further reduce the computational overhead incurred for sampling from this distribution. Let $\boldsymbol{Z} \in \mathbb{R}^{T\times J}, \boldsymbol{Z}_{i,j} \sim \mathcal{N}(0,1)$. 
Sampling by taking $\tilde{\boldsymbol{z}} = \boldsymbol{C}_\text{time}\otimes\boldsymbol{C}_\text{unit}\text{vec}(\boldsymbol{Z})$ is naively $\mathcal{O}(TJ)$. We can reduce this to $\mathcal{O}\left(T^2 + NJ\right)$ by using the ``vectorization trick,'' $\left(\boldsymbol{C}_\text{time}\otimes\boldsymbol{C}_\text{unit}\right)\text{vec}\left(\boldsymbol{Z}\right) = \text{vec}\left(\boldsymbol{C}_\text{unit}\boldsymbol{Z}\boldsymbol{C}_\text{time}^\top\right)$.}

\subsection{Extensions}
\label{sec:extensions}

Building on the broad GP literature, we can immediately extend the MTGP model above to better match our application. We focus on three main extensions here: incorporating a mean model; allowing for a non-Gaussian likelihood; and modeling multiple outcomes simultaneously.

\subsubsection{Incorporating a mean model}
\label{sec:outcome_model}

Thus far, we have considered mean-zero Gaussian Processes for the structural component of the control potential outcomes. A natural extension is to instead allow for the GP to have a non-zero mean, $m_{it}$:\footnote{The distinction between the mean function and the covariance matrix in the MTGP is somewhat artificial, and is typically chosen for clarity and to facilitate the choice of prior; see \citet[][Ch. 2.7]{rasmussen2004gaussian} and  \citet{kanagawa2018gaussian}. In this application, making the mean function explicit allows us to more directly connect the proposed approach to the panel data literature.}
\begin{align*}
    Y_{it}(0) &= f_{it} + \varepsilon_{it} \\
    \mathbf{f} &\sim \mathcal{GP}(\mathbf{m}, k) \\
    \varepsilon_{it} &\sim N(0, \sigma^2),  
\end{align*}
where $\mathbf{m} = (m_{11},\ldots,m_{NT})$ is a length $N \times T$ prior mean vector. Extending the formula in Equation \eqref{eq:gp_predmean}, the posterior mean for a new control potential outcome is therefore:
$$
    \boldsymbol{\mu} = \mathbf{m} +  \mathbf{K}_{\text{mis,obs}}(\mathbf{K}_{\obs} + \sigma^2\mathbb{I})^{-1}\left(
    \mathbf{Y}^\obs - \mathbf{m}\right).
$$
\noindent We return to this representation in Section \ref{sec:worst_case_error} below.

Explicitly incorporating a mean model is useful in a panel data setting because we can more directly connect the MTGP framework to common estimators, such as Comparative Interrupted Time Series and so-called fixed effects models \citep[for a recent review, see][]{samartsidis2019assessing}. One natural model is to set $m_{it}$ to have unit- and time-specific intercepts:
\begin{align*}
m_{it} &= \mu + \text{unit}_i + \text{time}_t \\
\text{unit}_i &\sim N(0, \sigma^2_{\text{unit}}) \\
\textbf{time} &\sim \mathcal{GP}(0, k_{\text{global}})
\end{align*}
where $\mu$ is a global intercept, $\text{unit}_i$ is a unit-specific intercept for unit $i$, and $\text{time}_t$ is a time-specific intercept for time $t$.
The unit-specific intercepts capture the fact that one state may have a higher baseline rate of homicides than another, even if they have similar trends.
The time-specific intercepts capture the trend in the homicide rate across the United States.
This allows us to estimate and account for underlying trends in the homicide rate over time.
Here we place independent Normal priors on the unit-specific intercepts, and a global Gaussian Process prior on the time-specific intercepts. We could similarly allow for a parametric trend over time, such as a baseline linear model common in Interrupted Time Series approaches \citep[see][]{miratrix2020ITS}. 

Finally, we can extend the mean model to incorporate (time invariant) auxiliary covariates, $X_i \in \mathbb{R}^p$:\footnote{With additional restrictions (e.g., ``exogeneity''), we can extend this to include time-varying covariates as well. See \citet{imai2019use} for discussion.}
$$m_{it} = \mu + \text{unit}_i + \text{time}_t + \eta'X_i,$$
where $\eta$ is a vector of regression coefficients, that we can place a prior on. As with the unit-specific intercepts, we could instead include this as part of the GP kernel, for example by incorporating it into the unit covariance.

\subsubsection{Non-Gaussian Likelihoods and heteroskedasticity}
\label{sec:non_gaussian_lik}

Another natural GP extension is to allow for non-Gaussian observation models, which is essential for well-calibrated uncertainty quantification.
Analogous to using alternative link functions with generalized linear models \citep[][Ch. 21.3]{gelman2013bayesian},  we can model the control potential outcomes as:
$$Y_{it}(0) \sim g^{-1}(f_{it}),$$
where $g^{-1}(\cdot)$ is an appropriate link function and $f_{it}$ is the latent GP factor. As above, we could also incorporate a mean model, $Y_{it}(0) \sim g^{-1}(m_{it} + f_{it})$. 
See \citet{naish2008generalized, hensman2015scalable} for discussions of related computational issues.

We highlight two link functions here. First, the outcome in our application is the number of murders, which is more naturally modeled via a Poisson (or, perhaps, Negative Binomial) link. A Poisson model also captures the restriction that the variance in the murder rate increases with the mean murder rate. We use this as our primary model below. 
Alternatively, we could allow for a ``robust'' Gaussian Process by allowing for $t$-distributed rather than Normal errors; see \citet{jylanki2011robust}. In our application, the estimates using $t$-distributed errors are nearly identical to those using Normal errors.

Finally, an important practical feature of our application is that populations vary considerably across states. Thus, we expect that the underlying variability in the murder rate will be much larger in, say Vermont or Wyoming than in California or Texas. With a Gaussian likelihood, we can parametrize this as:
$$Y_{it}(0) \sim N\left(f_{it}, \frac{\sigma^2}{N_{it}}\right)$$
where $N_{it}$ is the population for state $i$ at time $t$. We can similarly allow for a population offset in a Poisson model.\footnote{More broadly, we could incorporate much more complex noise models, such as an additional GP on the noise term \citep[e.g.,][]{naish2008generalized, hensman2015scalable}. This is a challenging model to fit given our limited sample size.}
Allowing for heteroskedasticity by size is an important departure from many Frequentist panel data methods used in this setting; see \citet{samartsidis2019assessing}.

\subsubsection{Multiple outcomes}
\label{sec:multiple_outcomes}

We can extend the MTGP approach to allow for multiple correlated outcomes, following the large literature on multi-output GPs \citep{alvarez2017gp}; see also \citet{samartsidis2020bayesian} for a non-GP setting. This is particularly important for our application because we are interested in assessing the impact of APPS on both gun- and non-gun-related homicides.

To do so, we can write each data point as a unit-time-outcome triple, $(i, t, \ell)$, with control potential outcome, $Y_{it\ell}(0)$. While many models are possible, we focus on a simple version with a shared factor structure across outcomes, which we can write as a multivariate Normal observation model with common covariance across units and time:
$$\mathbf{Y}_{it}(0) \sim \text{MVN}( \mathbf{f}_{it}, \Sigma),$$
where $\mathbf{f}_{it} = (f_{it1}, \ldots, f_{itL})$ and $\mathbf{Y}_{it}(0) = (Y_{it1}(0), \ldots, Y_{itL}(0))$ for outcomes $\ell = 1, \ldots, L$.
Equivalently, we can write this as an MTGP where the resulting kernel decomposes into separate unit, time, and outcome kernels:
$$k(~(i,t, \ell), (i', t', \ell')~) = k_{\text{unit}}(i, i') \;\times\; k_{\text{time}}(t, t') \;\times\; k_{\text{outcome}}(\ell, \ell').$$
Similarly, we can write this as a Kronecker structured overall kernel $\mathbf{K} = \mathbf{K_{\text{unit}}} \otimes \mathbf{K_{\text{time}}} \otimes \mathbf{K_{\text{outcome}}}$, as in \citet{flaxman2015fast}.
Unlike the unit kernel $k_{\text{unit}}(\cdot,\cdot)$, which we restrict to be low-rank, the outcome kernel $k_{\text{outcome}}(\cdot, \cdot)$ can be full rank because the number of outcomes (2) is less than the number of units (50).

\section{Guarantees and connections through a weighting representation}
\label{sec:weights}

Our discussion thus far has been entirely Bayesian. Gaussian Processes, however, are fundamentally linked with (Frequentist) kernel ridge regression; a large literature exploits this connection to describe the Frequentist property of GPs; see \citet{kanagawa2018gaussian} for a recent review. Here we adapt those results to our panel data setting. We first represent the MTGP estimate as a weighting estimator with separate unit and time weights, and give error bounds under fairly general conditions on the latent (noiseless) outcomes. We then describe several special cases to better illustrate this result.

\subsection{Weighting representation}
\label{sec:worst_case_error}
We now show that the MTGP estimate with separable unit and time kernels has a corresponding representation as a weighting estimator with separate unit and time weights. This setup differs from our fully Bayesian formulation above; to make this connection with weighting, we instead consider a Frequentist setup that assumes a known kernel. Specifically, conditional on hyperparameters, the MTGP posterior mean can be written as the solution to a constrained optimization problem that minimizes the worst-case mean-square error across the ``noiseless'' latent functions, $f$. This representation exploits the connection between Gaussian Processes and kernel ridge regression \citep{kanagawa2018gaussian}; see \citet{hazlett2018trajectory} for additional discussion of kernel weighting methods for panel data.

To set up the problem, we again assume that (control) potential outcomes are the sum of structural and error components, $Y_{it}(0) = f_{it} + \varepsilon_{it}$, where $\varepsilon_{it}$ is mean-zero noise. We assume that $f$ is ``well behaved'' in the sense of belonging to a Reproducing Kernel Hilbert Space (RKHS);
this setup is quite flexible and incorporates a wide range of function classes \citep[see, e.g.,][]{wainwright2019high}. We focus on the ICM kernel, $\mathbf{K} = \mathbf{K}_{\text{time}} \otimes \mathbf{K}_{\text{unit}}$, where, as in other Frequentist analyses, the kernel is assumed known.

We can then represent the MTGP as a weighting estimator. We consider two basic forms, corresponding to the weight-space and function-space interpretations of Gaussian processes~(\citet{williams2006gaussian}, Ch. 2). We first show that the MTGP minimizes the worst-case imbalance in the noiseless latent functions $f$. We then show the equivalent formulation in terms of imbalance on the observed outcomes $Y$. We initially consider mean-zero GPs, and incorporate a prior mean function below.

\begin{restatable}{proposition}{mtgpnllicm}
\label{prop:mtgp_nll_icm}
Let $Y_{it}(0) = f_{it} + \varepsilon_{it}$, where $f$ is a fixed function. 
Let $\mathcal{H}_k$ be the RKHS for kernel $k$ with Hilbert-Schmidt norm, $\|f\|_{\mathcal{H}_k}^2 = k(f, f)$.
Then let $f$ be contained in the unit ball of the reproducing Hilbert space, $\|f\|_{\mathcal{H}_k} \leq 1$, where $k(~(i,t), (i', t')~) = k_{\text{unit}}(i, i') \times k_{\text{time}}(t, t')$, and $\varepsilon_{it}$ are iid mean-zero random variables with observed variance $\sigma^2 > 0$.  Further, model the time covariance $\mathbf{K}_{\text{time}}$ with a squared exponential kernel with length scale $\rho$, and model the unit covariance $\mathbf{K}_{\text{unit}} = \beta\beta^\top$. 
Finally, let $\gamma^{(i)} \in \R^{N}$ and $\lambda^{(t)} \in \R^{T}$ be the vectors of unit and time weights for target unit $i$ at time $t$, and let $\boldsymbol{\gamma}$ and $\boldsymbol{\lambda}$, respectively, be the set of weights across all targets.
Then:
\begin{enumerate}
    \item[(a)]
    
The posterior (predictive) mean estimate for target observation $(1, t^\ast) \not\in \mathcal{C}$, for $t^\ast > T_0$ is:
\begin{align*}
    \hat{\mu}_{1t^\ast} = 
    \sum_{(i,t) \in \mathcal{C}}  \gamma_i^{(1)} \lambda_{t}^{(t^\ast)} Y_{it}.
\end{align*}
The unit and time weights, $\boldsymbol{\gamma}$ and $\boldsymbol{\lambda}$, have the following closed form:
\begin{align}\label{eq:weights_closed_form}
\gamma^{(i)} \otimes \lambda^{(t)} = \left(\mathbf{K}_{\text{time}} \otimes \mathbf{K}_{\text{unit}} + \sigma^2 \mathbb{I}\right)^{-1}\left(k_{\text{time}}(t, \cdot) \otimes k_{\text{unit}}(i, \cdot)\right).
\end{align}

\item[(b)] The unit and time weights also solve the following, equivalent optimization problems, in terms of (1) the weight-space formulation:
\begin{align}\label{eq:nll_icm}
    \min_{\boldsymbol{\gamma}, 
    \boldsymbol{\lambda}} \;\;\;
    \sup_{f \in \mathcal{H}_k} \;\;\;
    \sum_{(i,t) \in \mathcal{C}}\left(f_{it} - \sum_{(i', t')\in \mathcal{C}}\gamma^{(i)}_{i'} \lambda^{(t)}_{t'}f_{i't'}\right)^2 
    \;\;+\;\;    \sigma^2\|\boldsymbol{\gamma}\|^2_2\|\boldsymbol{\lambda}\|^2_2,
\end{align} 
and (2) the function-space formulation:
\begin{align} \label{eq:nll_krr}
\min_{\boldsymbol{\alpha}, \boldsymbol{\xi}} \;\;\; \sum_{(i,t) \in \mathcal{C}} \left(Y_{it} - \sum_{(i', t')\in \mathcal{C}} \underbrace{\alpha^{(i)}_{i'}k_{\text{unit}}(i, i')}_{\gamma^{(i)}_{i'}} \cdot  \underbrace{\xi^{(t)}_{t'} k_{\text{time}}(t, t')}_{\lambda^{(t)}_{t'}}Y_{i't'} \right)^2  \;\;+\;\; \sigma^2 \left(\boldsymbol{\xi} \otimes \boldsymbol{\alpha} \right)^\top \mathbf{K}\left(\boldsymbol{\xi} \otimes \boldsymbol{\alpha} \right),
\end{align}
where $\mathbf{K} = \mathbf{K}_{\text{time}} \otimes \mathbf{K}_{\text{unit}}$ and where $\alpha^{(i)} \in \R^{N}$ and $\xi^{(t)} \in \R^{T}$ are the coefficient vectors for target unit $i$ at time $t$, with $\gamma_{i'}^{(i)} = \alpha_{i'}^{(i)}k_{\text{unit}}(i, i')$ and $\lambda_{t'}^{(t)} = \xi_{t'}^{(t)}k_{\text{time}}(t, t')$, and with corresponding sets $\boldsymbol{\alpha}$ and $\boldsymbol{\xi}$.

\item[(c)] This estimate has the following out-of-training-sample estimation error for the structural component of the missing potential outcome $f_{1 t^\ast}$, 
\begin{align} \label{eq:error_bound}
\left| f_{1t^\ast} - \hat{\mu}_{1t^\ast}
  \right| \leq
  \left[ \underbrace{\left(k_{\text{time}}(t^\ast, \cdot)^\top  \lambda^{(t^\ast)}\otimes  k_{\text{unit}}(1, \cdot)^\top\gamma^{(1)}\right)}_{\text{error for } f_{1t^\ast}}
    \;\;+\;\; \underbrace{\sigma^2  \left\|\gamma^{(1)}\right\|^2_2 \left\|\lambda^{(t^\ast)}\right\|^2_2 }_{\text{irreducible noise}}  \right]^{1/2}.
\end{align}

\end{enumerate}
\end{restatable}

Proposition \ref{prop:mtgp_nll_icm} begins with the algebraic result that the posterior mean estimate of the MTGP can be written as a weighting estimator, and that these overall weights separate into unit- and time-weights. This separation does not hold in general, but, in our setting with an ICM kernel, follows immediately from the restriction that the unit and time covariances are separable \citep{bonilla2008multi}. As a result, we can view our proposed MTGP approach as a doubly-weighted panel data estimator, similar to recent proposals from e.g. \citet{ben2018augmented} and  \citet{arkhangelsky2019synthetic} that also impute the counterfactual as a double weighted average.\footnote{The weights we consider here can be negative, unlike some Frequentist weighting methods. Restricting the weights to obey a simplex constraint can be achieved by using the graph Laplacian and considering the Gaussian random field, an analogue of Gaussian processes \citep{zhu2005semi}.}  Importantly, the MTGP implicitly estimates unit and time weights \emph{simultaneously}, while previous proposals construct them separately.

Proposition \ref{prop:mtgp_nll_icm} then leverages Frequentist results on GPs to provide additional intuition for these weights \citep[see][]{kanagawa2018gaussian}.
First, Equation \eqref{eq:nll_icm} shows that the MTGP weights minimize the error on the noiseless outcome functions, subject to regularization on the weights. Importantly, this formulation does not rely on parametric assumptions, only that the noiseless outcomes, $f_{it}$, are relatively ``well behaved'' in the sense of belonging to an RKHS.
Equation \eqref{eq:nll_krr} shows that the weights can be seen as the solution to an equivalent optimization problem in terms of minimizing imbalance in the observed outcomes directly, but subject to regularization on the ``coefficients'' $\alpha$ and $\xi$, rather than on the weights themselves.
Finally, building off this, Equation \eqref{eq:error_bound} shows that the error bound depends directly on the distance between the control observations, $\mathcal{C}$, and the unobserved point, $(1, t^\ast)$.
Specifically, with a squared exponential kernel, $k_{\text{time}}(t^\ast, \cdot)$, the bound is increasing for $t^\ast > T_0$. See \citet{fiedler2021practical} for recent generalizations of such bounds under model mis-specification.

\paragraph{Using a prior mean model.}
The results above are for the MTGP where the prior mean model $m_{it}$ is zero.
We can extend these results to include the prior mean model $m_{it}$ from Section \ref{sec:outcome_model}, recognizing that the distinction between modeling the mean and modeling the covariance can often be arbitrary \citep{williams2006gaussian}.
The resulting posterior (predictive) mean estimate is
\begin{align}
  \label{eq:post_mean_aug_bias_correction}
  \hat{\mu}_{1t^\ast}^{\text{aug}} \;&=\; \sum_{(i,t) \in \mathcal{C}} \gamma_i^{(1)} \lambda^{(t^\ast)}_{t} Y_{it} \;\;+\;\;  \underbrace{\bigg( m_{1t^\ast} - \sum_{(i,t) \in \mathcal{C}}  \gamma_i^{(1)} \lambda^{(t^\ast)}_{t} m_{it} \bigg)}_{\text{estimated bias}} \\[0.5em]
  &=\; m_{1t^\ast} \qquad\qquad\qquad+\quad \underbrace{\sum_{(i,t) \in \mathcal{C}}  \gamma_i^{(1)} \lambda^{(t^\ast)}_{t} (Y_{it} - m_{it})}_{\text{re-weighted residuals}} .
  \label{eq:post_mean_aug_residuals}
  \end{align}
Equation \eqref{eq:post_mean_aug_bias_correction} has a form similar to bias correction for matching \citep{rubin1973matching}, which corrects the doubly-weighted average by an estimate of the difference in the outcome model $m_{it}$. 
Equation \eqref{eq:post_mean_aug_residuals} instead has the form of augmented inverse propensity score weighting \citep{robins1994estimation}, which re-weights the outcome model residuals. 
This echos several recent proposals that combine outcome modeling and weighting in the panel data setting \citep[see, e.g.][]{ben2018augmented, ferman2019synthetic, arkhangelsky2021double}, where here the ``outcome model'' is our prior for the control potential outcome and can include, e.g. fixed effects and time invariant auxiliary covariates as in Section \ref{sec:outcome_model}.

\subsection{Special cases}

To help build intuition for the weighting representation, we consider two special cases: (1) the unit kernel $\mathbf{K_{\text{unit}}}$ is the identity; and (2) the time kernel $\mathbf{K_{\text{time}}}$ is the identity matrix. For simplicity, we take the prior mean model to be zero here, but could extend accordingly.

\paragraph{Identity unit covariance.}
First, let the unit covariance
$\mathbf{K_{\text{unit}}}$ be the identity matrix. 
We can view this as the setting where there is no correlation between the underlying processes for each unit, so each model component is an independent GP. In this setting, there is nothing to be learned from the comparison units, and the MTGP estimate only compares to the pre-treatment outcomes of the treated unit. In particular, following the setup in Proposition \ref{prop:mtgp_nll_icm}, if we set the unit covariance matrix $\mathbf{K_{\text{unit}}}$ to be the identity, then the weights for post-treatment time period $t^\ast > T_0$ are
$$
\lambda^{(t^\ast)} = (\mathbf{K}_\text{time} + \sigma \mathbb{I}_{T_0})^{-1}k_{\text{time}}(t^\ast, \cdot),
$$
with posterior mean estimate:
\[
\hat{\mu}_{1t^\ast}^{\text{horiz}} = \sum_{t=1}^{T_0}\lambda_t^{(t^\ast)} Y_{1t}.
\] 

This is a \emph{horizontal regression} formulation \citep[see][]{athey2018matrix}.
As these weights ignore the outcomes for the comparison units, we can view this case as fitting $N$ separate single-task GP models, one for each unit. 
Thus, when the unit covariance is the identity matrix, this problem ignores the outcomes for comparison units and reduces to the single-task GP for the treated unit discussed in Section \ref{sec:single_gp_intro}.

\paragraph{Identity time kernel.}
Next, let the time kernel $\mathbf{K_{\text{time}}}$ be the identity matrix. This corresponds to the setting where there is no smoothness at all in the underlying latent functions, so the function value can change wildly between time periods. In this case, it is impossible to use outcomes from different time periods to inform our estimates, and so only the unit weights remain. 
Again following the setup in Proposition \ref{prop:mtgp_nll_icm}, if the time covariance matrix $\mathbf{K_{\text{time}}}$ is the identity, then the weights for the treated unit are:
$$
\gamma^{(1)} = (\mathbf{K}_\text{unit} + \sigma \mathbb{I}_{N})^{-1}k_{\text{unit}}(1, \cdot),
$$ with posterior mean estimate:
\[
\hat{\mu}_{1t^\ast}^{\text{vert}} = \sum_{i=2}^{N}\gamma^{(1)}_i Y_{it^\ast}.
\]

\noindent This is a \emph{vertical regression} setup \citep[see][]{doudchenko2016balancing}, and corresponds to the implicit weights of a linear factor model or interactive fixed effects model  \citep{gobillon2016regional, xu2017gsynth, athey2018matrix}. The estimator only uses the learned unit covariance to implicitly construct weights over comparison units.

\section{Model diagnostics and estimated impact for APPS}
\label{sec:apps_main_analysis}

We now use the MTGP framework to estimate the impact of APPS on homicides. We begin by using posterior predictive checks to guide the choice of model. We then estimate the impacts both on overall homicides and separately on gun- and non-gun-related homicides.

\subsection{Posterior Predictive Checks and other model diagnostics}
\label{sec:PPCs}

Diagnostics and other forms of model criticism are important components of the Bayesian workflow, helping researchers diagnose and compare the fit of the proposed models \citep{gelman2013bayesian}. 
We begin with \emph{posterior predictive checks} (PPCs). First, we choose a test statistic, $T(\text{data}, \theta)$, which reflects a relevant aspect of model fit and is a function of both data and model parameters.  We then compare the posterior predictive distribution of this statistic, $T(\text{data}_{\text{rep}}, \theta)$, to the value of the test statistic for the observed data,  $T(\text{data}_{\text{obs}}, \theta)$, over draws of $\theta$.

PPCs play an especially important role in panel data settings like ours, where it is often difficult to decide between modeling approaches given limited data.\footnote{See \citet{liu2020practical} for a review of Frequentist approaches to model checks for panel data.}
We focus on two common measures of model fit: the overall pre-treatment fit and the pre-treatment imbalance at each time point. 
First, we evaluate the pre-treatment fit for California by comparing the posterior predictive distribution of the root mean squared error (RMSE), 
$$T(\mathbf{Y}^\ast, \theta) = \frac{1}{T_0}\sqrt{\sum_{t = 1}^{T_0} (Y^\ast_{1t} - f_{1t})^2},$$
to the observed RMSE, $T(\mathbf{Y}^\text{obs}, \theta) = \frac{1}{T_0}\sqrt{\sum_{t = 1}^{T_0} (Y_{1t} - f_{1t})^2}$, where $Y^\ast_{1t}$ is a posterior predictive draw of the overall murder rate for California at (pre-treatment) time $t$.
Figure \ref{ppc:rmse} shows the distributions for both Gaussian and Poisson observation models with unit kernel ranks $J = 0,\ldots, 7$. For each combination, we also compute the posterior predictive p-value, the fraction of posterior predictive test statistics that are larger than the observed test statistic.  
For both Gaussian and Poisson models, the observed RMSE is larger than we would expect for ranks 0 through 3, with minimal discrepancies by rank 5 for both models. The Poisson model has slightly lower RMSE, suggesting that Poisson with rank 5 is a reasonable initial choice here.

\begin{figure}[btp]
\begin{subfigure}[b]{.5\linewidth}
    \centering
    \includegraphics[width=0.95\textwidth]{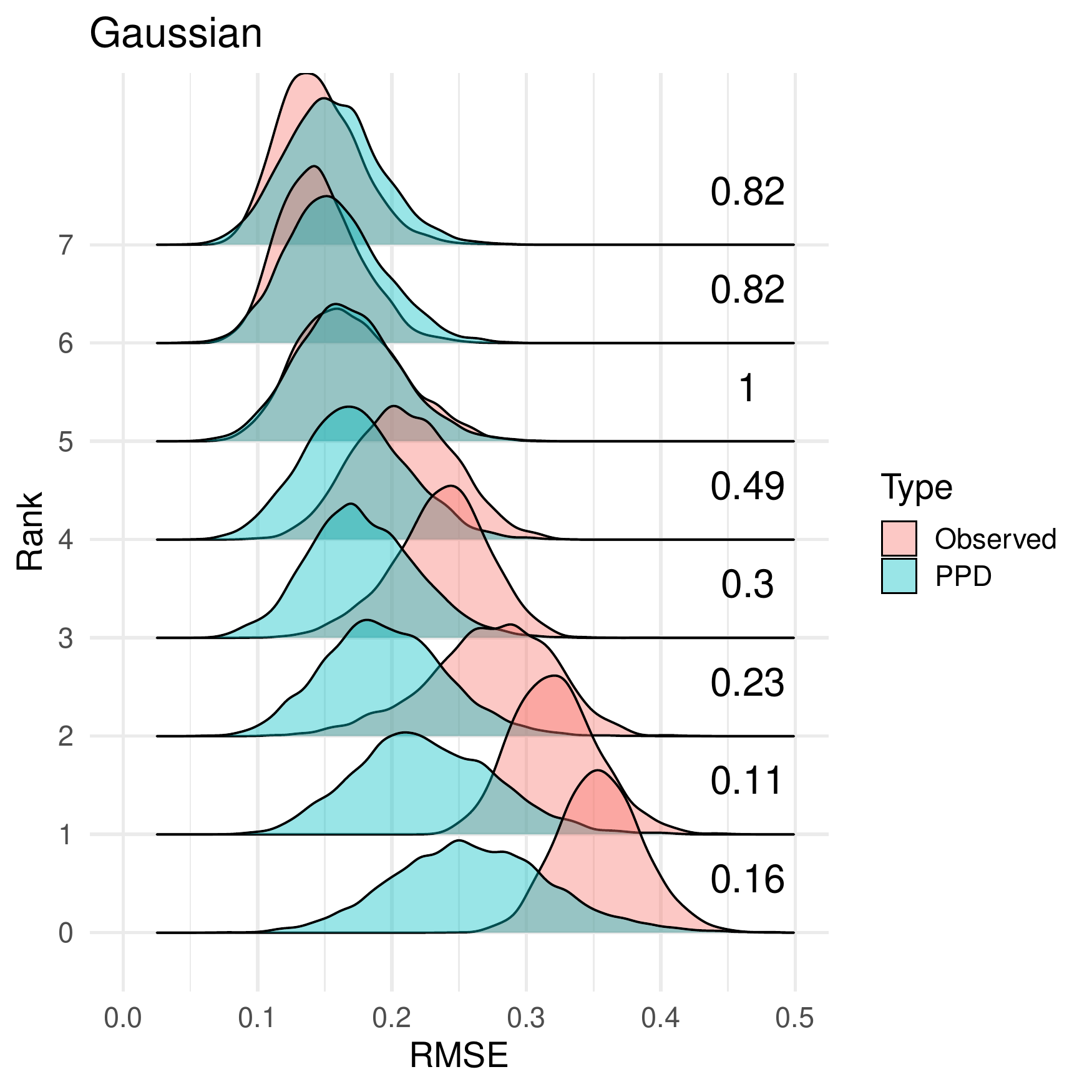}
    \caption{PPCs for Gaussian likelihood}
    \label{fig:mtgp_ests}
\end{subfigure}
\begin{subfigure}[b]{.5\linewidth}
    \centering
    \includegraphics[width=0.95\maxwidth]{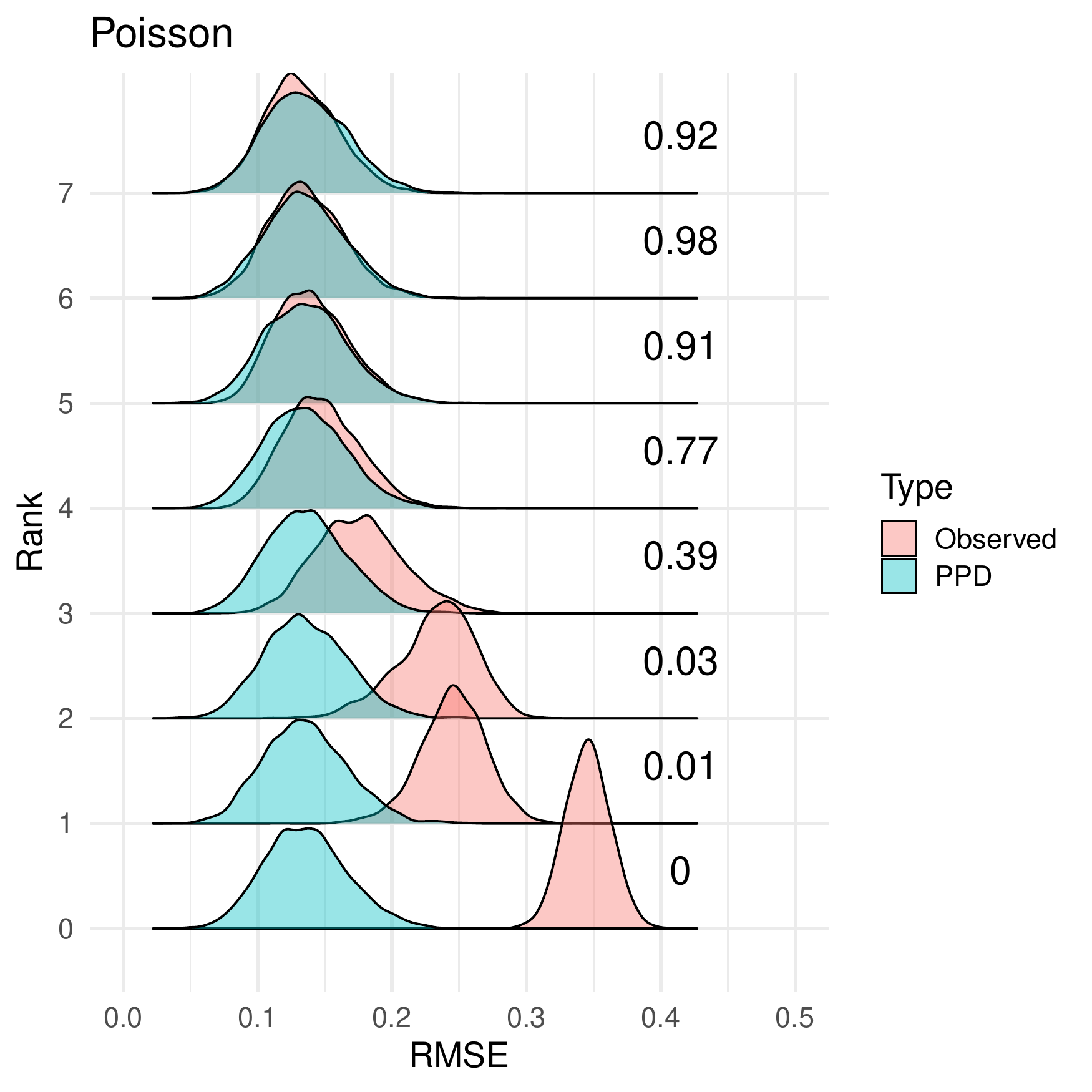}
    \caption{PPCs for Poisson likelihood}
    \label{fig:mtgp_ests_poisson}
\end{subfigure}
\caption{Posterior Predictive Checks (PPCs) for the pre-treatment RMSE for California.  Numbers indicate posterior predictive p-value associated with each rank, i.e., $\mathbb{P}\left(T(\text{data}_{\text{rep}}, \theta) > T(\text{data}_{\text{obs}}, \theta) \mid \text{data} \right)$.}
\label{ppc:rmse}
\end{figure}

\begin{figure}[tbp]
    \centering
    \includegraphics[width=0.6\maxwidth]{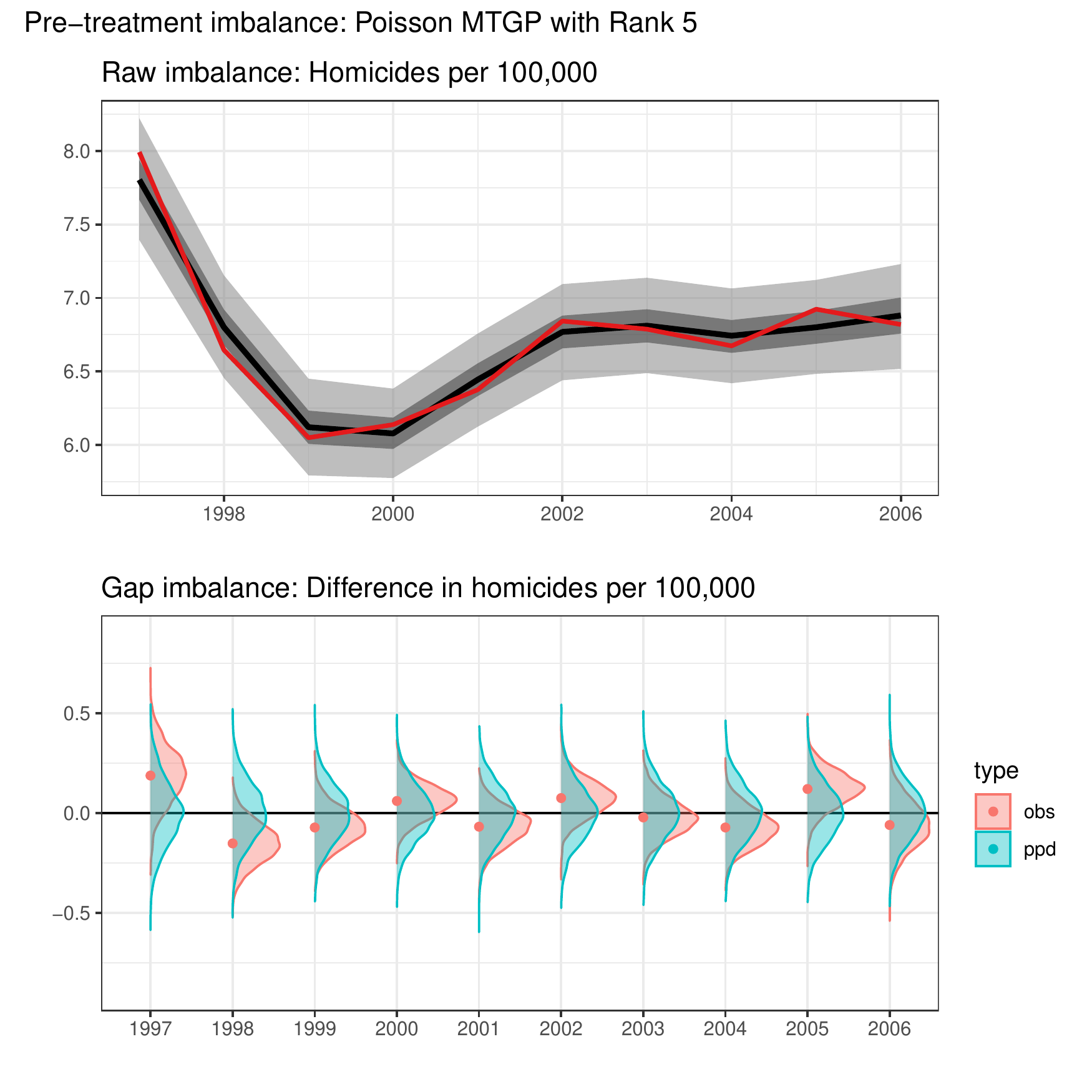}
    \caption{ Posterior predictive checks for pre-treatment imbalance with the Poisson model with Rank 5. Observed data is in red. The upper panel shows the posterior predictive distribution and original data on the scale of the raw imbalance. The lower panel shows the corresponding ``gap'' estimates for each year, comparing the distribution of $Y_{1t}- \hat{\mu}_{1t}$ to the distribution of $Y_{1t}^{\star} - \hat{\mu}_{1t}$ for times $t = 1,\ldots,T_0$. These placebo estimates are overlapping, giving further support to good model fit.}
    \label{fig:ca_balance_pois_rank5}
\end{figure}

Figure \ref{fig:ca_balance_pois_rank5} shows the imbalance at each pre-treatment time point for the Poisson MTGP model with rank 5. The upper pane shows the observed data (in red) as well as posterior predictive draws from the model on the scale of the original outcome (homicides per 100,000). The lower pane instead shows the ``gap'' between the outcome and the posterior mean, $T(\text{data}, \theta) = Y_{1t} - \hat{\mu}_{1t}$.\footnote{This is the Bayesian analog to estimates of placebo impacts that are standard in Frequentist panel data studies, although we compare the observed estimates to the posterior predictive distribution rather than to zero.}
Both plots show little difference between the observed data and the posterior predictive distribution, suggesting reasonable model fit along the dimensions we investigate here.

Finally, the Appendix includes several additional model diagnostics. Appendix Figure \ref{fig:coverage_ppc} shows the coverage of the 50\% and 95\% posterior predictive intervals for pre-treatment outcomes across all 50 states, computed as $\frac{1}{NT}\sum_{i, t}I(L_{it}^q < Y_{it} <  U^q_{it})$ where $L_{it}^q$ and $U^q_{it}$ are the lower and upper endpoints of the credible region with probability $q$ for unit $i$ at time $t$. The Poisson intervals have close to nominal coverage for ranks 4 and 5, but under-cover for lower ranks and over-cover for higher ranks. By contrast, the corresponding Gaussian intervals over-cover across all ranks, with especially poor coverage for the 50\% intervals for higher ranks, leading further credence to the Poisson model over the Gaussian model. This difference is likely due to the fact that the sampling variance increases with the mean in the Poisson model but is constant in the mean and across units for the Gaussian model.

In Appendix Figures \ref{fig:global_factor} and \ref{fig:all_factor} we show the inferred Gaussian Process for the time-specific intercept and the estimated factors from the rank 5 Poisson model next to the factor loadings for California. In Appendix Figures \ref{fig:time_weights} and \ref{fig:unit_weights}, we show the inferred time and unit weights for California in the first post-treatment period using a rank 5 Gaussian multitask Gaussian process without the global time series trend GP. Each of the time and unit weights are obtained by marginalizing over the non-relevant component (unit, or time, respectively) for each sample from the MCMC sampler.

\subsection{Impact of APPS on homicides}
\label{sec:impact_apps}
We now turn to estimating the impact of APPS on homicides in California. We begin with the overall impact and then disaggregate by gun and non-gun-related homicides.

\paragraph{Overall impact of APPS.}
Based on our model diagnostics above, we focus on the MTGP estimated with rank $J=5$, Poisson observation model, and unit- and time-specific intercepts. 
Figure \ref{fig:impact_time} shows the estimated impact of APPS on homicide rates in California over time.\footnote{Appendix Figure \ref{fig:mtgp_ests_poisson_adj} shows the corresponding estimates after adjusting for auxiliary covariates and several measures of crime, all for 2005. The auxiliary covariates are: the prison incarceration rate, the age distribution (percent 0-17, 18-24, 25-44, 45-64, 65 and older), the percent Black, the unemployment rate, the poverty rate, and log median income. The crime measures, normalized by population, are: overall violent crime, rape, robbery, assault, property crime, burglary, larceny, and motor/vehicular crime. As the results are largely unchanged, we focus on the unadjusted estimate here.}
The upper pane shows the posterior predictive distribution from the MTGP, with the posterior mean as well as 50\% and 95\% posterior predictive intervals. The red line shows the observed murder rate for California. The estimates pre-treatment repeat the pre-treatment imbalance checks in Figure \ref{fig:ca_balance_pois_rank5}; the counterfactual estimates post-treatment are consistently higher than the observed murder rate. 
The lower pane shows this difference explicitly, where zero is no impact. 
These estimates suggest that APPS reduced murders in the state, with both impacts and uncertainty growing over time.

\begin{figure}[tbp]
    \centering
    \includegraphics[width=0.6\textwidth]{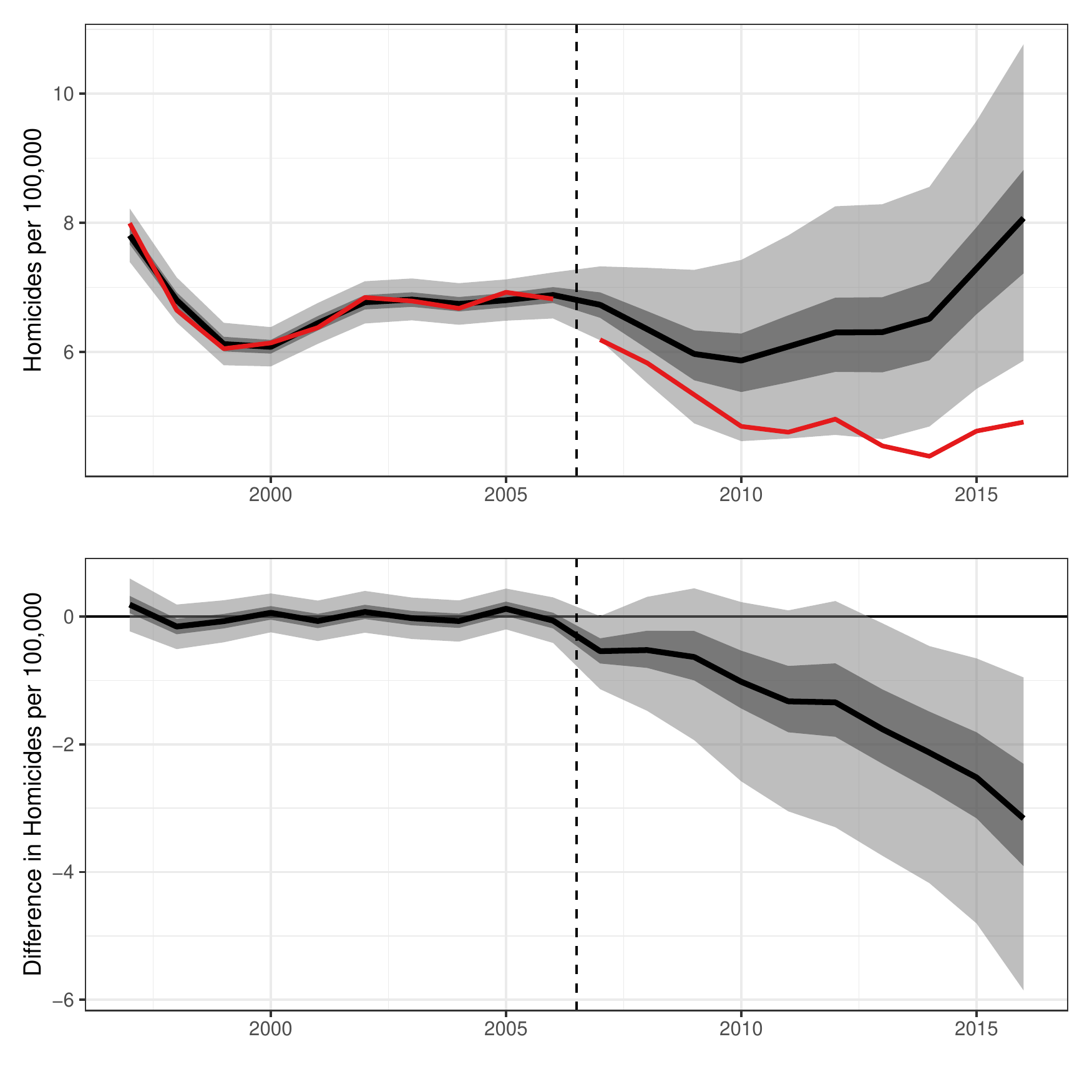}
    \caption{Impact of APPS on homicides per 100,000 in California by year. The observed data is show in red. The black line is the posterior mean estimate; the dark and light gray regions are, respectively, the 50\% and 95\% posterior predictive intervals.}
    \label{fig:impact_time}
\end{figure}

Figure \ref{fig:impact_overall} shows the posterior estimates for the average annual impact. The posterior mean estimate is roughly 1.5 murders avoided per 100,000 people per year, with a 95\% credible interval of $(0.32, 2.9)$ murders per 100,000 people; the posterior probability that the impact is negative is 99.6\%.
To contextualize these estimates, California's annual murder rate in 2004 and 2005 was 6.8 per 100,000. Focusing on the smaller end of the 95\% credible interval suggests a decline of at least 5 percent from baseline. 
We can also compare the annual number of murders prevented by this system to expenditures for APPS investigative teams in a single year, 2018, when California's population was 39.6 million.  A decrease in the murder rate of 0.32 corresponds to nearly 127 murders avoided per year, 
while the total funding for APPS investigation teams in fiscal year 2017-2018 was \$11.3 million \citep{petek2019budget}.
Dividing this cost of \$11.3 million by the estimated 127 avoided murders gives a (conservative) estimate of roughly \$89,000 per murder prevented.
This is dramatically lower than the estimated benefit from the avoided crime, which is typically on the order of \$10 million \citep{heaton2010hidden, dominguez2015role}. We include the posterior distribution of expenditures per average number of murders avoided in Figure \ref{fig:impact_overall}. 

\begin{figure}[tbp]
    \centering
    \includegraphics[width=0.6\textwidth]{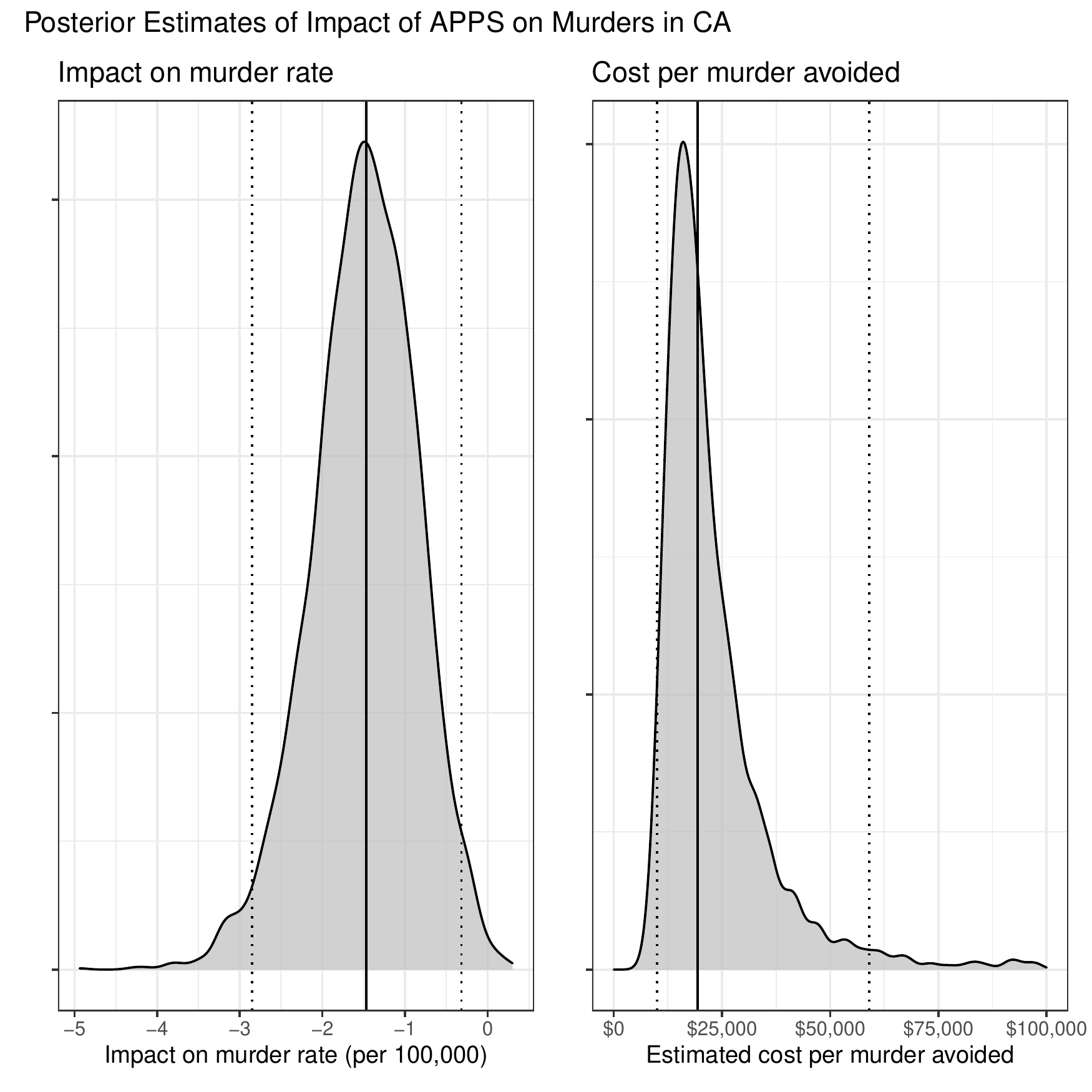}
    \caption{(left) Posterior distribution of the average impact of APPS on homicides per 100,000 in California for 2007 to 2016; (right) Posterior distribution of the cost per murder avoided, with a state population of 39.6 million and annual budget expenditure of \$11.3 million. The solid line denotes the posterior mean; the dashed lines denote 95\% credible intervals.}
    \label{fig:impact_overall}
\end{figure}

One way to assess whether these estimates are potentially spurious is to conduct an ``in-time'' placebo check, where we set the APPS adoption year to be earlier than 2006 \citep{abadie2010synthetic}. We can then compare the difference between the observed homicide rate and its posterior in the years between the placebo APPS year and 2006.
Because we do not expect an effect in these held-out years, the difference should be near zero. Figure \ref{fig:time_placebo} shows these placebo estimates for placebo APPS over four years (2002-2005). We see that the posterior distribution is indeed near zero for all four placebo estimates, indicating that our preferred MTGP model finds no treatment effects when we expect there to be none. Note, however, that the relatively small number of pre-intervention years leads to several limitations to this placebo check.
First, with the earlier placebo years there is a great deal of uncertainty due to the limited number of pre-intervention years. Second, we only evaluate placebo effects up to four years after the placebo APPS year. Settings with more time periods would have less uncertainty in the placebo check and would allow us to evaluate placebo effects for more years, which would enable us to assess the gradual increase in the magnitude of the treatment effect we discuss below.

\begin{figure}[tbp]
    \centering
    \includegraphics[width=0.5\textwidth]{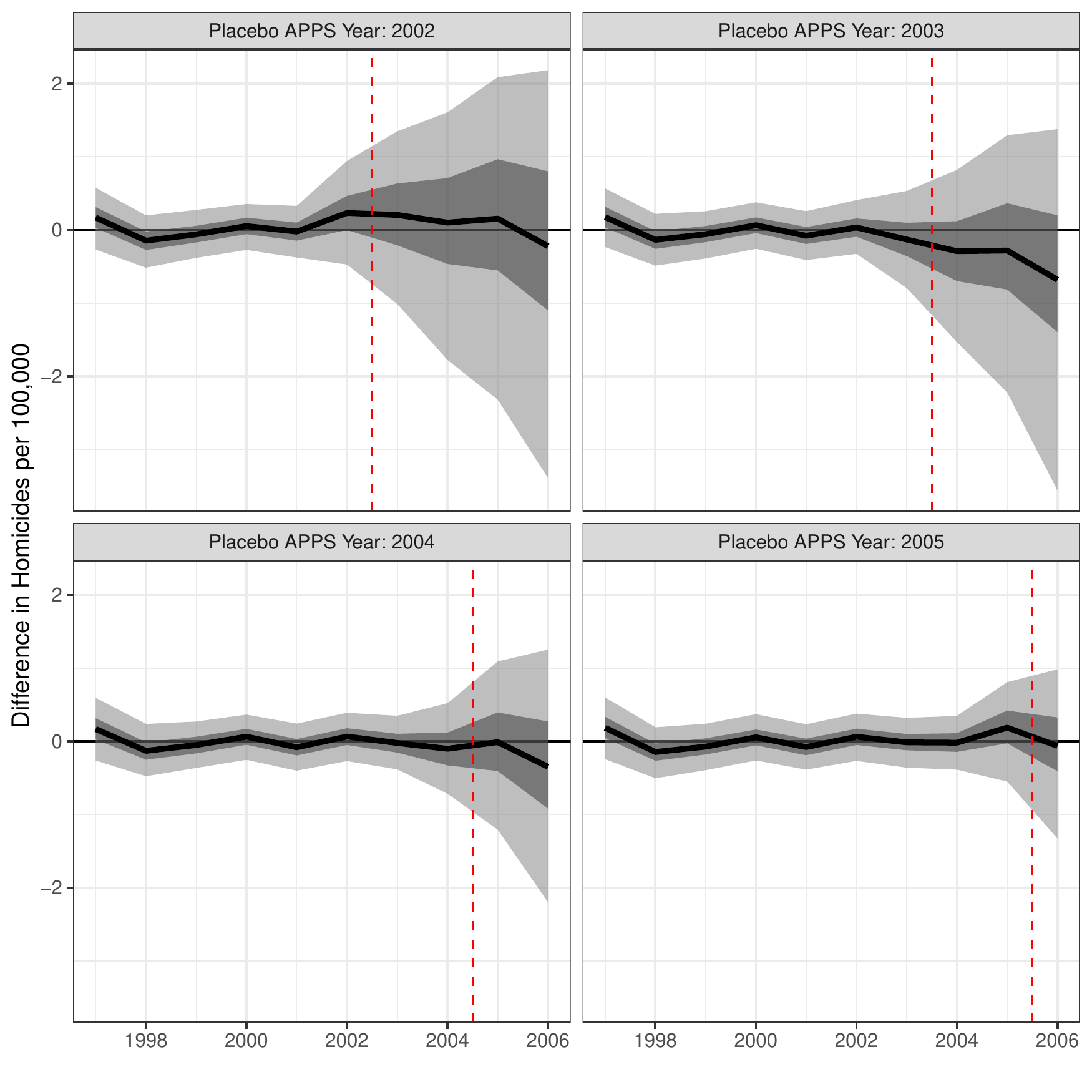}
    \caption{In-time placebo estimates for the rank 5 Poisson MTGP model, setting the placebo APPS year to be 2002-2005.}
    \label{fig:time_placebo}
\end{figure}

Our preferred estimates for the impact of APPS are similar to the preliminary SCM estimates in Section \ref{sec:synth}.
Both estimate that the counterfactual California would have seen the homicide rate drop more slowly and plateau at a higher level, before seeing a larger rise between 2014-2016.
Both procedures estimate similar average impacts (a posterior mean of -1.5 vs. -1.41 for SCM), but the MTGP approach allows us to coherently quantify the uncertainty in the impacts, incorporating the fact that we are measuring impacts on homicide \emph{counts}. Here we see that there is substantial uncertainty in the impact, but that there is a very high probability the impact is negative.
In Appendix Figure \ref{fig:alt_ests}, we compare these estimates to those from various
alternative estimators: (i) SCM without an intercept shift; (ii) Ridge Augmented SCM \citep{ben2018augmented}; (iii) fitting a low-rank factor model via \texttt{gsynth} \citep{xu2017gsynth}; and (iv) low-rank matrix completion \citep{athey2018matrix}.
These estimators broadly agree with our MTGP approach, with point estimates mostly within the 50\% posterior predictive intervals and all well within the 95\% posterior predictive intervals.
Finally, we consider the relative impact of the other 49 comparison states on our estimates via a ``leave one out'' approach that removes each comparison state from the panel, finding the posterior without conditioning on that state's homicide rate series.
We show the results in Appendix Figure \ref{fig:loo}.
Overall we find that removing each state has little effect on the posterior distribution; the largest differences come from removing Illinois, New Jersey, and Texas, but these estimates still broadly comport with our main estimates.

Our findings are also suggestive of larger effects in later years relative to earlier years.  While these results may be due to an increase in murder rates in the comparison states, there are aspects of the APPS program that could lead to larger effects over time. First, there might be a lag in the effect of removing a firearm on violent incidents.
Since the program reduces the stock of firearms over time, the effects may similarly increase along with the stock of firearms removed.
Moreover, since people usually enter the APPS database upon purchasing a new gun, the proportion of guns and gun owners covered by the system is likely increasing over time as well.

California has also implemented further restrictions on firearms since 2006, including:
banning the carrying of an unloaded handgun in public (2012); requiring new semi-automatic hand guns to have microstamping technology (2013);  requiring tougher safety tests for first time buyers (2014); requiring background checks for ammunition and limiing ammunition sales to individuals with a registered firearm in the APPS system (passed in 2016 but implemented in June 2019); and most recently imposing a 21 year age limit for firearm purchases and a lifetime ban for prior convictions for a violent misdemeanor \citep{christopher2019}.
While some of these reforms occur after the end of our study period, it is clear the state has increased the regulation of firearms above and beyond what has occurred in other states. To the extent that some of these other policies have reduced murder rates, our estimates will reflect the joint effect of the APPS intervention along with these other policies.
This also potentially contributes to the effect estimates increasing over time.
That being said, many of these policy changes were only possible due to the underlying information and infrastructure that support the APPS enforcement effort.

\paragraph{Separate impacts on gun- and non-gun-related homicides.}

As we discuss in Section \ref{sec:synth}, one falsification check involves disaggregating homicides by whether a firearm was used.  Specifically, we do not expect that the APPS system would have reduced non-gun homicides.  In fact, to the extent that there is substitution across weapons we might even expect that APPS may have increased homicide rates where firearms are not used.  

To explore this, we use the \citet{kaplan2019} concatenation of the UCR Supplementary Homicide Files, restricted to our study period (1997 through 2016) and use the incident level data to calculate totals by state and year of the number of murders involving a firearm and the number of murders where a gun is not used.\footnote{We identify gun homicides as any incidents where the weapon used by any of the offenders in the incident are flagged as ``firearm, type not stated,'' ``handgun,'' ``other gun,'' ``rifle,'' or ``shotgun.'' Note, for a small percent of cases the weapon used is unknown (roughly 5 percent).  We allot these observations to non-gun homicides, which should bias us towards finding an effect of APPS on non-gun homicides.}  We then use these totals in conjunction with state level population counts to tabulate rates per 100,000.\footnote{For these models, we restrict the donor pool to states where we observe homicide reports for all years from 1997 through 2016; inclusive of California, 43 states report data for all years.  Murder victims in the SHR data by year and state often do not add to the officially reported titles in summary tables provided in the official UCR panel data.   For some state-years, the ratio of the murder rate tabulated from   SHR data to that tabulated from the UCR summary data is less than one while for others the value is above one.  The median of this ratio is 1.02 and the mean is 1.005. } In general, gun murder rates are two to three times that of non-gun murder rates.

\begin{figure}[tbp]
    \centering

    \includegraphics[width=0.6\textwidth]{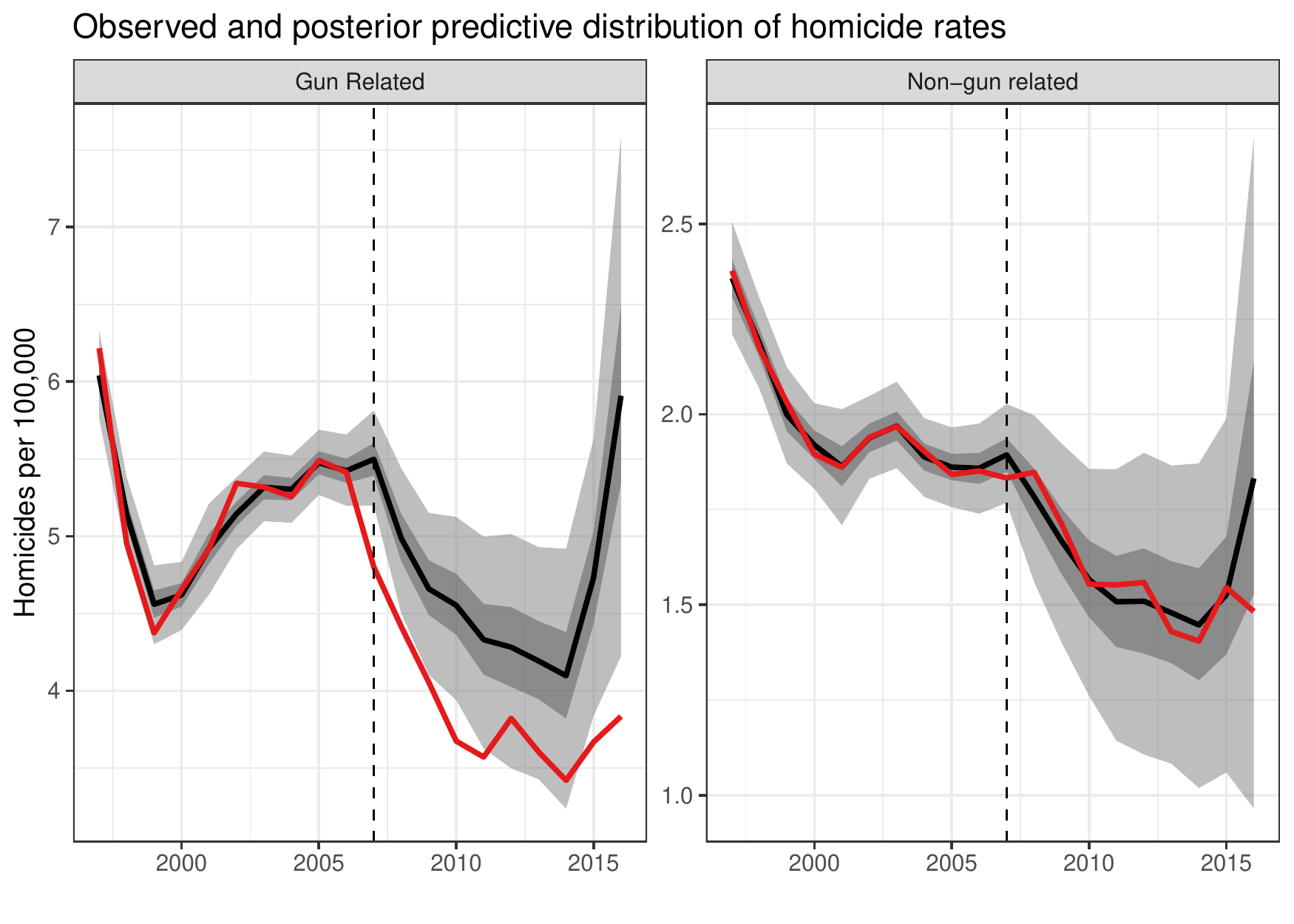}
    \caption{Impact of APPS on gun- and non-gun-related homicides per 100,000 in California by year. The observed data is show in red. The black line is the posterior mean estimate; the dark and light gray regions are, respectively, the 50\% and 95\% posterior predictive intervals.}
    \label{fig:multi_outcomes}
\end{figure}

Figure \ref{fig:multi_outcomes} shows estimates for gun- and non-gun-related homicides, estimated via a joint multitask GP. Fitting both outcomes jointly is advantageous since these two outcomes are positively correlated due to shared latent factors. (See Appendix Figure \ref{fig:mo_cor}, for the posterior predictive correlation between gun and non-gun related homicides for California). While estimates are derived from a different data source than used above, the impacts for gun-related homicides are very similar to the estimated overall impacts in Figure \ref{fig:impact_overall}.  In contrast, the estimated impact of APPS on non-gun-related homicides is essentially zero.

\section{Discussion}
\label{sec:discussion}

This paper uses a multitask Gaussian Processes framework in a panel data setting to estimate the causal effect of an innovative gun control program in California. MTGPs find a natural middle ground between pure ``horizontal'' time series forecasting models, based on extrapolating the treated unit's outcomes series alone, and pure ``vertical'' regression models, based on a weighted average of control units' outcomes. Further, MTGPs automatically quantify the relevant sources of uncertainty, providing a unified Bayesian framework in a setting where uncertainty quantification is notoriously challenging. They also allow for extensive model diagnostics via posterior predictive checks. Building on a robust GP literature, there are also many immediate extensions to the standard MTGP setup, including allowing for multiple outcomes and incorporating a mean model.
Using this flexible approach, we find large effects of APPS on homicides in California, with program benefits far exceeding costs under typical assumptions.

There has been a flurry of policy activity since 2006 in the criminal justice domain, which could change crime rate patterns and the interpretation of our findings. First, since 2011 the state has cut back on using state prison for parole violations and for less serious felony offenses.
Second, various state ballot initiatives have scaled back the severity of sentencing \citep{lofstrom2016california}.
These reforms have reduced the overall incarcerated population (prisons and jails combined) of the state by roughly one-quarter, and may have led to higher crime rates through reverse incapacitation effects.
However, several studies of these reforms fail to find any evidence of an impact of these changes on violent offenses \citep{lofstrom2016incarceration, bartos2018can}, and so we do not believe that our results are biased by these policy changes.

Methodologically, there are several promising directions for future research. 
First, there are many existing results for Gaussian Processes that would apply immediately in other settings but are not as relevant in our application. For instance, standard panel data models face a range of practical challenges when the data structure deviates from a so-called ``balanced panel,'' either due to irregularly sampled observations over time, missing data, or comparison states enacting policies similar enough to the policy of interest to warrant exclusion from the sample \citep[see][]{imai2019use}.
The MTGP framework naturally accommodates both irregularly sampled time points and missing values, and, unlike most Frequentist approaches, will automatically propagate the corresponding uncertainty. 

Additionally, while California and the APPS program is uniquely comprehensive, several states have enacted forms of firearm surrender or removal policies for, e.g., those subject to a domestic violence restraining order.
Therefore, an important extension of this approach is
to allow for multiple treated units that adopt treatment over time, also known as staggered adoption \citep[e.g.,][]{ben2019synthetic}. First, we can extend to the case where multiple treated units adopt treatment at the \emph{same} time (i.e., simultaneous adoption) by changing the target to the \emph{average} of the treated units. However, the extension to the case when units adopt over time (i.e., staggered adoption) is slightly more complicated. For instance, we could adapt the model to have a separate kernel for each treatment time, possibly including a hierarchical structure on the kernels \citep{flaxman2015fast}.

Next, we can develop a framework for assessing sensitivity to departures from key modeling assumptions for the MTGP approach, especially the restriction that the time and unit kernels are separable. While there are many possible generalizations \citep[see, for example][]{alvarez2017gp}, these are inherently under-specified in settings like ours with relatively few units and time periods. Thus, a sensitivity analysis approach in the spirit of \citet{franks2019flexible} is a natural avenue for appropriately quantifying uncertainty in this case.  

Finally, our proposed MTGP approach implicitly puts the same weight on the error for California as on other states, even though California is the target unit. We can explore alternative approaches for parameterizing similarity across units and time, such as by imposing an ``arrowhead'' structure on the overall covariance, or by learning a richer dependence structure across units; see, for example, \citet{li2020negative}. Moreover, our approach focuses on modeling the control potential outcomes for California, in the spirit of finite sample causal inference --- but allowing for potentially arbitrary treatment effects. We could also directly model the treatment effect, including by incorporating a prior, as part of the MTGP framework.

\clearpage
\singlespacing
\bibliographystyle{chicago}
\bibliography{references}

\clearpage
\appendix

\setcounter{figure}{0}  
\setcounter{table}{0}
\setcounter{theorem}{0}
\setcounter{corollary}{0}

\renewcommand\thefigure{\thesection.\arabic{figure}} 
\renewcommand\thetable{\thesection.\arabic{table}}    
\renewcommand\thetheorem{\thesection.\arabic{theorem}}    
\renewcommand\thecorollary{\thesection.\arabic{corollary}}

\section{Estimation and implementation choices and further discussion}
\label{sec:implementation}

\subsection{Single outcome}

Both the Gaussian and Poisson model share the following three components to model the latent process
\begin{enumerate}
    \item Each unit, $i$, is given an independent offset drawn from a standard normal distribution, $\nu \sim \mathcal{N}(0,1)$, providing a unit-specific intercept.
    \item The global trend is modeled as a single-task Gaussian process with lengthscale $\rho_{\text{global}}$, 
    \begin{align*}
        \rho_{\text{global}} &\sim \text{InvGamma}(5,5)\\
        g &\sim \mathcal{N}(0, \mathbf{K}_{\text{global}}; \rho_{\text{global}})
    \end{align*}
    where $\mathbf{K}_{\text{global}}$ is a covariance matrix with each entry given by the squared exponential kernel with lengthscale, $\rho$, i.e. $\mathbf{K}_{\text{global}}(i,j) = \exp\left(-\frac{\|x_i - x_j\|^2}{2\rho}\right)$.
    \item Individual components are modeled with a multitask Gaussian process with a rank $k$ task-covariance matrix,
    \begin{align*}
        \rho_{\text{time}} &\sim \text{InvGamma}(5,5)\\
        \alpha_{\text{time}} &\sim \mathcal{N}(0, 1)\\
        \mathbf{\beta} &\sim \mathcal{N}(0, \mathbf{I}_k)\\
        \mathbf{K}_{\text{unit}} &= \mathbf{\beta}\mathbf{\beta}^T\\
        \text{vec}(\mathbf{f}) &\sim \mathcal{N}(0, \mathbf{K}_{\text{time}} \otimes \mathbf{K}_{\text{unit}}; \rho_{\text{time}}, \alpha_{\text{time}})
    \end{align*}
    \item The variance for a Normal outcome is drawn from $\sigma^2 \sim \mathcal{N}(0, 1)$
\end{enumerate}
With $\mathbf{K}_{\text{time}}$ is a squared exponential kernel of the same form as $\mathbf{K}_{\text{global}}$. 
For the Gaussian link, a shared standard deviation is drawn, $\sigma \sim \mathcal{N}(0,1)$, and
the observed control outcomes for each unit, $i$, are modeled as
\begin{align*}
    \frac{\mathbf{Y}_i}{\mathbf{N}_i} \sim \mathcal{N}\left(\nu_i + \mathbf{f}_i + g, \frac{\sigma^2}{\sqrt{\mathbf{N}_i}}\right),
\end{align*}
where $\mathbf{N}$ is a matrix containing the population for each state for all time periods.
Similarly, for the Poisson link function the observed untreated outcomes for each unit is modeled as
\begin{align*}
    \mathbf{Y_i} \sim \text{Pois}\left(\mathbf{N}_i +  \exp\left(\nu_i + \mathbf{f}_i + g\right)\right),
\end{align*}
where $N_i$ is an offset.

\subsection{Multiple outcomes}
The multi-outcome models analogously use the following components, 
\begin{enumerate}
    \item The hyperparameters of the global and individual kernels are drawn as 
    \begin{align*}
        \rho_{\text{global}} &\sim \text{InvGamma}(5,5)\\
        \alpha_{\text{global}} &\sim \mathcal{N}(0, 1)\\
    \end{align*}
    \item The global time trend for each outcome is drawn as
    \begin{align*}
    g &\sim \mathcal{N}(0, \mathbf{K}_{\text{global}}; \rho_{\text{global}}, \alpha_{\text{global}})
    \end{align*}
    \item Individual components are modeled as
    \begin{align*}
        \rho_{\text{time}} &\sim \text{InvGamma}(5,5)\\
        \mathbf{\beta} &\sim \mathcal{N}(0, \mathbf{I}_k)\\
        \mathbf{K}_{\text{unit}} &= \mathbf{\beta}\mathbf{\beta}^T
    \end{align*}
    \item Each unit is given an independent intercept for each outcome drawn from a standard normal distribution $\nu \sim \mathcal{N}(0,1)$.
    \item The covariance between outcomes is modeled as 
    \begin{align*}
        \mathbf{Z} \sim \mathcal{N}(0, \mathbf{I})\\
        \mathbf{K}_{\text{outcome}} := \mathbf{ZZ}^T
    \end{align*}
    and the latent function is drawn as
    \begin{align*}
        \text{vec}(\mathbf{f}) &\sim \mathcal{N}(0, \mathbf{K}_{\text{time}} \otimes \mathbf{K}_{\text{unit}} \otimes \mathbf{K}_{\text{outcomes}}; \rho_{\text{time}})
    \end{align*}
    \item The variance for each Normal outcome is drawn from $\sigma^2 \sim \mathcal{N}(0, 1)$
\end{enumerate}
Finally the probability of each outcome is modeled as
\begin{align*}
\frac{\mathbf{Y}_{i,j}}{\mathbf{N}_i} \sim \mathcal{N}(\nu_{i,j} + \mathbf{f}_{i,j}  + g_{j}, \frac{\sigma^2_j}{\sqrt{\mathbf{N}}_i})
\end{align*}
for the normal outcome for unit $i$ and outcome $j$, and 
\begin{align*}
    \mathbf{Y}_{i,j} \sim \text{Pois}\left(\mathbf{N}_i + \exp{\nu_{i,j} + \mathbf{f}_{i,j} + g_j} \right)
\end{align*}
for the Poisson.

\subsection{Synthetic Control Implementation}
\label{sec:appendix_synth}
The goal of the \emph{synthetic control method} \citep{abadie2010synthetic, abadie2019synth_review} is to find a weighted average of other states that approximate the pre-intervention murder rate for California. 
We use a modified version of SCM  originally proposed by \citet{doudchenko2016balancing}, known as \emph{SCM with an intercept shift}; see also \citet{ferman2019synthetic} and \citet{ben2018augmented}.
This finds weights $\hat{\gamma} \in \R^{N-1}$ and an intercept $\hat{\alpha}$ that solve the following constrained optimization problem
\begin{equation}\label{eq:synth_minimization}
    \begin{aligned}
    \min_{\alpha, \gamma} \;\;\;\;& \sum_{t=1}^{T_0} \left(Y_{1t} - \alpha - \sum_{j=2}^N \gamma_j Y_{jt}\right)^2\\
    \text{subject to} \;\;\;\; & \sum_{j=2}^{N} \gamma_j = 1\\
& \gamma_j \geq 0 \;\;\; j = 2, \ldots, N.
\end{aligned} 
\end{equation}
The SCM estimator of $Y^\ast_{1t}$ at time $t$ is then $$\hat{Y}_{1t}^\ast = \hat{\alpha} + \sum_{j=2}^{N} \hat{\gamma}_jY_{jt},$$
with effect estimate $\hat{\tau}_t = Y_{1t} - \hat{\alpha} - \hat{Y}_{1t}^\ast$.
There is a large literature on extensions and modifications to SCM; see \citet{doudchenko2016balancing} and \citet{abadie2019synth_review}. The implementation in Equation \eqref{eq:synth_minimization} differs from the original proposal in \citet{abadie2010synthetic} in two ways.
First, this optimization gives equal weight to all lagged outcomes, and does not use auxiliary covariates.
Second, it includes the intercept shift.

The intercept has a closed form solution, 
\[\hat{\alpha} = \bar{Y}_1 - \sum_{j=2}^N \hat{\gamma}_j \bar{Y}_j,\]
where $\bar{Y}_j = \frac{1}{T_0}\sum_{t=1}^{T_0}Y_{jt}$.
Plugging this in to Equation \eqref{eq:synth_minimization} shows that the intercept-shifted weights are equivalent to finding synthetic control weights with the residuals $Y_j-\bar{Y}_j$ in place of the unadjusted outcomes $Y_j$,
\begin{equation}\label{eq:synth_minimization_2}
    \begin{aligned}
    \min_{\gamma} \;\;\;\;& \sum_{t=1}^{T_0} \left(Y_{1t}  - \bar{Y}_1 - \sum_{j=2}^N \gamma_j (Y_{jt} - \bar{Y}_1)\right)^2\\
    \text{subject to} \;\;\;\; & \sum_{j=2}^{N} \gamma_j = 1\\
& \gamma_j \geq 0 \;\;\; j = 2, \ldots, N.
\end{aligned} 
\end{equation}
The corresponding estimate of the treatment effect at post-treatment time $t$ has the form of a weighted difference-in-difference estimator:
$$\hat{\tau}_t = (Y_{1t} - \bar{Y}_1) - \sum_{j=2}^{N} \hat{\gamma}_j (Y_{jt} - \bar{Y}_j).$$

\subsection{Comparison with \texttt{CausalImpact}}
\label{sec:appendix_causal_impact}

The \texttt{CausalImpact} approach from \citet{causalImpact} proposes a Bayesian Structural Time Series model for forecasting the missing counterfactual outcomes for a single unit. Figure \ref{fig:causal_impact} shows the estimates from applying \texttt{CausalImpact} to our application.

\paragraph{Single treated unit only.} The original paper largely focuses on the setting with a single treated unit, without control units, which corresponds to the ``interrupted time series'' or single-task GP setup in Section \ref{sec:single_gp_intro}. 
In this setting, the \texttt{CausalImpact} Bayesian Structural Time Series can be viewed as a special case of the more general Gaussian Process framework, with specific choices for the covariance \citep{solin2016stochastic}. These implementation choices can have meaningful advantages, especially in applications with high frequency time series and seasonality. However, such models can be difficult to generalize and are also restricted to a pre-defined set of regularly spaced time points.

\paragraph{Incorporating comparison units.} While not a focus of the original paper, the \texttt{CausalImpact} approach can also be adapted to include untreated comparison units. This is where the \texttt{CausalImpact} approach and our MTGP framework start to diverge. 

The first critical difference is that \texttt{CausalImpact} conditions on the \emph{observed} outcomes of the comparison units, in the same way that the approach would condition on observed (time-varying) covariates. By contrast, our framework models the relationship across units in terms of latent (rather than noisy) outcomes. This is especially important in applications like ours when we anticipate large differences in the variance of the noise across units, such as when looking at murder rates in Texas versus Wyoming. This also allows us to appropriately propagate uncertainty.

Another important difference is in model parameterization. In \texttt{CausalImpact}, they suggest incorporating a spike-and-slab prior to encourage sparsity among the weights of the control units.  While not pursued in this paper, a similar strategy can be employed by leveraging variants of the horseshoe prior \citep{carvalho2009horseshoe}.  In contrast to \texttt{CausalImpact}, our model is closer to a reduced rank regression model, which has a natural interpretation in terms of the number of (time-varying) latent confounders.

Finally, while not an important constraint in our application, we note that the \texttt{CausalImpact} approach cannot easily accommodate missing data in the comparison units, which \citet{causalImpact} emphasize in their package documentation. By contrast, this is straightforward in our MTGP framework.

\begin{figure}[htbp]
    \centering
    \begin{subfigure}{0.5\textwidth}
    \includegraphics[width=0.9\maxwidth]{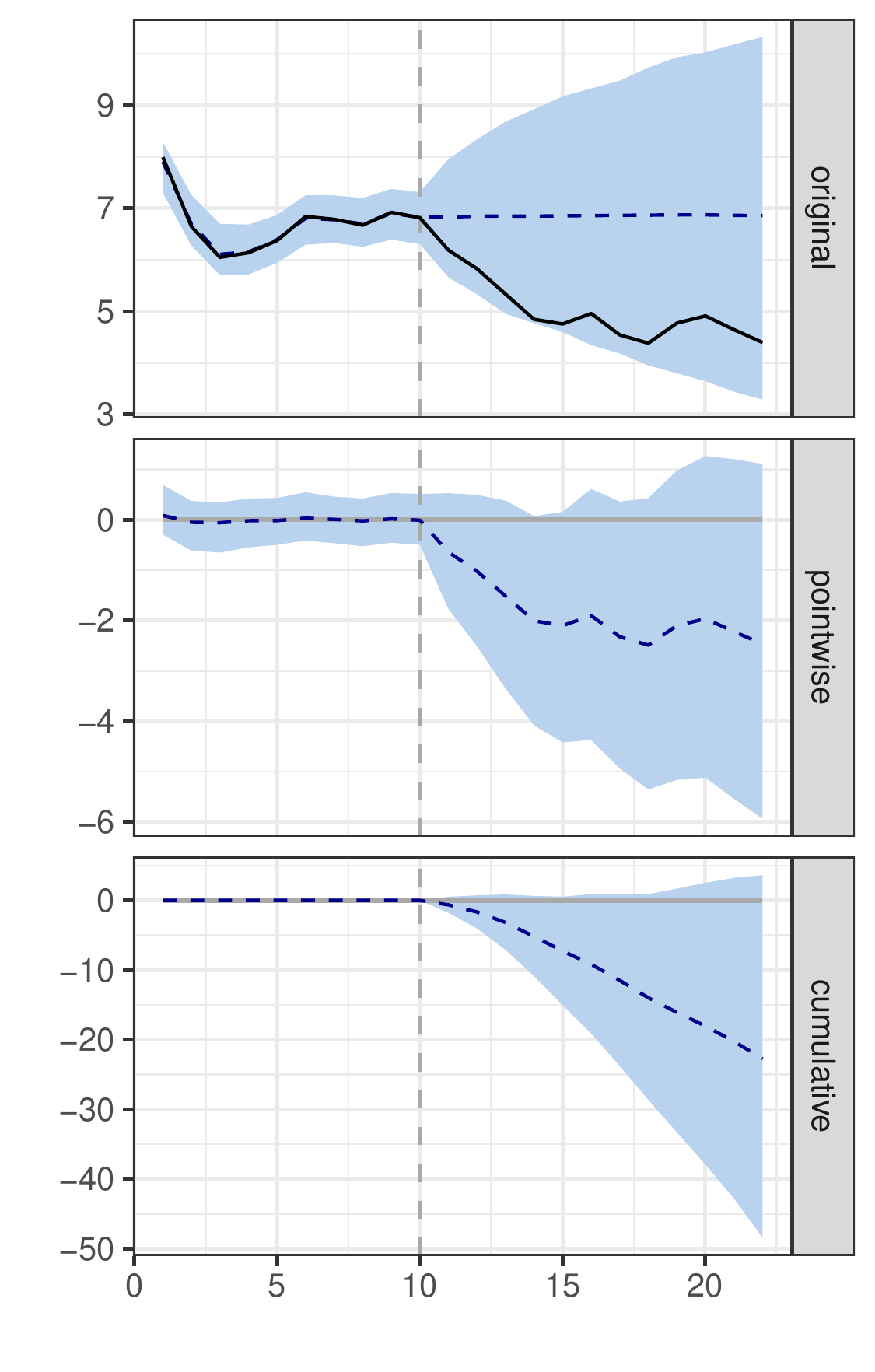}
    \caption{California Only}
    \label{fig:causal_impact_ITS}
    \end{subfigure}~%
    \begin{subfigure}{0.5\textwidth}
        \centering
    \includegraphics[width=0.9\maxwidth]{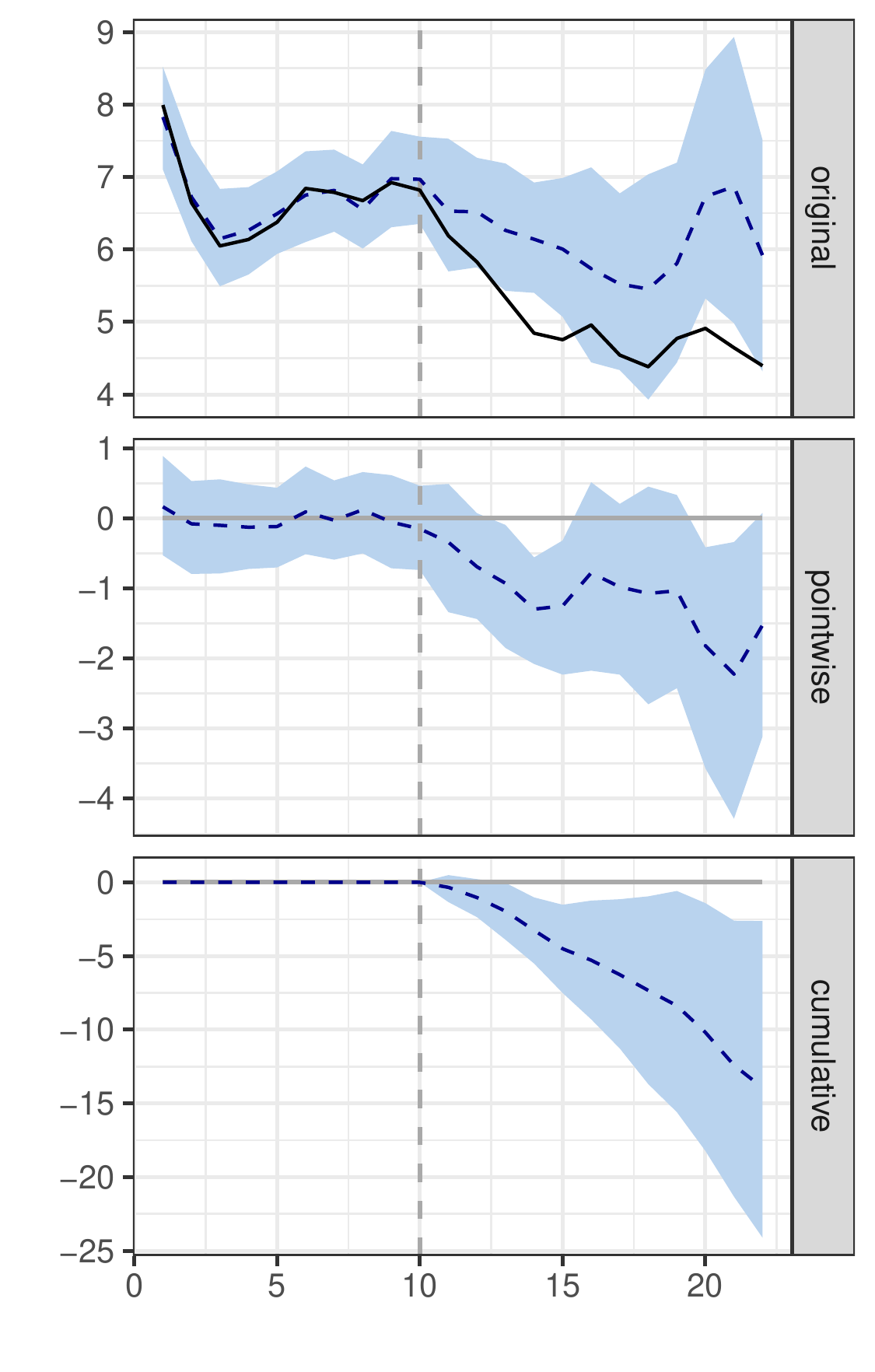}
    \caption{With Comparison States}
    \label{fig:causal_impact_CITS}
    \end{subfigure}
    \caption{Estimates using the \texttt{CausalImpact} package from \citet{causalImpact}. The left pane is based on California alone; the right pane includes comparison states in the model fit.}
    \label{fig:causal_impact}
\end{figure}

\clearpage
\section{Proofs}

\mtgpnllicm*
\begin{proof}
The proof largely follows from propositions in \citet{kanagawa2018gaussian}, presented here for the sake of completeness.

\paragraph{Part (a).}
The form of the weights follows directly from the conditional expectation of the multivariate normal in Equation \eqref{eq:gp_predmean}. Specifically, the posterior (predictive) mean of the MTGP for unit $i$ at time $t$ is:
\begin{align*}
\hat{\mu}_{it} &= 
(\mathbf{K}_{\obs} + \sigma^2\mathbb{I})^{-1} k((i, t), \cdot)\mathbf{Y}_1 \\
&= \underbrace{\left(\mathbf{K}_{\text{time}} \otimes \mathbf{K}_{\text{unit}} + \sigma^2 \mathbb{I}\right)^{-1}\left(k_{\text{time}}(t, \cdot) \otimes k_{\text{unit}}(i, \cdot)\right)}_{\boldsymbol{w}} \mathbf{Y}_1 \\
&= \sum_{(i',t') \in \mathcal{C}} w_{i't'}^{(i, i)} Y_{i't'},
\end{align*}
for general weights $w^{(i,t)} \in \mathbb{R}^{N\times T}$ for target observation $(i,t)$.

The separation of the overall weights $w$ into the product of unit and time weights $\gamma$ and $\lambda$ follows from fundamental properties of kernels.
Specifically, the form of the ICM kernel is given by the product of its constituent components, i.e. $k_{\text{time}}(t,t')k_{\text{unit}}(i,i')$. 
This resulting kernel is equivalent to 
$$
k_{\text{time}}(t,t')k_{\text{unit}}(i,i') 
= \mathbb{E}_{f\sim \mathcal{N}(0, \mathbf{K}_{\text{time}})}\left[f(t)^\top f(t')\right]\mathbb{E}_{g\sim \mathcal{N}(0, \mathbf{K}_{\text{unit}})}\left[g(i)^\top g(i')\right].
$$
Finally, considering the full covariance we have
$$
\mathbf{K}_{\text{time}}\otimes\mathbf{K}_{\text{unit}} = (f^\top f) \otimes (g^\top g) = (f \otimes g)^\top(f \otimes g)
$$
which we can see is the product of independent functions, with equivalent representations given in terms of functions of the respective basis representations, $f = \lambda^\top \phi(t)$, and $g = \gamma^\top \psi(i)$.

\paragraph{Part (b).} This is immediate from proposition 3.12 (equation 43) of \citet{kanagawa2018gaussian} after applying decomposing the weights into the time and unit components.

\paragraph{Part (c).}
We start with the following result from Kanagawa. \begin{proposition}\citep{kanagawa2018gaussian}
Assume that $Y_i = f_i + \varepsilon_i$ where $\varepsilon$ is independent zero-mean noise, and $f$ is a fixed function contained in the reproducing Hilbert space implied by the kernel $k$. 
For any previously unobserved point, denoted $t^*$, we have
$$
\sqrt{\sum_{i=1}^T{k}(t^*, \cdot)_i w^{(t^*)}_i+\sigma^{2}}=\sup _{f \in \mathcal{H}_{k}}\left(f_{t^*}-\sum_{i=1}^{T} w_{i}^{(t^*)} f_i\right) ,
$$
where $i$ indexes over the $T$ observed points, and $w = \left(\boldsymbol{K} + \sigma \mathbb{I}\right)^{-1}k(t^*, \cdot)$.
\end{proposition}

The error bounds follow by applying this result, noting that the supremum forms an upper bound, rearranging terms, and applying the separation of weights shown earlier. 
\end{proof}

\clearpage
\section{Additional results}
\label{sec:appendix_additional_results}

\subsection{Additional diagnostics}
\label{sec:appendix_additional_diagnostics}

\begin{figure}[htbp]
    \centering
    \includegraphics[width=0.9\maxwidth]{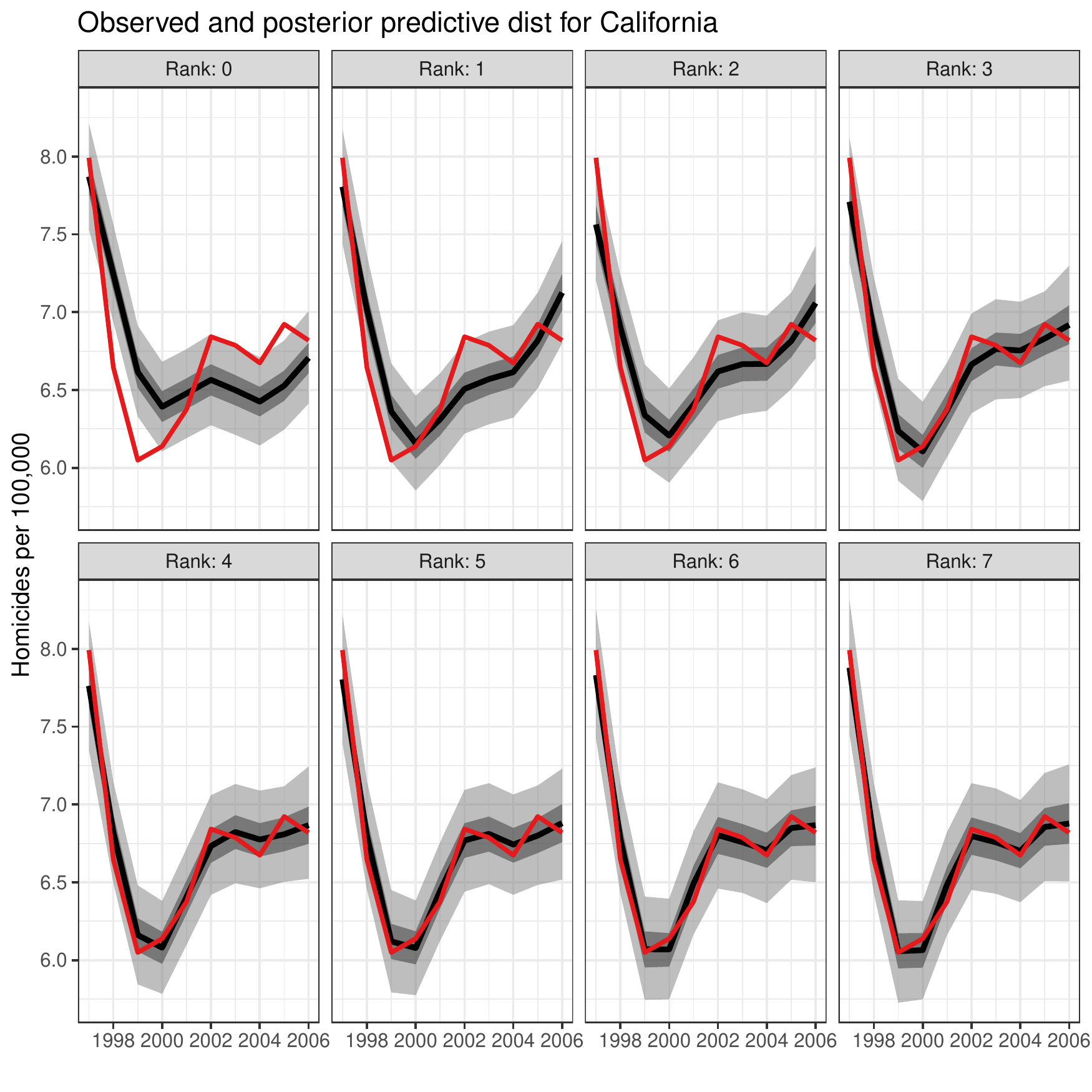}
    \caption{Pre-treatment imbalance for Poisson model. 50\% interval and 95\% interval plotted. Observed data is in red.}
    \label{fig:ca_balance_pois_raw}
\end{figure}

\begin{figure}[htbp]
    \centering
    \begin{subfigure}{0.5\textwidth}
    \includegraphics[width=0.9\maxwidth]{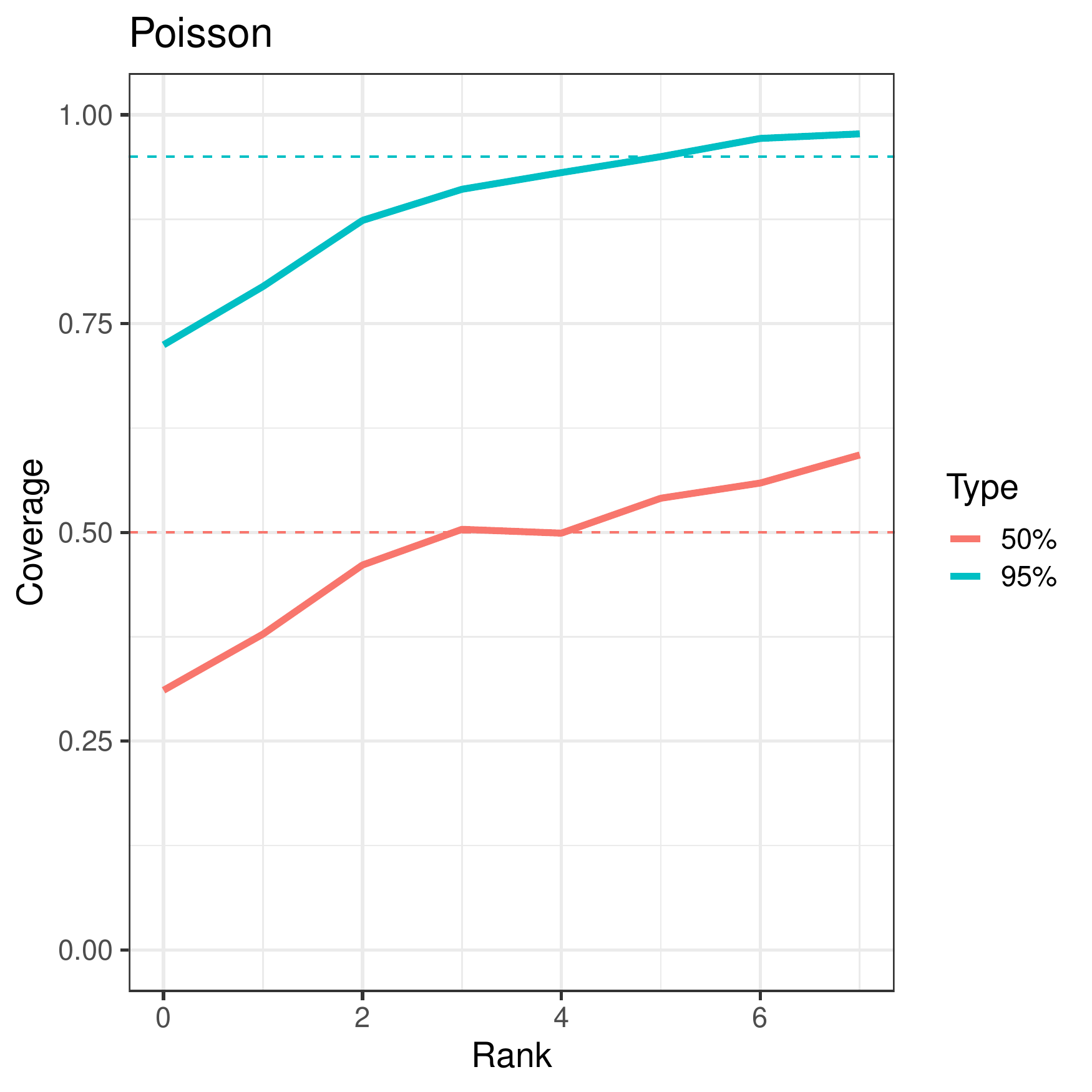}
    \caption{Poisson}
    \label{fig:coverage_pois}
    \end{subfigure}~%
    \begin{subfigure}{0.5\textwidth}
        \centering
    \includegraphics[width=0.9\maxwidth]{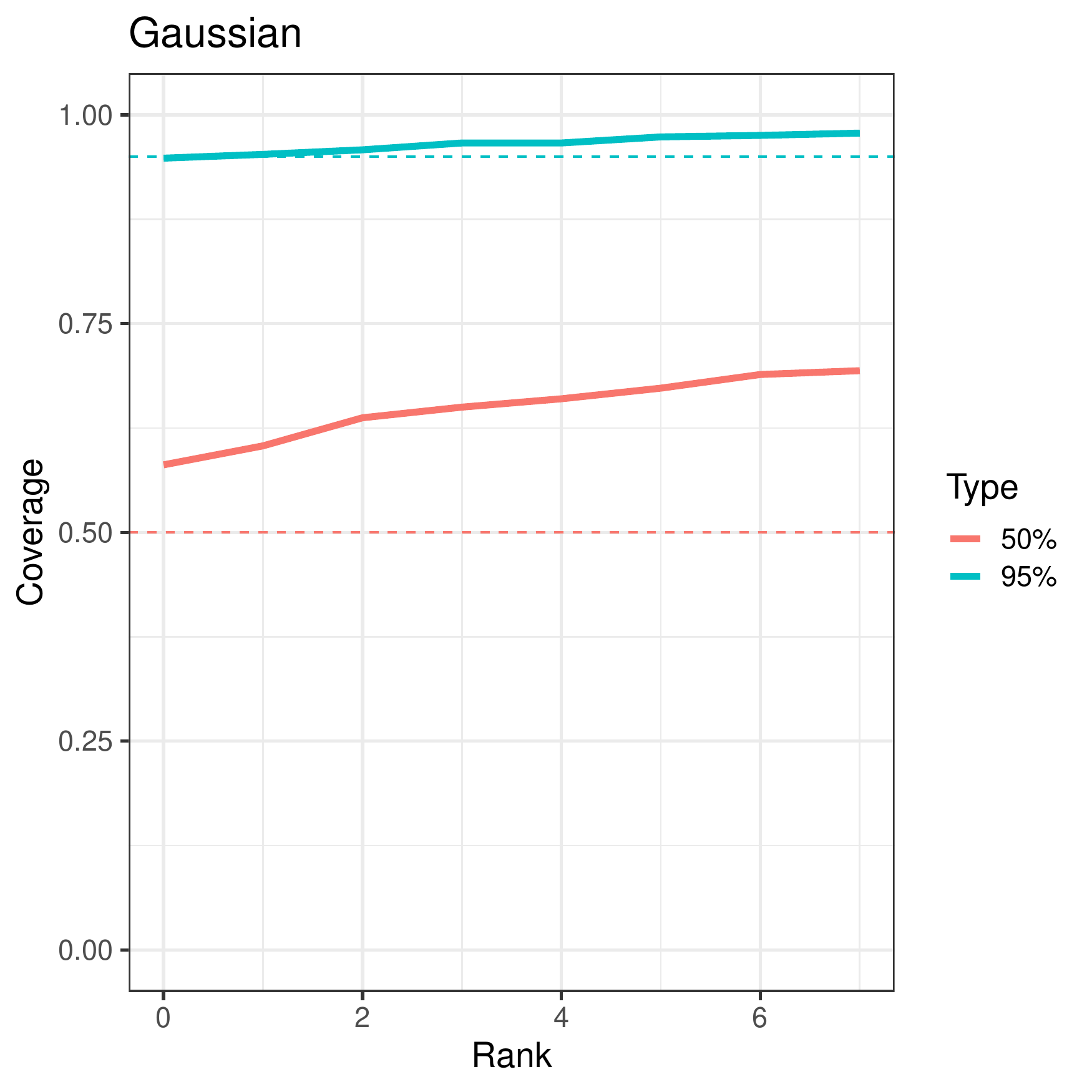}
    \caption{Gaussian}
    \label{fig:coverage_gaus}
    \end{subfigure}
    \caption{Coverage vs rank for Poisson and Gaussian models}
    \label{fig:coverage_ppc}
\end{figure}

\begin{figure}[htbp]
    \centering
    \includegraphics[width=0.6\maxwidth]{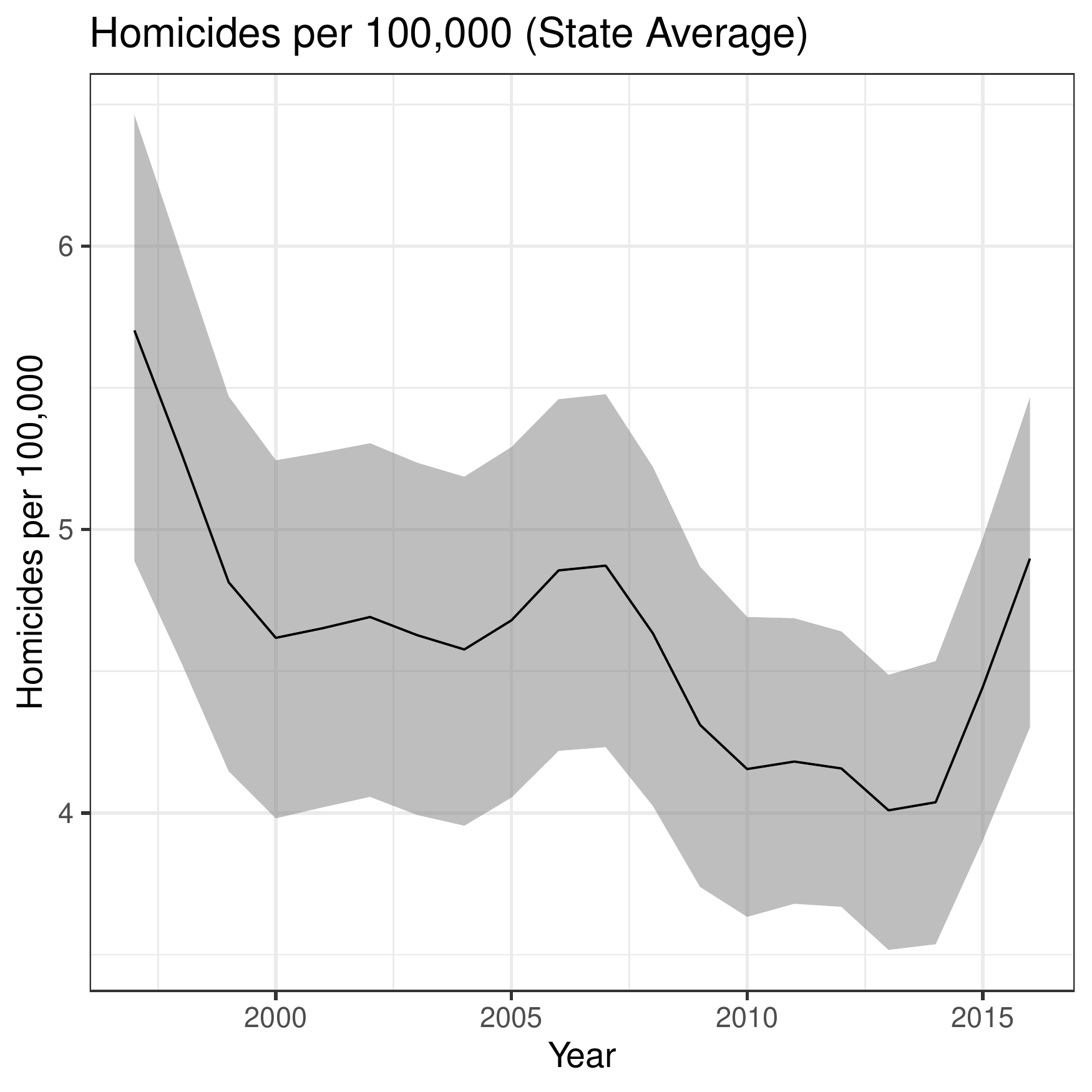}
    \caption{Global Gaussian Process Mean}
    \label{fig:global_factor}
\end{figure}

\begin{figure}[htbp]
    \centering
    \includegraphics[width=0.9\maxwidth]{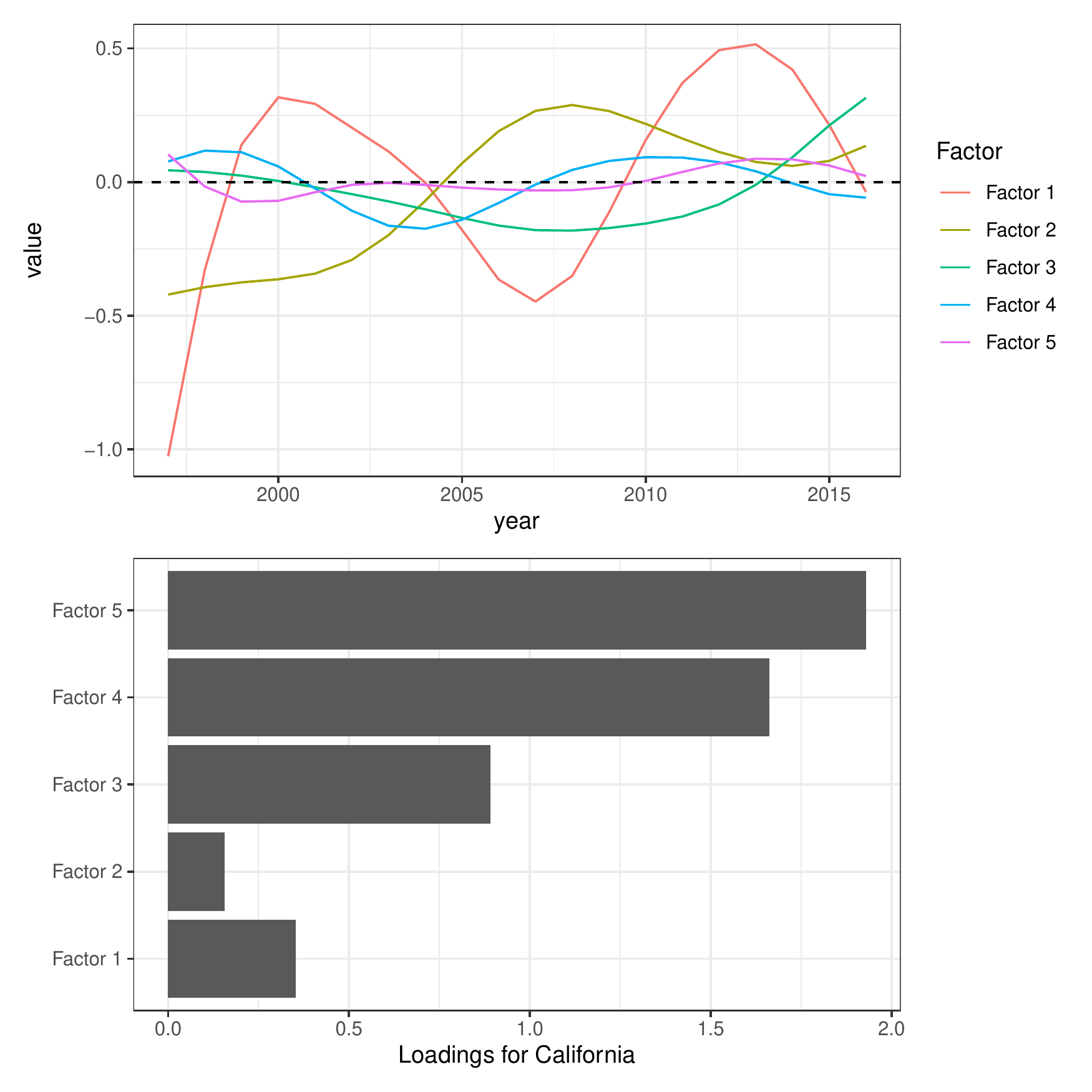}
    \caption{Inferred latent factors (top) and loadings for California. Note that the latent factors are only identifiable up to rotation.  Here, we order the factors according to the fraction of variance explained across outcomes from all states, and we choose the sign of the factors so that the loadings are all positive for California. All the estimated factors are smooth over time due to the kernel, $\mathbf{K}_\text{time}$. }
    \label{fig:all_factor}
\end{figure}

\begin{figure}[htbp]
    \centering
    \includegraphics[width=0.5\maxwidth]{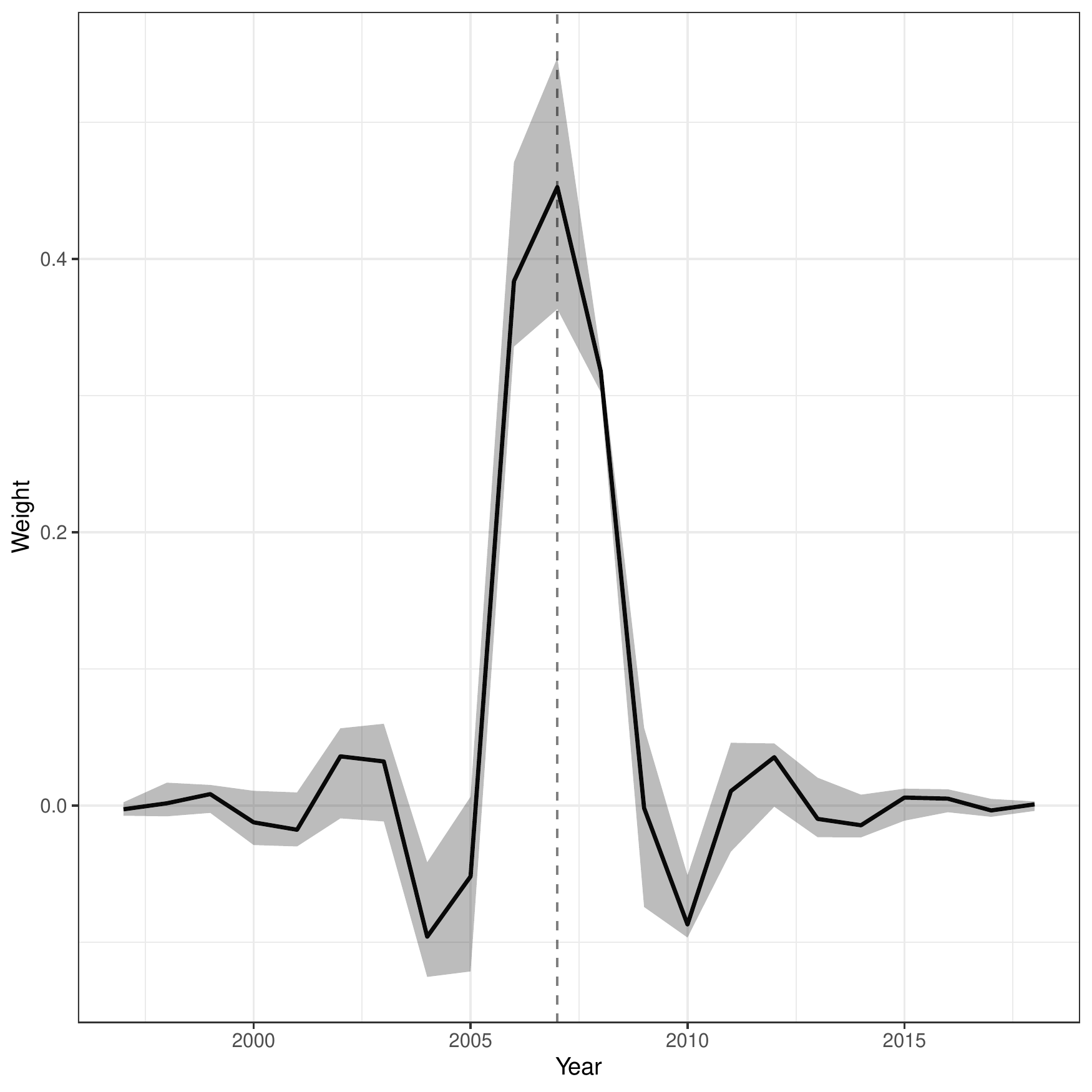}
    \caption{Inferred time weights for California.}
    \label{fig:time_weights}
\end{figure}
\begin{figure}[htbp]
    \centering
    \includegraphics[width=0.5\maxwidth]{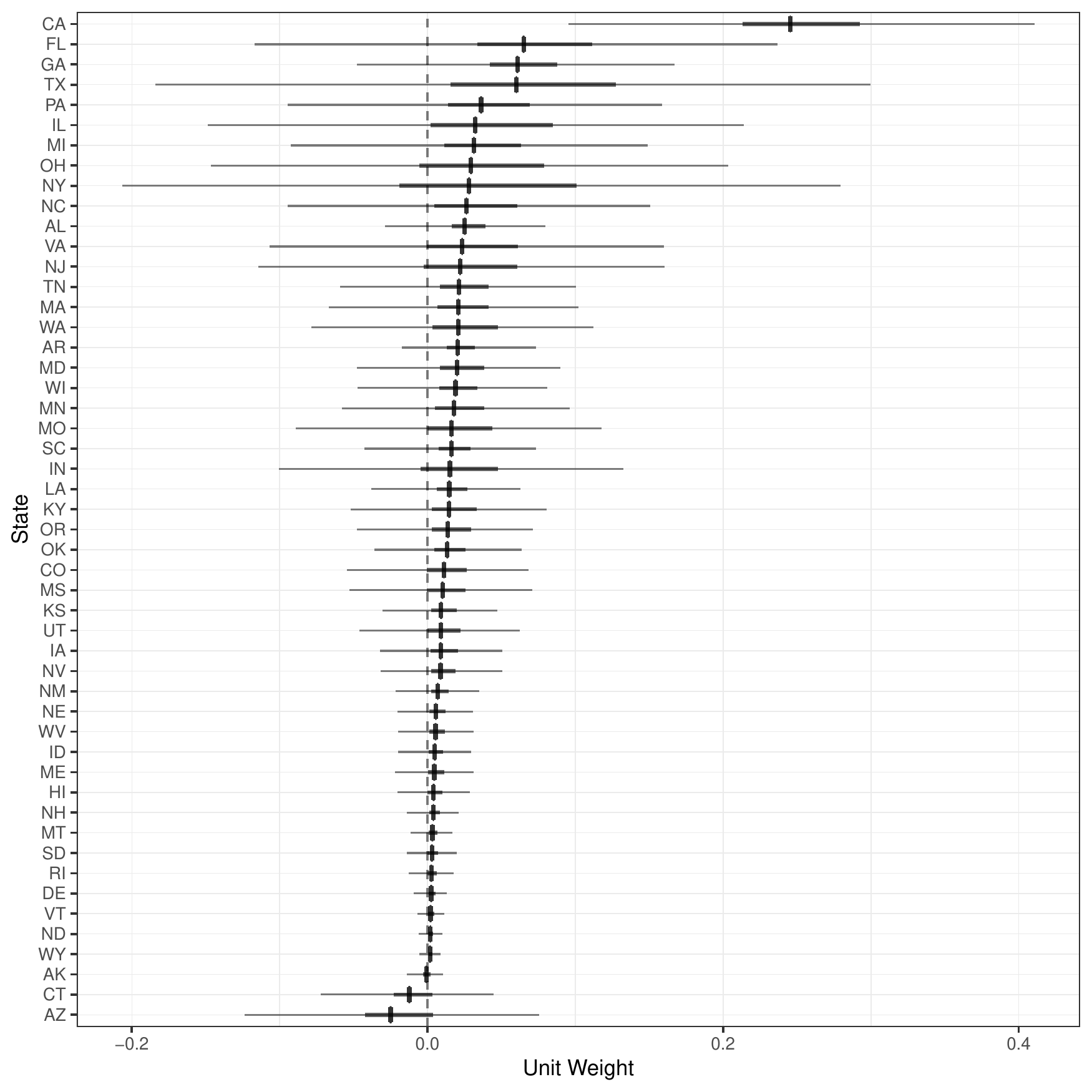}
    \caption{Inferred unit weights for California.}
    \label{fig:unit_weights}
\end{figure}

\clearpage
\subsection{Additional impact estimates}
\label{sec:appendix_additional_impacts}

\begin{figure}[htbp]
    \centering
    \includegraphics[width=0.9\maxwidth]{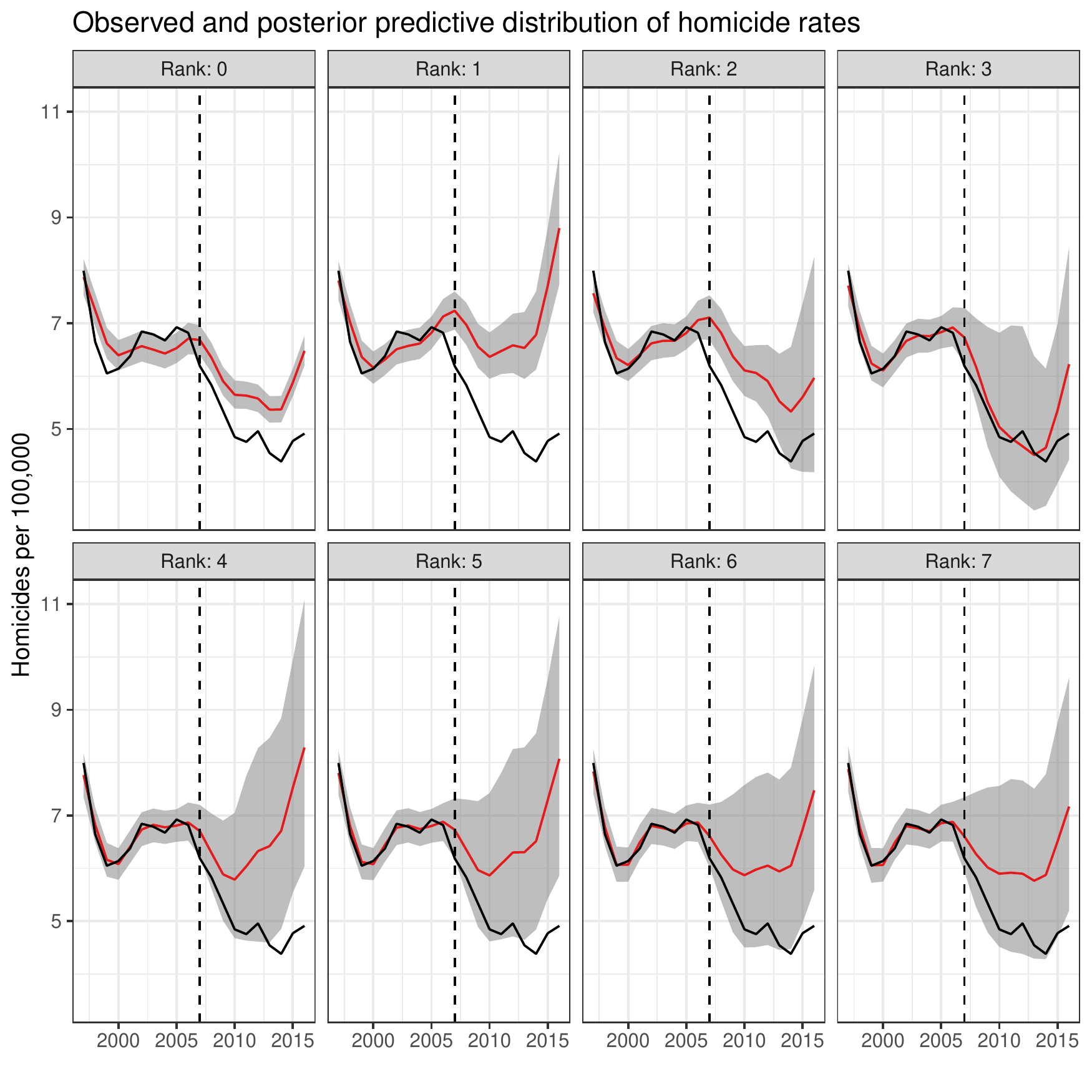}
    \caption{Estimates from MTGP for Normal model}
    \label{fig:mtgp_ests_appendix}
\end{figure}

\begin{figure}[htbp]
    \centering
    \includegraphics[width=0.9\maxwidth]{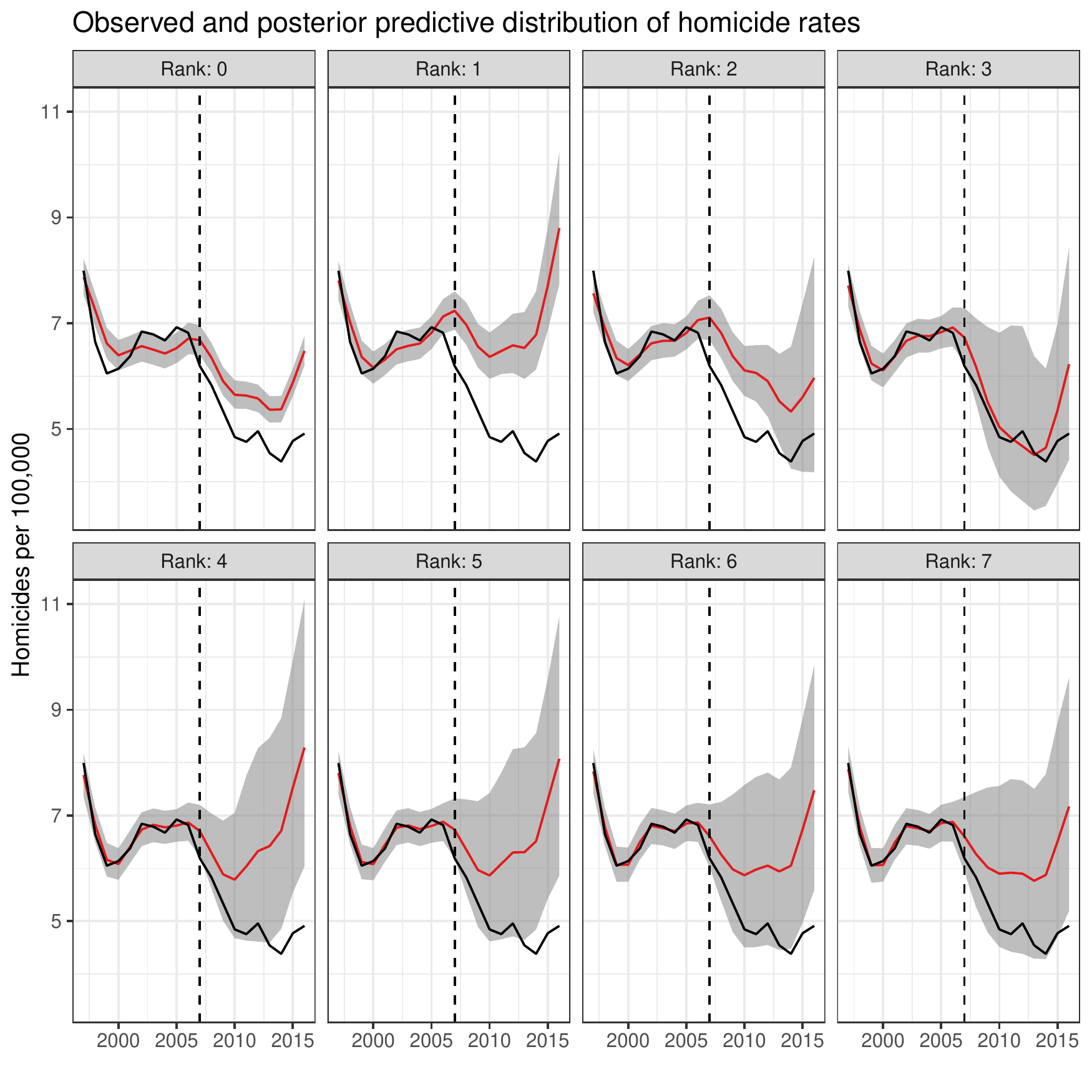}
    \caption{Estimates from MTGP for Normal model with a linear adjustment for background covariates.}
    \label{fig:mtgp_ests_appendix_covariates}
\end{figure}

\begin{figure}[htbp]
    \centering
    \includegraphics[width=0.9\maxwidth]{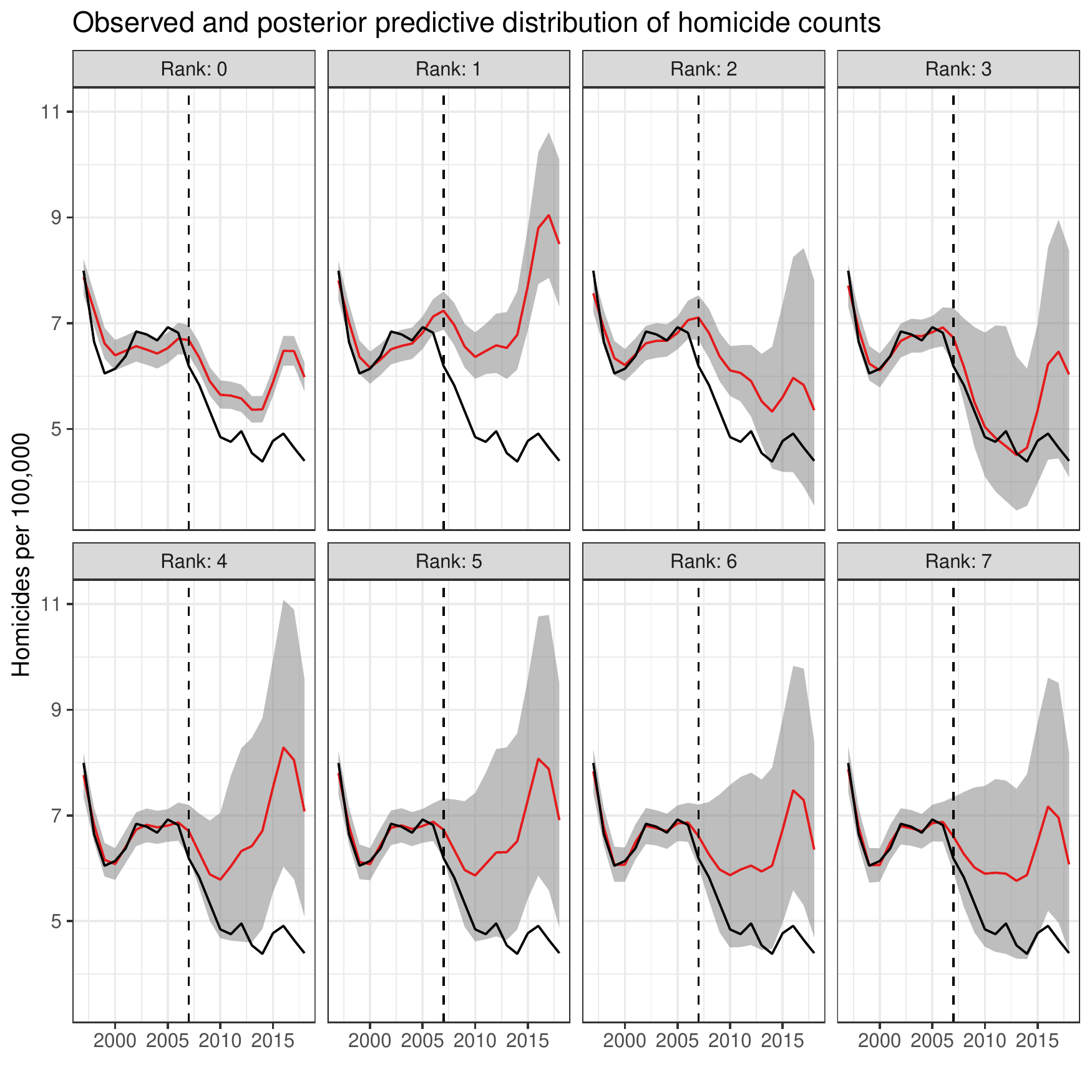}
    \caption{Estimates from MTGP with Poisson link}
    \label{fig:mtgp_ests_poisson_appendix}
\end{figure}

\begin{figure}[htbp]
    \centering
    \includegraphics[width=0.9\maxwidth]{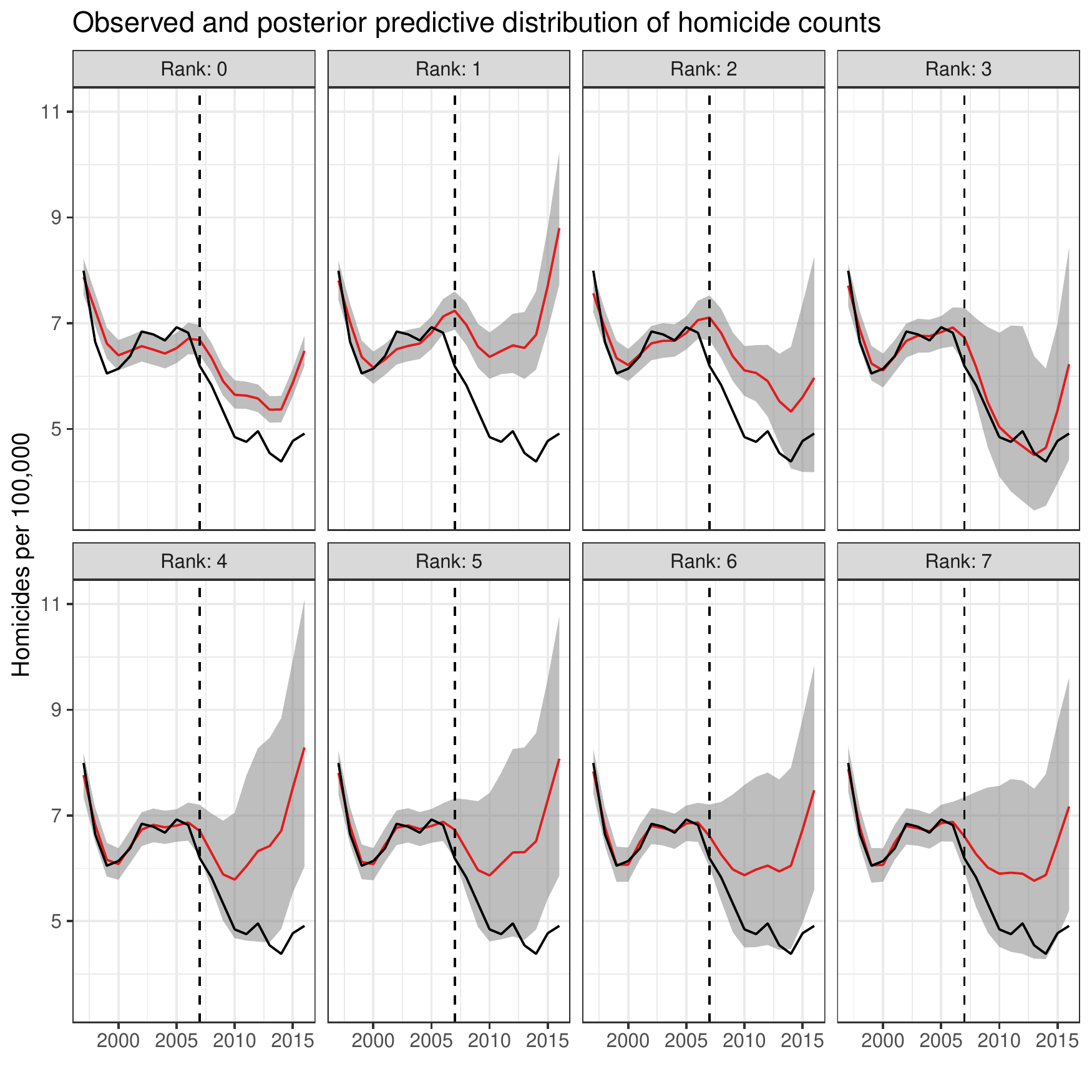}
    \caption{Estimates from MTGP with Poisson link and linear covariate adjustment.}
    \label{fig:mtgp_ests_poisson_adj}
\end{figure}

\begin{figure}[htbp]
    \centering
    \includegraphics[width=0.7\maxwidth]{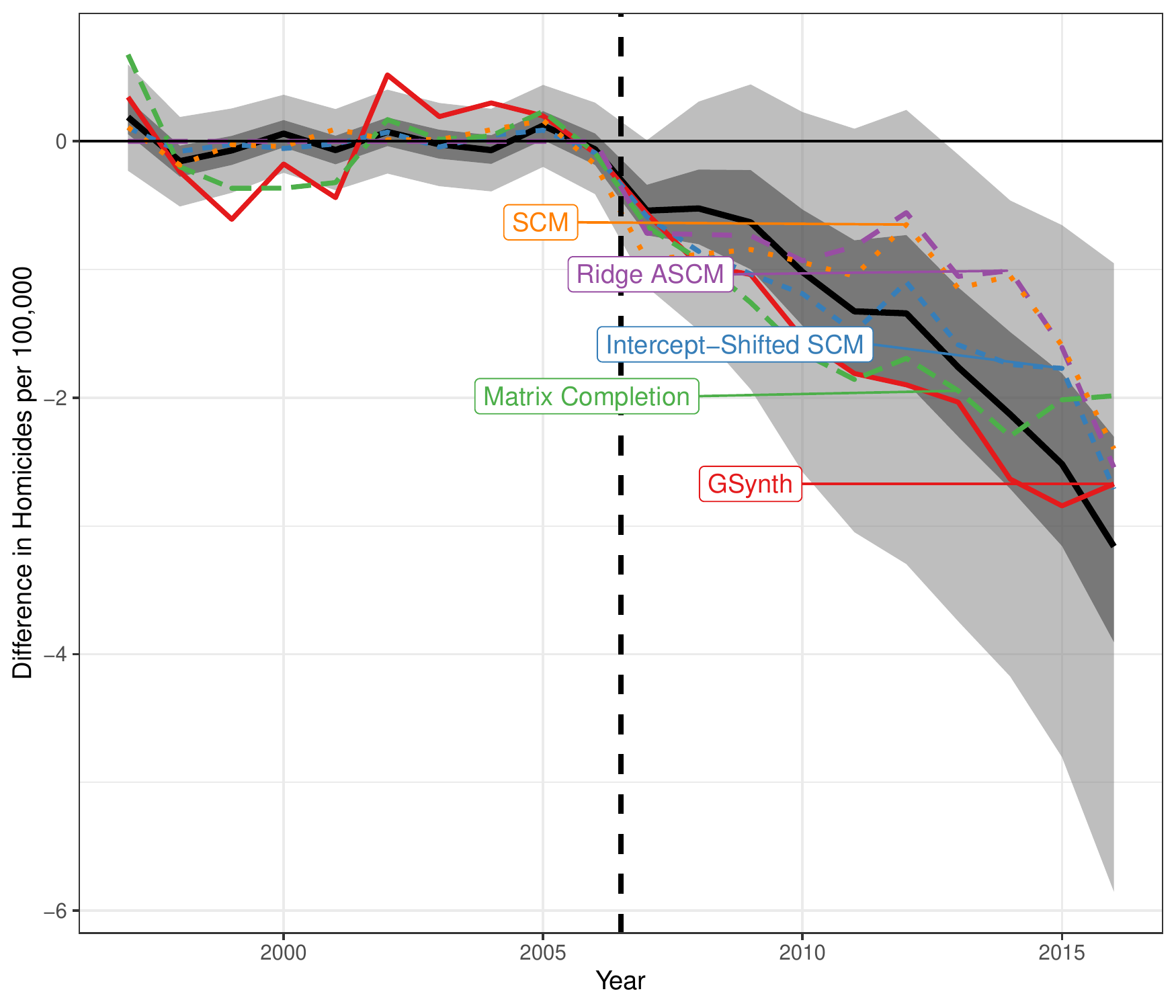}
    \caption{Estimates from SCM with and without an intercept shift, ridge ASCM, \texttt{gsynth}, and matrix completion, plotted over the posterior distribution of the effect of APPS on homicide rates with the rank 5 Poisson model.}
    \label{fig:alt_ests}
\end{figure}

\begin{figure}[htbp]
    \centering
    \includegraphics[width=0.7\maxwidth]{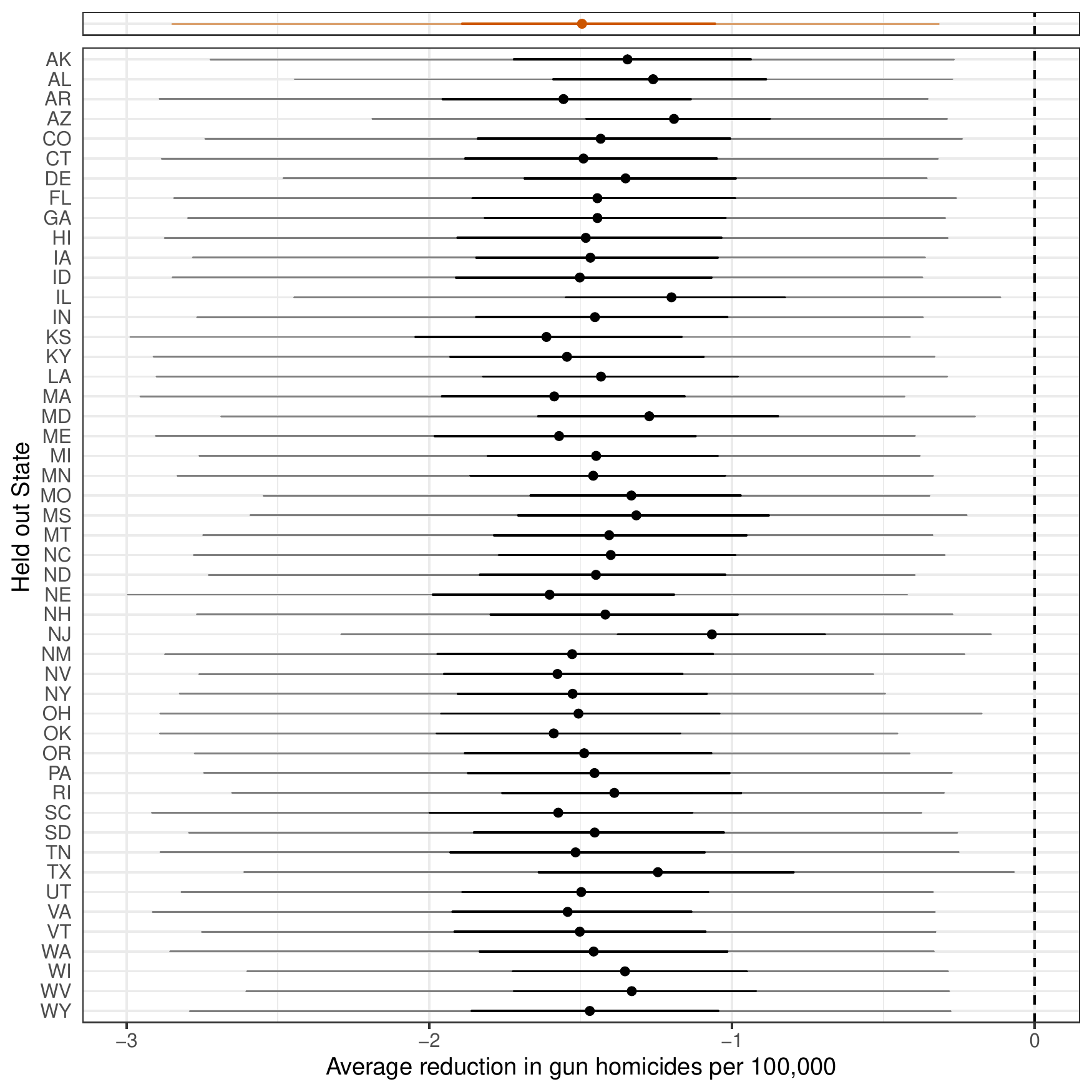}
    \caption{The posterior distribution for the rank 5 Poisson model when leaving out the homicide rate series for each of the 49 comparison states. The posterior with all states is shown in orange.}
    \label{fig:loo}
\end{figure}

\begin{figure}[htbp]
    \centering
    \includegraphics[width=0.45\maxwidth]{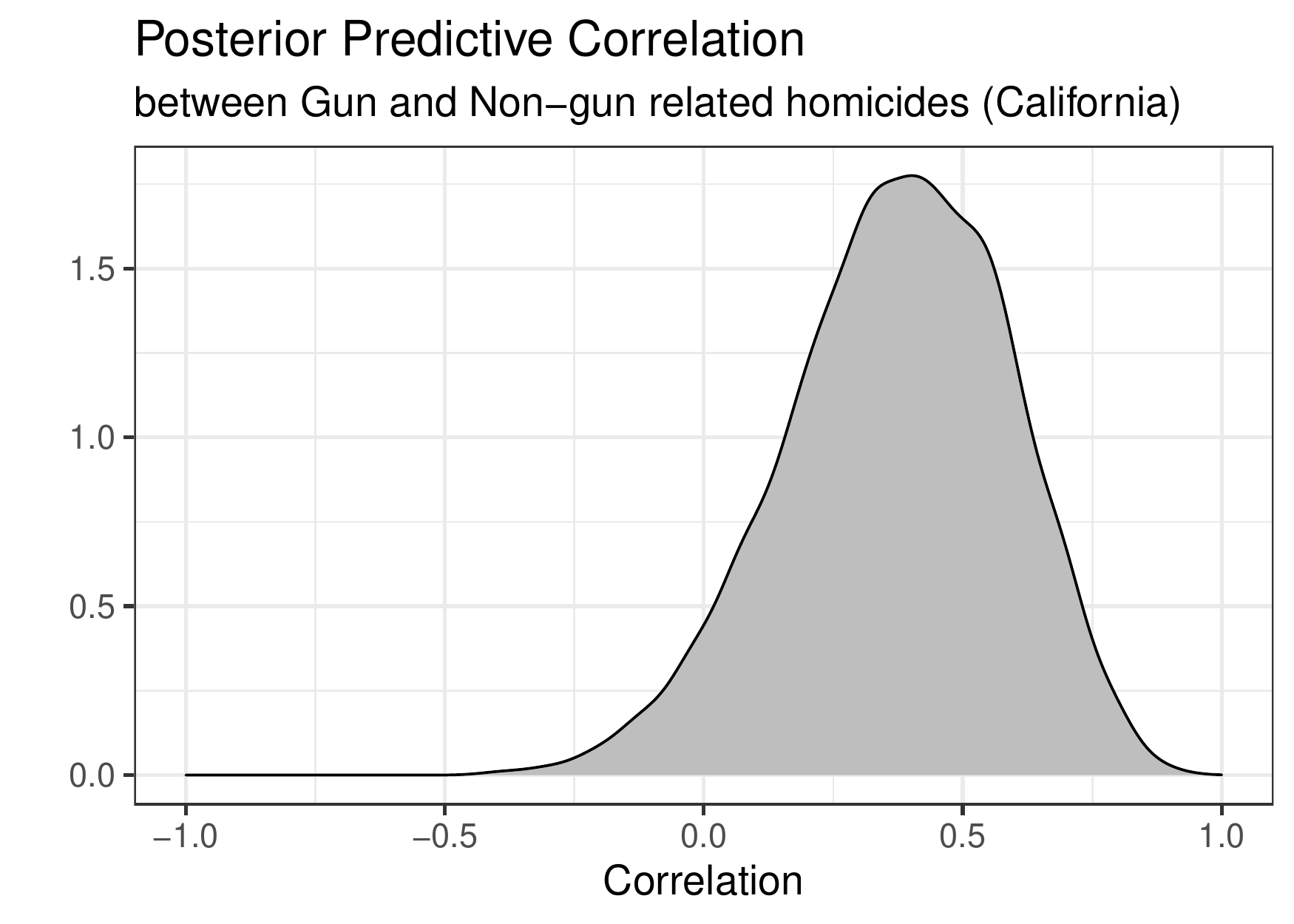}
    \caption{Posterior predictive correlation between gun and non-gun related homicides in California, pre-treatment.}
    \label{fig:mo_cor}
\end{figure}

\pagebreak
\subsection{Select MCMC Convergence Diagnostics}

\begin{figure}[htbp!]
    \centering
    \includegraphics[width=0.75\maxwidth]{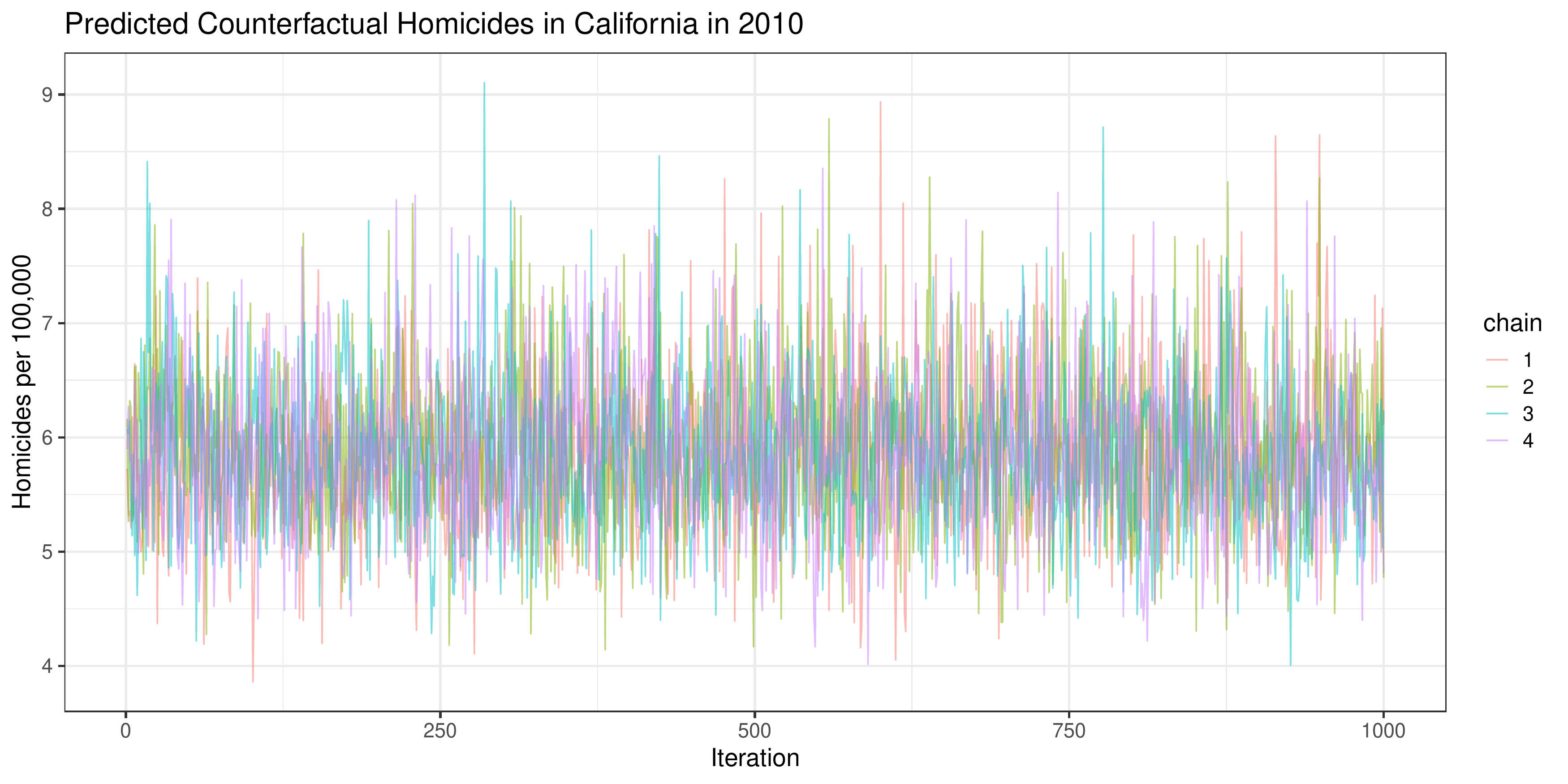}
    \caption{Predicted counterfactual homicides for California in 2010 for the Rank 5 Poisson model. For this model, $\hat{R} \approx 1$ and the bulk effective sample size is approximately 3500, indicating excellent mixing.}
    \label{fig:traceplot_ca2010_Poisson}
\end{figure}

\begin{figure}[htbp!]
    \centering
    \includegraphics[width=0.75\maxwidth]{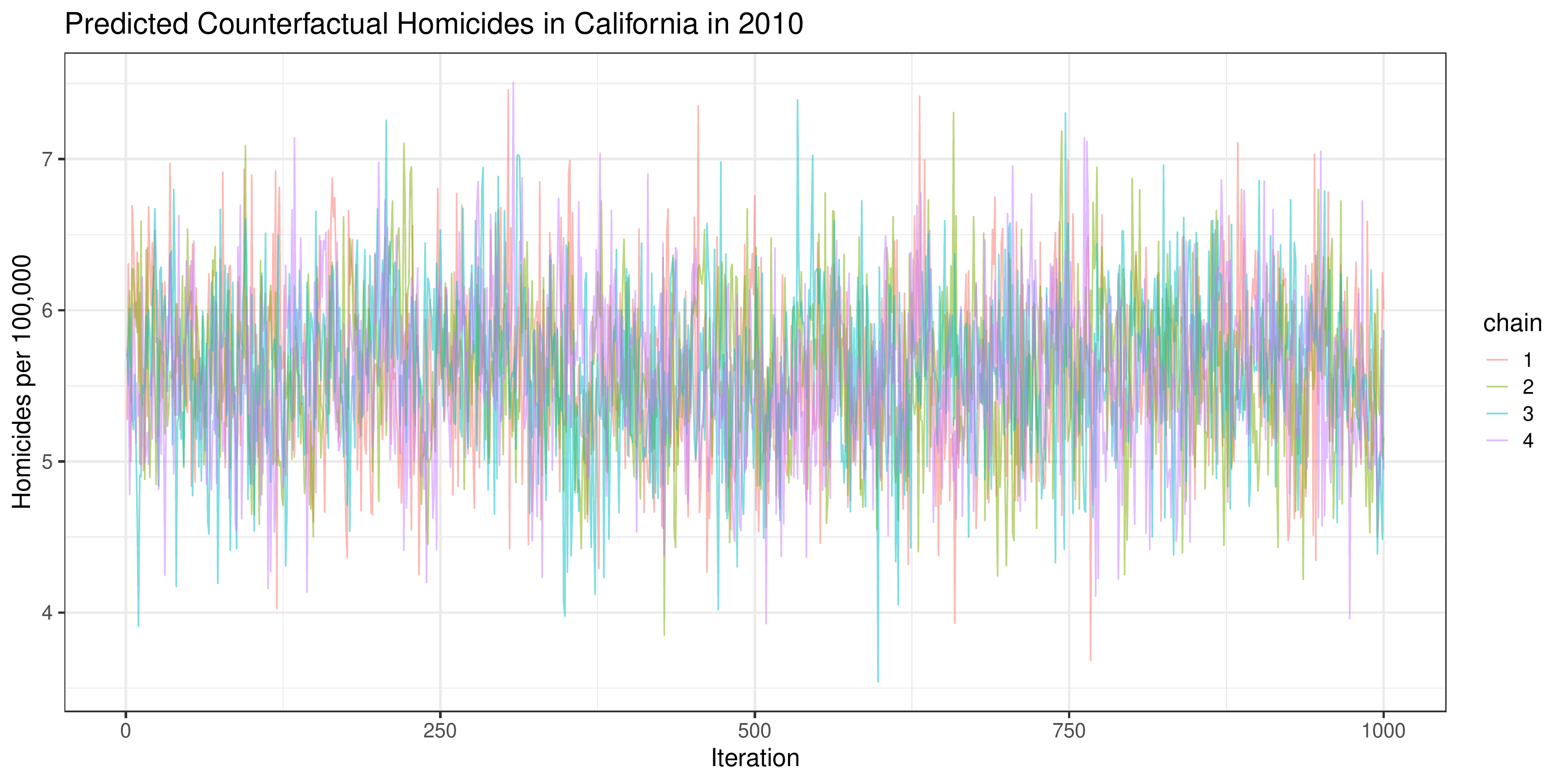}
    \caption{Predicted counterfactual homicides for California in 2010 for the Rank 5 Gaussian model. For this model, $\hat{R} \approx 1$ and the bulk effective sample size is approximately 1000, indicating excellent mixing.}
    \label{fig:traceplot_ca2010_Gaussian}
\end{figure}

\end{document}